%% file: main.tex
\documentclass[tbtags,authorcolumns,numberwithinsect,footnotebackref]{no-lipics-v2022}
\pdfoutput=1
\usepackage{amssymb}
\usepackage{mathtools}
\usepackage{bm}
\usepackage{stmaryrd}
\usepackage[draft]{fixme}
\usepackage{booktabs}

\usetikzlibrary{calc}

\input{macros}

\title{Near-Optimal-Time Quantum Algorithms for Approximate~Pattern~Matching}

\author{Tomasz Kociumaka}{Max Planck Institute for Informatics\\Saarland Informatics
Campus\\Saarbrücken, Germany}{tomasz.kociumaka@mpi-inf.mpg.de}{https://orcid.org/0000-0002-2477-1702}{}
\author{Jakob Nogler}{ETH Zürich\\Zurich, Switzerland}{jnogler@ethz.ch}{https://orcid.org/0009-0002-7028-2595}{}
\author{Philip Wellnitz}{National Institute of Informatics\\The Graduate University for Advanced Studies, SOKENDAI\\Tokyo, Japan}{wellnitz@nii.ac.jp}{https://orcid.org/0000-0002-6482-8478}{}

\authorrunning{T. Kociumaka, J. Nogler, and P. Wellnitz}

\nolinenumbers

\begin{document}
\pagenumbering{roman}
\maketitle
\begin{abstract}
\input{abstract}
\end{abstract}
\clearpage
\thispagestyle{plain}
\tableofcontents
\clearpage
\pagenumbering{arabic}

\input{s1_intro}
\input{s1b_tech_ov}
\input{s2_prelim}

\input{s3_gs_encoding}

\pagebreak
\input{s4_qpmwm}
\input{s5_qpmwe}

\bibliographystyle{alphaurl}
\bibliography{main}

\end{document}

%% file: macros.tex
\usetikzlibrary{arrows,arrows.meta,shapes.misc, decorations.pathreplacing,decorations.pathmorphing,calligraphy, patterns.meta}
\usetikzlibrary{calc}
\tikzset{dotmark/.style={circle,fill,inner sep=1.5pt}}
\tikzset{emptymark/.style={circle,draw,fill=white,inner sep=1.5pt}}
\tikzset{crossmark/.style={thick,inner sep=1.5pt}}

\usepackage{accents}

\newcommand{\altanonymous}[2]{\ifx\authoranonymous\relax#1\else#2\fi}

\SetKwInput{KwInput}{Input}
\SetKwInput{KwOutput}{Output}

\def\ShowAuthNotes{1}
\ifnum\ShowAuthNotes=1
\newcommand{\authnote}[3]{\textcolor{#3}{[{\bf #1:} { {#2}}]}}
\else
\newcommand{\authnote}[3]{}
\fi

\def\twoheadleadsto{\tikz[baseline=(a.base)]{\draw[%
    decorate,decoration={zigzag,segment length=4, amplitude=.9},%
    ] (0,0) -- (.25, 0);%
    \draw[%
    -{Classical TikZ Rightarrow}.{Classical TikZ Rightarrow},%
    ] (.25, 0) -- (.4, 0);%
    \node (a) at (.4/2,-.03) {\phantom{\(\leadsto\)}};%
}}

\newcommand{\onto}{\twoheadleadsto}

\newcommand{\uwidehat}[1]{%
  \mathpalette\douwidehat{#1}%
}

\makeatletter
\newcommand{\douwidehat}[2]{%
  \sbox0{$\m@th#1\widehat{\hphantom{#2}}\vphantom{t}$}%
  \sbox2{$t$}%
  \dimen2=\ht0
  \advance\dimen2 -\ht2
  \sbox2{$#2$}%
  \dimen0=\ht0
  \rlap{%
    \raisebox{\dimexpr-\dimen0-\dp2-1pt}[0pt][\dimexpr\dimen2+\dp2]{\box0}%
  }
  {#2}%
}
\makeatother

\newcommand{\Oh}{\mathcal{O}}
\newcommand{\Ohtilde}{\tilde{\Oh}}
\newcommand{\Ohhat}{\hat{\Oh}}
\DeclareMathOperator*{\poly}{poly}

\def\substr{\ensuremath \preccurlyeq}

\def\intvl#1#2{\bm{[}\,#1\,\bm{,}\,#2\,\bm{]}}
\def\fragmentco#1#2{\bm{[}\,#1\,\bm{.\,.}\,#2\,\bm{)}}
\def\fragmentoc#1#2{\bm{(}\,#1\,\bm{.\,.}\,#2\,\bm{]}}
\def\fragmentoo#1#2{\bm{(}\,#1\,\bm{.\,.}\,#2\,\bm{)}}
\def\fragment#1#2{\bm{[}\,#1\,\bm{.\,.}\,#2\,\bm{]}}

\def\position#1{\bm{[}\,#1\,\bm{]}}

\newcommand{\ceil}[1]{\lceil #1 \rceil}
\newcommand{\floor}[1]{\lfloor #1 \rfloor}
\newcommand{\rev}[1]{\overline{#1}}
\newcommand{\per}{\operatorname{per}}
\newcommand{\rot}{\operatorname{rot}}
\newcommand{\LZ}{\mathsf{LZ}}

\newcommand{\ed}{\delta_E}
\newcommand{\hd}{\delta_H}
\newcommand{\edl}[2]{{\delta_E}(#1,{}^*\!#2^*)}

\newcommand{\edal}[1]{\delta_E^{#1}}

\newcommand{\edshifted}[1]{\delta_E^{#1}} 
\newcommand{\OccE}{\mathrm{Occ}^E}
\newcommand{\OccH}{\mathrm{Occ}^H}
\newcommand{\Occ}{\mathrm{Occ}}
\newcommand{\Mis}{\mathrm{Mis}}
\newcommand{\MI}{\mathrm{MI}}

\newcommand{\gaped}{{\sc Gap Edit Distance}\xspace}
\newcommand{\shifted}{{\sc Shifted}\xspace}
\newcommand{\yes}{{\sc yes}\xspace}
\newcommand{\no}{{\sc no}\xspace}

\newcommand{\mA}{\mathcal{A}}

\newcommand{\mX}{\mathcal{X}}
\newcommand{\mY}{\mathcal{Y}}

\newcommand{\mP}{\mathcal{P}}

\newcommand{\mS}{\mathcal{S}}

\newcommand{\w}{\operatorname{w}}

\newcommand{\pref}{\mathsf{pref}}
\newcommand{\suf}{\mathsf{suf}}
\newcommand{\mXpref}{\mX_{\mathsf{pref}}}
\newcommand{\mXsuf}{\mX_{\mathsf{suf}}}

\newcommand{\Z}{\mathbb{Z}}
\newcommand{\Zz}{\mathbb{Z}_{\ge 0}}

\newcommand{\hi}{\hat{\imath}}
\newcommand{\hj}{\hat{\jmath}}
\newcommand{\hx}{\hat{x}}
\newcommand{\hy}{\hat{y}}

\newcommand{\sE}{\mathsf{E}}

\newcommand{\bG}{\mathbf{G}}

\newcommand{\bc}{\operatorname{bc}}

\def\modelname{{\tt PILLAR}\xspace}
\def\lceOp#1#2{{\tt LCP}(#1, #2)}
\def\lcbOp#1#2{{\tt LCP}^R(#1, #2)}
\def\ipmOp#1#2{{\tt IPM}(#1, #2)}

\def\accOp#1#2{#1\position{#2}}

\def\accOpName{{\tt Access}\xspace}

\def\lenOpName{{\tt Length}\xspace}

\DeclarePairedDelimiter{\ket}{\lvert}{\rangle}

\def\problembox#1{%
    \vspace{2mm}%
    \noindent\fbox{%
    \begin{minipage}{.985\linewidth}%
        #1
    \end{minipage}%
    }%
    \vspace{2mm}%
}

\newcommand{\defproblemc}[4]{%
    \problembox{%
        \textsc{#1}\\
        {\bf{Input:}} #2  \\
        {\bf{Output:}} #3
    }%
}

\makeatletter
\renewenvironment{cases}{%
  \matrix@check\cases\env@cases
}{%
  \endarray\right.%
}
\def\env@cases{%
  \let\@ifnextchar\new@ifnextchar
  \left\lbrace
  \def\arraystretch{1.1}%
  \array{@{\;}c@{\quad}l@{}}%
}
\makeatother

\def\mid{\ensuremath :}

\def\emptyset{\varnothing}

\newcommand{\G}{\mathcal{G}}
\newcommand{\N}{\mathcal{N}}
\newcommand{\rhs}{\mathsf{rhs}}
\newcommand{\val}{\mathsf{exp}}

\newcommand{\one}{\mathtt{1}}
\newcommand{\zero}{\mathtt{0}}

\def\eso#1{\ensuremath{\raisebox{.25ex}{\ensuremath{\uwidehat{#1}}}}}
\def\exw#1#2#3{\ensuremath E(#1, #2, #3)}
\def\exwf#1#2#3#4{\ensuremath E(#1, #2, #3, #4)}

\newcommand{\hG}{\hat{\mathcal{G}}}
\newcommand{\D}{\mathcal{D}}
\newcommand{\LCE}{\mathsf{LCE}}
\newcommand{\IPM}{\mathsf{IPM}}
\newcommand{\hS}{\hat{S}}
\newcommand{\hO}{\hat{O}}

\newcommand{\hp}{\hat{p}}
\newcommand{\hq}{\hat{q}}

\SetKwProg{Fn}{}{}{}
\def\pn#1{\textsc{#1}}
\def\PMwE{\pn{Pattern Matching with Edits}\xspace}
\def\PMwM{\pn{Pattern Matching with Mismatches}\xspace}

\def\pr#1{\ensuremath{\mathsf{Pr}\!\bm{\left[}\,#1\,\bm{\right]}}}

\def\ex#1{\ensuremath{\mathbb{E}\!\bm{\left[}\,#1\,\bm{\right]}}}

\def\GS{Grover's Search\xspace}

%% file: abstract.tex
Approximate Pattern Matching is among the most fundamental string-processing tasks.
Given a text \(T\) of length \(n\), a pattern \(P\) of length \(m\), and a threshold \(k\), the task is to identify the fragments of \(T\) that are at \emph{distance} at most \(k\) to \(P\).
We consider the two most common distances: Hamming distance (the number of mismatches or character substitutions) in \emph{Pattern Matching with Mismatches} and edit distance (the minimum number of character insertions, deletions, and substitutions) in \emph{Pattern Matching with Edits}.
We revisit the complexity of these two \mbox{problems in the quantum setting}.

\altanonymous{Very recently, Kociumaka, Nogler, and Wellnitz [STOC'24] showed that}{Our recent work [STOC'24] shows that} $\Ohhat(\sqrt{n/m} \cdot \sqrt{mk}) = \Ohhat(\sqrt{nk})$%
\footnote{Throughout this paper, the $\Ohtilde(\cdot)$ and $\Ohhat(\cdot)$ notations hide poly-logarithmic factors $(\log N)^{\Oh(1)}$ and sub-polynomial factors $N^{o(1)}$, respectively, with respect to the total input size $N$ of the considered problems.}
quantum queries are sufficient to solve (the decision version of) the Pattern Matching with Edits problem.
However, the quantum time complexity of the underlying solution (that is, the computational overhead to recover the solution from the quantum queries) does not provide any improvement over classical computation.
On the other hand, the state-of-the-art quantum algorithm for Pattern Matching with Mismatches [Jin and Nogler; SODA'23] achieves query complexity $\Ohhat(\sqrt{nk^{3/2}})$ and time complexity $\Ohtilde(\sqrt{nk^2})$, falling short of a known unconditional lower bound of $\Omega(\sqrt{nk})$ quantum queries.

In this work, we present quantum algorithms with a time complexity of
$\Ohtilde(\sqrt{nk}+\sqrt{n/m}\cdot k^2)$ for Pattern Matching with Mismatches and
$\Ohhat(\sqrt{nk}+\sqrt{n/m}\cdot k^{3.5})$ for Pattern Matching with Edits; both
algorithms use $\Ohhat(\sqrt{nk})$ quantum queries.
These running times are near-optimal for $k \ll m^{1/3}$ and $k\ll m^{1/6}$, respectively,
and they offer advantage over classical algorithms for $k \ll (mn)^{1/4}$ and $k\ll
(mn)^{1/7}$, respectively.
Our solutions can also report the starting positions of all approximate occurrences of $P$
in $T$ (represented as collections of arithmetic progressions); in this case, both the
unconditional lower bound and the complexities of our algorithms increase by a $\Theta(\sqrt{n/m})$ factor.

As a major technical contribution, we give a faster algorithm to solve a system of $b$ substring equations of the form $S\fragmentco{x}{x'}=S\fragmentco{y}{y'}$. 
The goal is to construct a generic solution string whose characters are equal only when necessary.
While this is known to be possible in $\Oh(n+b)$ time [Gawrychowski, Kociumaka, Radoszewski, Rytter, Waleń; TCS'16], we show that $\Ohtilde(b^2)$ classical time is sufficient to obtain an $\Ohtilde(b)$-size representation of $S$ that supports random access and other standard queries (of the so-called \modelname model) in $\Oh(\log n)$ time.
We apply this tool to efficiently construct an $\Ohtilde(k)$-size representation of strings $P^\#$ and $T^\#$ obtained from $P$ and $T$ by carefully masking out some characters so that the output to the studied approximate pattern matching problems does not change.  

%% file: s1_intro.tex
\section{Introduction}
Pattern matching is one of the foundational string-processing tasks.
Given two strings, a pattern $P$ of length $m$ and a text $T$ of length $n$, the classical
Exact Pattern Matching problem asks to identify the fragments of $T$ that match $P$
exactly.
Equivalently, the goal is to construct the set
\[\Occ(P,T) \coloneqq \{i\in \fragment{0}{n-m} : P = T\fragmentco{i}{i+m}\}.\]
Exact pattern matching admits efficient solutions in numerous models of computation,
including textbook linear-time algorithms for the standard setting~\cite{KMP77,KR87} and
more recent near-optimal algorithms for the dynamic~\cite{ABR00}, quantum~\cite{HV03},
streaming~\cite{PP09}, and compressed~\cite{Jez15} settings.

Unfortunately, imperfections in real-world data significantly limit the applicability of
Exact Pattern Matching: instead of looking for the exact occurrence of $P$ in $T$, it is
often necessary to also identify approximate occurrences, that is, fragments of $T$ that
are similar to $P$ but do not necessarily match $P$ exactly.
Given the flexibility in deciding when two strings are deemed similar, Approximate Pattern
Matching serves as an umbrella for many problems; see~\cite{HD80,Nav01} for surveys.

A natural choice for quantifying similarity is to impose a limit $k$ on the number of
character substitutions (that is, the Hamming distance $\hd$) or the (minimum) number of character
insertions, deletions, and substitutions (that is, the edit distance $\ed$) between the
pattern and its approximate occurrences.

\defproblemc{Pattern Matching with Mismatches}
{a pattern $P$ of length $m$, a text $T$ of length $n$, and an integer threshold $k>0$.}
{the set \[
    \OccH_{k}(P,T) = \{i \in \fragment{0}{n - m} \mid \hd(P, T\fragmentco{i}{i+m}) \leq k\}.
\]}
\par

\defproblemc{Pattern Matching with Edits}
{a pattern $P$ of length $m$, a text $T$ of length $n$, and an integer threshold $k>0$.}
{the set \[
    \OccE_k (P,T)\coloneqq \{i \in \fragment{0}{n} \mid \exists_{j \in \fragment{i}{n}}\; \ed(P, T\fragmentco{i}{j}) \leq k\}.
\]}
\par

Pattern Matching with
Mismatches~\cite{GG86,Abr87,Kos87,LV88,AmirLP04,CFPSS16,GU18,CGKKP20} and
Edits~\cite{S80,LandauV89,SV96,ColeH98,CKW22} are the two most studied Approximate Pattern
Matching variants in the theoretical literature.
Nevertheless, we still lack a complete understanding of the complexity of these two
problems.
For \PMwM, Gawrychowski and Uznański~\cite{GU18} developed an
$\Ohtilde(n+n/m\cdot k\sqrt{m})$-time algorithm and showed that no significantly faster
\emph{combinatorial} algorithm exists.
Nevertheless, it remains completely unclear how to improve the state-of-the-art running
time using fast matrix multiplication.
For two decades, the fastest two algorithms for \PMwE ran in
$\Oh(nk)$~\cite{LandauV89} and $\Oh(n+n/m\cdot k^{4})$~\cite{ColeH98} time.
Only two years ago, Charalampopoulos, Kociumaka, and Wellnitz~\cite{CKW22} managed to
circumvent the latter two-decades-old barrier and achieve an \(\Ohtilde(n + n/m\cdot
k^{3.5})\)-time algorithm.
Nevertheless, this is still very far from the best (conditional) lower bound~\cite{CKW22},
which only rules out algorithms faster than \(\Oh(n + k^{2 - \Omega(1)}n/m)\).

In parallel to improvements for the standard model of computation, Pattern Matching with
Mismatches and Edits were studied in other settings, including
streaming~\cite{PP09,CFPSS16,CKP19,KPS21,BK23}, where the text can only be read once from
left to right, fully compressed~\cite{bkw19,CKW20,CKW22}, where the strings are given in a
compressed form, dynamic~\cite{CKW20,CGKMU22,CKW22}, where the strings change over time,
and, last but not least, quantum~\cite{JN23,KNW24}.
On at least two occasions~\cite{CFPSS16,CKW22}, the insights learned from these models
were crucial to progress for the standard setting.

\paragraph*{Algorithmic Contributions in the Quantum Setting.}

In this paper, we study pattern matching under Hamming and edit distance within the
quantum computing framework.
While string algorithms have been studied in several works before such as~\cite{HV03,AGS19,ABIKKPSS20,BEGHS21,GS22,AJ22,GJKT24,WY20},
the quantum advantage for approximate pattern matching has only been demonstrated recently
Jin and Nogler~\cite{JN23}
and Kociumaka, Nogler, and Wellnitz~\cite{KNW24}.

The first of these works devises an algorithm for verifying whether \(\OccH_{k}(P,T) \neq
\emptyset\)
(and, if so, reporting a position), which requires \(\Ohhat(\sqrt{n/m} \cdot
k^{3/4}\sqrt{m})\) queries and \(\Ohhat(\sqrt{n/m} \cdot k\sqrt{m})\) time,
leaving open the question of whether the query lower bound \(\Omega(\sqrt{n/m} \cdot
\sqrt{km})\) can be achieved.

The second work shows that it is indeed possible to achieve this bound, but, quite surprisingly,
for the more challenging problem involving edits, without addressing the problem for mismatches.
They provide an algorithm for verifying whether \(\OccE_{k}(P,T) \neq \emptyset\),
requiring \(\Ohhat(\sqrt{n/m} \cdot \sqrt{km})\) queries and \(\Ohhat(\sqrt{n/m} \cdot (\sqrt{k}m + k^{3.5}))\) time.
Additionally, they give an algorithm capable of computing \(\OccE_{k}(P,T)\) using \(\Ohhat(n/m \cdot \sqrt{km})\) queries and \(\Ohhat(n/m \cdot (\sqrt{k}m + k^{3.5}))\) time.
While they achieve near-optimal query complexity for reporting and computing,
they fall short of obtaining sublinear time, which is one of the main attractions of quantum algorithms.
This shortfall arises because reducing the query complexity requires adaptations of communication complexity results from their work,
which do not lend themselves to straightforward efficient implementation.

In this work, we show that for both mismatches and edits, and for both computing and reporting,
it is indeed possible to achieve the two desired goals simultaneously,
without requiring any tradeoff: optimal query complexity (up to subpolynomial factors) and sublinear quantum time.

\begin{restatable*}{mtheorem}{qpmwm}\label{thm:qpmwm}
    Let $P$ denote a pattern of length $m$, let $T$ denote a text of length $n$, and let $k > 0$ denote an integer threshold.
    \begin{enumerate}
        \item There is a quantum algorithm that solves the \PMwM
        problem using $\Ohtilde(n/m \cdot \sqrt{km})$ queries and
        $\Ohtilde(n/m \cdot (\sqrt{km} + k^{2}))$ time.
        \item There is a quantum algorithm deciding whether $\OccH_k(P,T)\ne \emptyset$
        using $\Ohtilde(\sqrt{n/m} \cdot \sqrt{km})$ queries and
        $\Ohtilde(\sqrt{n/m} \cdot (\sqrt{km} + k^{2}))$ time.\qedhere
    \end{enumerate}
\end{restatable*}
\begin{restatable*}{mtheorem}{qpmwe}\label{thm:qpmwe}
    Let $P$ denote a pattern of length $m$, let $T$ denote a text of length $n$, and let $k > 0$ denote an integer threshold.
    \begin{enumerate}
        \item There is a quantum algorithm that solves the \PMwE
        problem using $\Ohhat(n/m \cdot \sqrt{km})$ queries and
        $\Ohhat(n/m \cdot (\sqrt{km} + k^{3.5}))$ time.
        \item There is a quantum algorithm deciding whether $\OccE_k(P,T)\ne \emptyset$
        using $\Ohhat(\sqrt{n/m} \cdot \sqrt{km})$ queries and
        $\Ohhat(\sqrt{n/m} \cdot (\sqrt{km} + k^{3.5}))$ time.\qedhere
    \end{enumerate}
\end{restatable*}
\medskip

The query complexity for both problems has a remarkable characteristic:
For the case \(n = \Oh(m)\), justified by the so-called Standard Trick%
\footnote{The Standard Trick divides  \(T\) into blocks of length \(\Theta(m)\), with the
    last block potentially being shorter.
This allows us to solve the pattern matching problem on each block independently.
This approach is optimal in the quantum setting.},
verifying whether the Hamming distance or edit distance between two strings is bounded by
\(k\) has the same query complexity (up to polylogarithmic and subpolynomial factors)
as verifying whether \(\OccH_k(P,T)\ne \emptyset\) and \(\OccE_k(P,T)\ne \emptyset\).
A naive implementation might incur an additional \(\Oh(\sqrt{n})\) factor for the latter
verification on top of the former.
However, our results, together with those from \cite{KNW24}, demonstrate that no overhead
is necessary.

\paragraph*{Technical Contributions in the Quantum Setting and Beyond.}

Our algorithm brings two technical contributions.

\begin{enumerate}[(1)]
    \item Our first contribution is limited to the quantum setting,
        yet extends beyond string problems:
        we present a generalization of the Høyer, Mosca, and de Wolf~\cite{groverext} search algorithm.
        Despite being a relatively straightforward extension,
        we believe this search algorithm could prove highly useful in various contexts.
        For the sake of conciseness, we postpone the detailed discussion of this contribution to
        the Technical Overview.

    \item On the contrary, our second contribution is limited to string problems,
        yet extends beyond the quantum setting:
        we give an efficient implementation of communication complexity results
        for approximate pattern matching.

        Here, \emph{communication complexity} refers to the minimal space required to encode the
        sets \(\OccH_k(P,T)\) and \(\OccE_k(P,T)\), respectively.
        Applying the Standard Trick, we restrict our analysis to the scenario where \(n = \Oh(m)\).
        Clifford, Kociumaka, and Porat~\cite{CKP19}   demonstrated that \(\Oh(k)\) space suffices
        for the mismatches case
        and Kociumaka, Nogler, and Wellnitz~\cite{KNW24} demonstrated that \(\Oh(k \log n)\) space
        suffices for the edits case.
        Both results are surprising, as \(\Oh(k)\) represents the space required to
        encode the mismatches or edits that occur in the pattern for a single occurrence.

        The authors of \cite{KNW24} argue that in both cases,
        the communication complexity can be reframed as follows:
        given \(P, T\) and \(k\), we identify a subset of positions in \(P\) and \(T\) and replace the characters within this subset with sentinel characters \(\#^c\),
        where \(c\) belongs to a set of identifiers.
        This transformation yields \(P^\#, T^\#\), each with distinct substitution rules for mismatches and edits.

        The strings \(P^\#, T^\#\) possess two notable properties.

        First, they preserve occurrences,
        that is, \(\OccH_k(P^\#,T^\#) = \OccH_k(P,T)\) and \(\OccE_k(P^\#,T^\#) = \OccE_k(P,T)\) in their respective scenarios.

        Second, the encoding of \(P^\#\) and \(T^\#\) requires only \(\Oh(k \log n)\) space.
        Despite this concise representation, the challenge lies in efficiently constructing \(P^\#\) and \(T^\#\) from this information,
        breaking the linear barrier.
        This constraint is a primary reason why the quantum algorithm of \cite{KNW24} does not achieve sublinear quantum time.

        In this paper, we demonstrate the \(\Ohtilde(k^2)\)-time construction of grammar-like representations for \(P^\#\) and \(T^\#\),
        each of size \(\Ohtilde(k)\), indeed breaking the linear barrier.
        Our new representations support \modelname operations, a set of basic operations supported
        needed to run the recent algorithms from~\cite{CKW20} and~\cite{CKW22}.
        This allows us to formulate our algorithm for edits such that any improvement in the \(\Oh(k^{3.5})\) term
        for the \PMwE Problem in the \modelname model is reflected in \cref{thm:qpmwe}.

        We believe that the construction of such data structures could play a crucial role in
        future developments aimed at improving the time complexity for \PMwE
        across multiple computational settings.
        Among these settings is the classical one,
        where closing the gap between the current \(\Ohtilde(n + k^{3.5} n/m)\) and the lower
        bound \(\Omega(n + k^{2} n/m)\) remains highly relevant.

        The technical contribution of the grammar-like representation of \(P^\#\) and
        \(T^\#\) is
        achieved through another result of this paper: we show that a generic length-$n$
        string $S$ satisfying a system of $b$ substring equations of the form
        $S\fragmentco{x}{x'}=S\fragmentco{y}{y'}$ admits a compressed representation of
        size $\Ohtilde(b)$ that can be constructed in $\Ohtilde(b^2)$ time.
        Previous work~\cite{GKRR20} provided an $\Oh(n+b)$-time algorithm that constructs
        a generic solution explicitly.
        Once again, we postpone the detailed discussion of this contribution to the
        Technical
        Overview and \cref{sec:encoding}.
\end{enumerate}

\subsection{Open Questions}

We conclude the introduction with open questions and potential future research directions.
\begin{itemize}
    \item Can the \(k^2\) and \(k^{3.5}\) terms be improved in \cref{thm:qpmwm} and
        \cref{thm:qpmwe}?
        These terms are the only obstacles preventing our algorithm from achieving
        near-optimal quantum time for all regimes of~$k$.

        For mismatches, nothing precludes the possibility of replacing the \(k^2\) term with a
        \(k\) term, thereby aligning the time complexity with the \(\Ohtilde(\sqrt{km})\)
        query complexity.
        However, this improvement seems realistic only for verifying whether \(\OccH_k(P,T) \neq \emptyset\).
        For computing \(\OccH_k(P,T)\), already \(\Oh(\min(n,k^2))\) space appears
        necessary for its representation when $n = \Oh(m)$.

        For edits, the best improvement we can realistically hope for is to reduce the
        \(k^{3.5}\) term to a value between \(k^{1.5}\) and \(k^2\).
        This constraint arises from the conditional lower bound for computing the edit distance.
        Indeed, a quantum counterpart of the Strong Exponential Time Hypothesis (SETH)
        rules out an \(\Oh(n^{1.5-\epsilon})\) time algorithm \cite{BPS19},
        leaving a major open problem for the exact optimal quantum time complexity
        \cite{Rubinstein19}, which lies between \(\Omega(n^{1.5})\) and \(\Oh(n^2)\).

        We observe that a \(\Ohtilde(k^2)\) \modelname algorithm for the Pattern Matching with
        Edits Problem would directly yield an \(\Ohhat(\sqrt{n/m} \cdot (\sqrt{km} +
        k^{2}))\)-time quantum algorithm.
        Hence, if it is indeed possible to close the gap mentioned in the introduction and
        to express it in the \modelname model,
        this would provide such an algorithm with no additional effort.

    \item Our classical subroutine solving a system of $b$ substring equations takes
        $\Ohtilde(b^2)$ time to construct an $\Oh(b)$-size compressed representation of a
        universal solution of the system.
        An intriguing open question asks whether the same task can be achieved in
        $\Ohtilde(b)$ time.
        We find it already challenging to decide in $b^{2-\Omega(1)}+n^{1-\Omega(1)}$ time
        whether a given system of $b$ substring equations admits any non-trivial solution
        (other than $S=a^n$ for $a\in \Sigma$; strings of this form satisfy every system
        of substring equations).

    \item We leave open the discovery of new applications for our technical tools.
        Our results, which allows for an efficient algorithmic implementation of communication
        complexity findings,
        further raise the question of whether the running time for both pattern matching
        problems can be improved across various settings.
\end{itemize}

%% file: s1b_tech_ov.tex
\section{Technical Overview}

Throughout this overview, write \(P\) for a pattern of length  \(m\) and \(T\) for a text
of length \(n\).
Via the \emph{Standard Trick}, we can focus on the case when \(n \le 3/2\cdot m\): otherwise, we
consider \(\Oh(n/m)\) fragments of \(T\) that are each of length at most \(3/2\cdot m\) and that
start \(m/2\) positions apart from each other.
While in the standard setting (or when listing all occurrences), the Standard Trick incurs a
running time factor of \(\Oh(n/m)\);
in the quantum setting (and when we wish to only check for the existence of an occurrence),
we can in fact save a factor of \(\Oh(\sqrt{n/m})\) over the standard setting by using
\GS (which we also use in several other places of our algorithms).

\begin{restatable*}[\GS (Amplitude amplification)~\cite{DBLP:conf/stoc/Grover96,brassard2002quantum}]{theoremq}{grover}\label{prp:grover}
    Let $f\colon \fragmentco{0}{n} \to \{0,1\}$ denote a function implemented as a
    black-box unitary operator (oracle) $U_f$ such that
    $U_f\ket{i}\ket{b} = \ket{i}\ket{b\oplus f(i)}$.
    There is a quantum algorithm that, in $\Ohtilde(\sqrt{n})$ time, including $\Oh(\sqrt{n})$
    applications of $U_f$, correctly with probability at least $2/3$, finds an element $x\in f^{-1}(1)$ or reports that $f^{-1}(1)$ is empty.
\end{restatable*}

In the remainder of this Technical Overview, we thus focus on the case when \(n \le 3/2
\cdot m\); in particular, our main objective is to obtain quantum algorithms in this case
that given an integer threshold \(k\) and (quantum) oracle access to \(P\) and \(T\) (refer to
\cref{7.6-2} for a precise definition of our model of computation) compute
\begin{itemize}
    \item (a representation of) \(\OccH_k(P, T)\) using \(\Ohtilde(\sqrt{km})\) quantum queries and
        \(\Ohtilde(\sqrt{km} + k^2)\) extra time; and
    \item (a representation of) \(\OccE_k(P, T)\) using \(\Ohhat(\sqrt{km})\) quantum queries and
        \(\Ohhat(\sqrt{km} + k^{3.5})\) extra time.
\end{itemize}
Combined with our previous discussion, this then yields \cref{thm:qpmwm,thm:qpmwe}.

\subsection{The \modelname Model and How (Not) to Use the Algorithms for the \modelname
Model}

In order to unify algorithms for different settings (standard, fully-compressed, and more),
in~\cite{CKW20}, the authors introduced the \modelname model: instead of working with the
input of a given setting directly, the \modelname model introduces an abstraction in the form
of a small set of basic operations%
\footnote{such as access at a given position, compute the length,
compute longest common prefixes and compute exact internal occurrences; refer to
\cref{sec:pillar} for a complete list and~\cite{CKW20} for a more detailed discussion)}
to work with the input instead.
The complexity of a \modelname algorithm is then measured in the number of basic
\modelname operations it uses.

The authors of~\cite{CKW20} also give \modelname algorithms for approximate pattern matching
Problems; for \PMwE,~\cite{CKW22} yields an improved \modelname
algorithm over~\cite{CKW20}.

\begin{restatable*}[{\cite[Main Theorem 8]{CKW20}}]{theoremq}{hdalg}\label{thm:hdalg}
    Given a pattern $P$ of~length $m$, a text $T$ of~length $n$, and a positive threshold
    $k\le m$, we can compute (a representation of) the set $\OccH_k(P,T)$
    using $\Oh(n/m \cdot k^2)$ \modelname operations and \(\Oh(n/m \cdot k^2 \log\log k)\)
    extra time.
\end{restatable*}

\begin{restatable*}[{\cite[Main Theorem 2]{CKW22}}]{theoremq}{edalg}\label{thm:edalg}
    Given a pattern $P$ of~length $m$, a text $T$ of~length $n$, and a positive threshold
    threshold $k \leq m$, we can compute (a representation of) the set $\OccE_k(P,T)$
    in $\Ohtilde(n/m \cdot k^{3.5} )$ \modelname operations.
\end{restatable*}

As was observed in \cite[Appendix~C]{CPR24}, existing results for the quantum setting
allow implementing all basic \modelname operations using \(\Ohhat(\sqrt{m})\) quantum
queries to the input.
Combined with the \modelname algorithms, we thus directly obtain the following quantum
algorithms for approximate pattern matching problems in our setting.

\begin{corollaryq}\label{7.6-3}
    Given a pattern $P$ of~length $m$, a text $T$ of~length $n \le 3/2\cdot m$, and a positive threshold
    $k\le m$, we can compute
    \begin{itemize}
        \item (a representation of) $\OccH_k(P,T)$ using $\Ohhat(\sqrt{m} \cdot k^2)$ quantum queries
            and time; and
        \item (a representation of) $\OccE_k(P,T)$ using $\Ohhat(\sqrt{m} \cdot k^{3.5})$ quantum queries
            and time.
        \qedhere
    \end{itemize}
\end{corollaryq}

Clearly, the algorithms of \cref{7.6-3} are too slow to prove \cref{thm:qpmwm,thm:qpmwe}.
Hence, we proceed as follows.

\begin{enumerate}
    \item
        Using \(\Ohtilde(\sqrt{km})\) quantum queries (Hamming distance) or
        \(\Ohhat(\sqrt{km})\) quantum queries (edit distance), we
        learn a set \(C\) of \emph{candidates} for starting positions of occurrences,
        where \(C\) is either a single arithmetic progression or a set of \(\Ohtilde(k)\)
        positions.
        While we can ensure (w.h.p.) that \(C\) contains all occurrences with at most
        \(k\) mismatches or edits, \(C\) may also contain some ``occurrences'' with \(\Oh(k)\)
        mismatches or edits.
        This relaxation allows us to compute the set \(C\) efficiently, but necessitates
        that we filter out some of the positions in \(C\).
    \item
        To be useful to us, we need to enhance \(C\) with more information about each
        single occurrence. In particular, for mismatches, we need the \emph{mismatch
        information} for each \(c \in C\), that is, the set of positions \(i\) where
        \(T\position{c + i}\) and \(P\position{i}\) differ (as well as the corresponding
        characters); for edits, we correspondingly need a suitable \emph{alignment}.

        However, computing said extra information for a single position \(c\in C\)
        directly takes \( \Ohtilde(\sqrt{km} + k^2)\)
        time and \(\Oh(k)\) space.%
        \footnote{Observe that if we could directly compute the mismatch/edit information for each
        \(c \in C\), then this would allow us to filter the positions in \(C\) with more
        than \(k\) mismatches/edits and thus solve the overall problem.}
        Hence, we proceed as outlined by the structural results of \cite{CKP19} for
        \PMwM in the streaming setting and the generalization
        of said structural result by \cite{KNW24} to \PMwE.
        As shown in \cite{CKP19,KNW24}, by computing the mismatch information for a set
        \(C' \subseteq C\) of \(\Oh(\log n)\) \emph{important} occurrences,
        we can infer the mismatch information also for the remaining positions in \(C\).

        Using said mismatch/edit information, we write down \(C'\) (and thus implicitly
        also \(C\)) as a set \(E\) of
        \(\Ohtilde(k)\) \emph{substring equations};
        that is, constraints of the form \(S\fragmentco{i}{j} = S\fragmentco{i'}{j'}\).
        We write \(P^{\#}\) and \(T^{\#}\) for the \emph{most general} solution to \(E\).

        Our construction ensures (by \cite{CKP19,KNW24}) that we can correctly compute
        \(k\)-bounded Hamming and edit distances of
        \(P\) and (fragments of) \(T\) based on computing Hamming and edit distances of \(P^{\#}\) and
        \(T^{\#}\); that is, we show that \(\OccH_k(P^{\#}, T^{\#}) = \OccH_k(P, T)\) (in
        the Hamming case) and that \(\OccE_k(P^{\#}, T^{\#}) = \OccE_k(P, T)\) (in the
        edit case).
    \item
        We then develop and apply a classical \(\Ohtilde(k^2)\) algorithm to transform the
        substring equations into an (eXtended) straight-line program (xSLP) of size
        \(\Ohtilde(k)\) (which represents \(P^\#\) and \(T^\#\)).
    \item
        Finally, we run the existing approximate pattern matching algorithms for (x)SLPs.
        While this recomputes all \(k\)-mismatch (or \(k\)-edit) occurrences from scratch
        and thus may feel somewhat wasteful, this allows us to reuse the existing
        algorithms in an opaque manner. In particular, any future improvements to said
        existing algorithms transfer directly to the quantum setting.
\end{enumerate}

We start with a detailed exposition of our algorithm for \PMwM
and then briefly describe the challenges that need to be overcome when generalizing said
algorithms to \PMwE.

\subsection{Step 1: Computing a Set of Good Candidates for Approximate Occurrences}

As our first main step, we aim to efficiently compute a good overestimation \(C\) for the set of
occurrences. To this end, our general approach is very similar to recent approximate
pattern matching algorithms for \modelname model from \cite{CKW20,CKW22}. Crucially, we
have to exploit the slack in \(C\) compared to the actual solution here, so that we can
obtain (quantum) speed-ups. We detail our exploits next; in particular, we discuss
how we obtain the following result.

\begin{restatable*}{lemma}{qhdanalyze}\label{lem:qhd_analyze}
    Let $P$ denote a pattern of length $m$, let $T$ denote a text of length $n \le
    3/2\cdot m$, and let $k > 0$ denote an integer threshold.

    Then, there is a quantum algorithm that using $\Ohtilde(\sqrt{km})$ queries and
    quantum time,
    computes a set $C$ such that $\OccH_k(P,T) \subseteq C \subseteq \fragment{0}{n-m}$
    and
    \begin{itemize}
        \item $|C|=\Ohtilde(k)$, \textbf{or}
        \item $C$ forms an arithmetic progression and $\OccH_k(P,T) \subseteq C \subseteq \OccH_{10k}(P,T)$.
            \qedhere
    \end{itemize}
\end{restatable*}

Toward a proof of \cref{lem:qhd_analyze}, we proceed similarly to many of the recent
state-of-the-art algorithms for approximate pattern matching: we first analyze the
pattern (using \cite[Lemma~3.6]{CKW20}), which gives rise to three possible cases; for
each case we then present separate algorithms.

\begin{restatable}[{\tt Analyze}, {\cite[Lemma~3.6]{CKW20}}]{lemmaq}{analyzeold}\label{prp:I}
    Let $P$ denote a string of~length $m$ and let $k \le m$ denote a positive integer.
    Then, at least one of~the following holds.
    \begin{enumerate}[(a)]
        \item The string $P$ contains $2k$ disjoint \emph{breaks} $B_1,\ldots, B_{2k}$
            each having period $\per(B_i)> m/128 k$ and length $|B_i| = \lfloor
            m/8 k\rfloor$.
            \label{item:prp:I:a}
        \item The string $P$ contains disjoint \emph{repetitive regions} $R_1,\ldots, R_{r}$
            of~total length $\sum_{i=1}^r |R_i| \ge 3/8 \cdot m$ such
            that each region $R_i$ satisfies
            $|R_i| \ge m/8 k$ and has a primitive \emph{approximate period} $Q_i$
            with $|Q_i| \le m/128 k$ and $\hd(R_i,Q_i^*) = \ceil{8 k/m\cdot |R_i|}$.
            \label{item:prp:I:b}
        \item The string $P$ has a primitive \emph{approximate period} $Q$ with $|Q|\le
            m/128 k$ and $\hd(P,Q^*) < 8 k$.
            \label{item:prp:I:c} \qedhere
    \end{enumerate}
\end{restatable}

\medskip

While the proof of \cref{prp:I} is already algorithmic in  \cite{CKW20}, the running time
of said algorithm is too slow for the quantum setting.
Fortunately, we can use the quantum implementation from \cite{JN23} to analyze the
pattern; compared to the statement in \cite{JN23}, we make explicit how the output is
represented.%
\footnote{The algorithm of \cite{JN23} ensures that \(Q_i =
    R_i\fragmentco{0}{|Q_i|}\). Similarly, it is guaranteed that we return at most
\(\Oh(k)\) repetitive regions.}

\begin{restatable*}[{\cite[Lemma 6.4]{JN23}}]{lemmaq}{quantumanalyze} \label{lem:hd_analyzeP}
    Given a pattern $P$ of length $m$,
    in $\Ohtilde(\sqrt{km})$ query and quantum time, we can compute at least one of the
    following.
    \begin{enumerate}[(a)]
        \item A collection of \(2k\) disjoint breaks
            $B_1,\ldots, B_{2k}$,
            each having period $\per(B_i)> m/128 k$ and length $|B_i| = \lfloor
            m/8 k\rfloor$.

            For each break, we return its starting position in \(P\).
        \item A collection of disjoint repetitive regions $R_1,\ldots, R_{r}$
            of~total length $\sum_{i=1}^r |R_i| \ge 3/8 \cdot m$ such
            that each region $R_i$ satisfies
            $|R_i| \ge m/8 k$ and has a primitive \emph{approximate period} $Q_i$
            with $|Q_i| \le m/128 k$ and $\hd(R_i,Q_i^*) = \ceil{8 k/m\cdot |R_i|}$.

            For each repetitive region \(R_i\), we return its starting position in \(P\), its
            length \(|R_i|\), as well as the length \(|Q_i|\).
        \item A primitive \emph{approximate period} $Q$ of \(P\), with $|Q|\le
            m/128 k$ and $\hd(P,Q^*) < 8 k$.

            We return \(Q\) as a position \(i\) in \(P\) together with the length \(|Q|\)
            and an integer \(r\) such that \(Q = \rot^r(P\fragmentco{i}{i + |Q|})\).
            \qedhere
    \end{enumerate}
\end{restatable*}
\medskip

For each of the three different cases, we give a different fast quantum algorithm to
achieve the results claimed in \cref{lem:qhd_analyze}; consult \cref{sec:qhd_analyze} for
details.

\begin{description}
    \item[Breaks.] We sample u.a.r.\ a break \(B^*\) from the set of \(2k\) breaks and
        find the exact occurrences of \(B^*\) in \(T\). We use an adaption of \cite{HV03}
        to compute the exact occurrences in $\Ohtilde(\sqrt{\mathsf{occ} \cdot n})$
        quantum time and queries, where \(\mathsf{occ} = \Oh(k)\) is the solution size
        (which is bounded in size by \cite[Main~Theorem~5]{CKW20}).
        For each exact occurrence \(B^* = P\fragmentco{p}{p + |B^^*|} = T\fragmentco{t}{t + |B^*|}\), we include the
        position \(t - p\) into the result set \(C\).

        First, the set \(C\) that we obtain in this way is of size \(\Oh(k)\).
        Next, as every \(k\)-mismatch occurrence of \(P\) in \(T\) has to match exactly at
        least \(k\) of the breaks, selecting \(B^*\) u.a.r.\ means that each single
        position is included in \(C\) with probability at least a half.
        Hence, \(\Oh(\log n)\) repetitions and taking the union of all sets \(C\) obtained
        boosts the success probability to \(1 - 1/\poly(n)\) while keeping the size of \(C\) at
        \(\Ohtilde(k)\).

        Classical algorithms would now proceed to \emph{verify} each identify position in
        \(C\) to filter out the false positive. This is where we gain: we postpone the
        verification to a later step and just return \(C\) for now.
    \item[Almost Periodicity.]
        It is instructive to discuss the almost periodic case next.
        To that end, observe that \(P\) is at Hamming distance at most \(8k\) to a
        (prefix of the) string \(Q^\infty\), where \(Q\) is a primitive string of length
        at most \(|P|/128k\). In particular, this means that \(P\) contains at least
        \(100k\) exact occurrences of \(Q\); thus also \(T\fragmentco{m - 50k}{m}\) has to
        contain at least \(40k\) exact copies of \(Q\) (or there are no \(k\)-mismatch
        occurrences of \(P\) in \(T\)).

        Hence, we first try to align \(Q^\infty\) with (the center part of) \(T\) such that the number of
        mismatches is minimized.
        In practice, this means that we partition the center
        part of \(T\) into blocks of length \(Q\) and compute (using \cite{HV03}) the unique rotation of \(Q\) that is
        most common among said blocks (if such a unique rotation does not exist, we easily
        obtain that there are no \(k\)-mismatch occurrences).

        Once we have identified a suitable alignment of \(Q^\infty\) to the center part of
        \(T\), we extend said alignment to the left and to the right as much as possible
        while keeping the number mismatches bounded by \(\Oh(k)\).
        In a fashion similar to the algorithm of \cite{bkw19}, this then allows reporting
        good candidates
        \begin{itemize}
            \item  either as an arithmetic progression (take the first position of
                each repetition of an appropriate rotation of \(Q\) in \(T\))
            \item or---if there are at least \(k+1\) mismatches between
                \(P\) and \(Q^\infty\)---as a small set: similar to the breaks case, we
                have to align a mismatch of \(P\) to \(Q^\infty\) to a mismatch \(T\) to
                \(Q^\infty\).
        \end{itemize}

        This is the only case where we might return \(C\) as an arithmetic progression.
    \item[Repetitive Regions.]
        Consider a fixed repetitive region \(R\). Roughly speaking, we can view \(R\) as a
        ``miniature'' (almost) periodic case: as in \cite{CKW20,CKW22}, we run the
        algorithm for the periodic case for \(R\) and an suitably adapted threshold
        \(k'\). As we have a lower bound on the number of mismatches between \(R\) and the
        corresponding period \(Q\), the periodic case returns a small set of occurrences.

        Now, as with the breaks case, we sample (multiple times) a repetitive region;
        however, this time proportionally to the number of its mismatches.
        Finally we return the union of all candidate sets the we obtained in intermediate
        steps.
\end{description}

\subsection{Step 2: From Good Candidates to Substring Equations}

As the next major step, we wish to transform the candidate positions \(C\) into a set of
substring equations.

\begin{restatable*}[Substring Equation~\cite{GK17,GKRR20}]{definition}{nineeighteenone}\label{def:substring_equation}\label{def:universal_solution}
    For integers $0\le x \le x'\le n$ and $0\le y \le y' \le n$ satisfying $x'-x=y'-y$, we
    define a \emph{substring equation} on length-$n$ strings to be a constraint of the form $e
    : T\fragmentco{x}{x'} = T\fragmentco{y}{y'}$, where $T$ is a formal variable representing
    a length-$n$ string.

    A \emph{system} of substring equations is a set of substring equations $E$ on length-$n$
    strings. A string $S\in \Sigma^n$ \emph{satisfies} the system $E$ if and only if, for
    every equation $E\ni e:T\fragmentco{x}{x'} = T\fragmentco{y}{y'}$, the fragments
    $S\fragmentco{x}{x'}$ and $S\fragmentco{y}{y'}$ match, that is, they are occurrences of the
    same substring.

    A string $S\in \Sigma^n$ is a \emph{universal solution} of $E$ if it satisfies $E$ and
    every string $\hat{S}\in \hat{\Sigma}^n$ that satisfies $E$ is an image of $S$ under a
    letter-to-letter morphism (there is a function $\phi : \Sigma \to \hat{\Sigma}$ such
    that $\hat{S}\position{x}=\phi(S\position{x})$ for $x\in \fragmentco{0}{n}$).
\end{restatable*}

In particular, we are interested in a system of substring equations that has a universal
solution from which we can infer Hamming distances between our original strings \(P\) and
\(T\).
To that end, the starting of the candidates \(C\) alone is unfortunately not sufficient;
additionally, we need the \emph{mismatch information} for each \(c \in C\), that is, the
set
\[\MI(P,T\fragmentco{c}{c + m}) \coloneqq \{(i,P\position{i},T\position{c + i})  \mid
P\position{i} \neq T\position{c + i} \}.\]
We call the pair \((c, \MI(P,T\fragmentco{c}{c + m}))\) an \emph{enhanced position} and
say a set of (candidates for) occurrences is enhanced if each of its contained positions
is enhanced.

Using \GS, we are able to compute the mismatch information for a single
position \(c\) in time \(\Oh(\sqrt{km})\) (see \cref{lem:find_mismatches}).
Naturally, this running time is too slow to enhance all of the \(\Oh(k)\) positions in \(C\).
Hence, we first have to take a detour to identify a small set \(C' \subseteq C\) of  \emph{improtant}
positions in \(C\)---for which we then indeed compute the mismatch information.
Selecting \(C'\) cleverly allows us to recover the mismatch information also for the
positions \(c  \in C \setminus C'\).

As before, consider a string \(P\) of length \(m\) and a string \(T\) of length \(n \le
3/2 \cdot m\).
For a simpler exposition, assume that \(C\) contains both \(0\) and \(n-m\),
that is, there is an occurrence of \(P\) both at the start and at the end of \(T\).
(Otherwise, we have to trim \(T\) accordingly.)

Recall our goal to compute (and store) the \emph{mismatch information} for each occurrence in
\(C\); that is, the positions where \(P\) and \(T\) differ as well as the corresponding
characters.
First, observe that due to \(\left\{ 0, n-m \right\} \subseteq C\) and \(n \le 3/2 \cdot m\), this is easy if \(C\)
is an arithmetic progression---the large overlap in the occurrences directly implies that
we can compute and store all mismatches.

For the more general case that \(|C| = \Ohtilde(k)\),
let us first have a quick look at
exact pattern matching (that is, suppose that all occurrences in \(C\) are in fact exact
occurrences).

From the large overlap of the two occurrences, we can infer that both \(P\) and
\(T\) are \emph{periodic} with a common period \(q\) (which is not necessarily the smallest period of
\(P\) and \(T\)); in particular, we learn that there are also exact occurrences of \(P\)
in \(T\) at every position \(0 \le iq \le n- m\).

Now, iterate over the other occurrences in \(C\).
For a fixed \(c \in C\), there are two cases.
\begin{itemize}
    \item If \(q\) divides \(c\), then we do not learn anything; we already concluded that
        there must be an occurrence at position \(c\).
    \item Otherwise, the famous Periodicity Lemma~\cite{FW65} yields that \(P\) and \(T\)
        must in fact be periodic with period \(q' \coloneqq \gcd(q, c) \le q / 2\).
\end{itemize}
Now, call an occurrence in \(C\) \emph{important} if it falls into the second case.
Clearly, which occurrence we label as important depends on the order in which we
process occurrences. However, it is easy to see that it does not matter in which order we process
occurrences,
\begin{itemize}
    \item at most \(\Oh(\log m)\) occurrences are important (for each important
        occurrence, the period \(q\) of \(P\) and \(T\) decreases by a factor of \(2\));
        and
    \item in the end, we have always learned the same period \(q^*\) of \(P\) and \(T\),
        namely \emph{the} period of \(P\) and \(T\).
\end{itemize}
To summarize, \(\Oh(\log m)\) occurrences suffice to learn enough structure about \(P\)
and \(T\) to infer all of the remaining occurrences.

While for exact pattern matching, the above process is somewhat unnecessarily complicated
for computing the period of \(P\) and \(T\) (we could just compute the period of \(P\) and
\(T\) directly); it can be generalized to
both \PMwM~\cite{CKP19} and Pattern Matching with
Edits~\cite{KNW24} to obtain that in said settings as well at most \(\Oh(\log n)\)
(approximate) occurrences suffice to infer the starting positions of all of the remaining
occurrences---and in particular the corresponding mismatch/edit information.

For a somewhat more formal discussion as to why and how the above results are useful to
us, let us first formalize ``what we learn'' from a set of (approximate) occurrences.%
\footnote{As mentioned earlier, we focus on \PMwM for now.}%
\footnote{It is instructive to use a different symbol for the set in \cref{def:bg_h}, as
we intend to iteratively construct inference graphs.}

\begin{restatable*}{definition}{mismgs}\label{def:bg_h}
    For two positive integers \(m \le n\) and a set of positive integers \(S \subseteq
    \fragment{0}{n-m}\), the \emph{inference graph \(G_S\) (of $m$, $n$, and \(S\))} is an undirected graph with
    \begin{itemize}
        \item \(m+n\) vertices \(V(G_S) \coloneqq \{ p_0,\dots,p_{m-1},
            t_0,\dots,t_{n-1}\}\) and
        \item \(p\cdot|S|\) undirected edges \(E(G_S) \coloneqq \{ \{ p_i, t_{i+s}\}  \mid
                i\in\fragmentco{0}{m}, s \in S\}\).
    \end{itemize}

    For strings \(P\) and \(T\) and a set $S \subseteq \OccH_k(P, T)$, the \emph{(mismatch) inference graph}
    $\bG_{S}$ is the inference graph of \(|P|, |T|,\) and \(S\), where we additionally
    color every edge \(\{p_i, t_j\} \in E(\bG_{S})\)
    \begin{itemize}
        \item \emph{black} if \(P\position{i} = T\position{j}\), and
        \item \emph{red} if \(P\position{i} \neq T\position{j}\).
    \end{itemize}
    We say that a connected component of $\bG_{S}$ is \emph{red} if it contains at least
    one red edge; otherwise, we say that the connected component is \emph{black}.
\end{restatable*}

Clearly, in the inference graph for strings, all characters of a black component are
equal; further one can show that \(\bG_{S}\) has exactly \(\gcd(S)\) connected components
(black and red) in total.

Our first key insight is that retaining the information of just the red components of
\(\bG_{S}\) suffices to correctly compute Hamming distances.
Formally, based on \(\bG_{S}\), we first single out the crucial parts of \(P\) and \(T\)
(which we dub the \emph{\(S\)-core} of \(P\) and \(T\), denoted by $P^\#_S$ and $T^\#_S$).

\begin{restatable*}[Construction of $P^\#$ and $T^\#$]{definition}{defpsh}\label{def:psh}
    For a string \(P\) of length $m$, a string \(T\) of length \(n \le 2m\), and a fixed
    set $S\subseteq \fragment{0}{n-m}$,
    we transform \(P\) and \(T\) into
    we denote by $P^\#_S$ and $T^\#_S$ the \emph{\(S\)-cores} of \(P\) and \(T\); which
    are obtained from \(P\) and \(T\) by
    iterating over black connected components of $\bG_S$ and placing a sentinel
    character (unique to the component) at every position that belongs to the component.

    If \(S\) is understood from the context, we also write
    \(P^\# \coloneqq P_S^\#\) and \(T^\# \coloneqq T_S^\#\).
\end{restatable*}

\begin{restatable*}[adapted from~\cite{CKP19}]{theorem}{thmhdsubhash} \label{thm:hd_subhash}
    For a string \(P\) of length $m$, a string \(T\) of length \(n \le 2 m\), and a set
    \(\{0,n-m\} \subseteq S \subseteq \fragment{0}{n-m}\), write
    \(P^\# \coloneqq P_S^\#\) and \(T^\# \coloneqq T_S^\#\) for the \(S\)-cores of \(P\)
    and \(T\).
    \begin{enumerate}
        \item For every $a\in \fragmentco{0}{m}$ and $b \in \fragmentco{0}{n}$, we have
            that $P^\#_S\position{a} = T^\#_S\position{b}$ implies  $P\position{a} = T\position{b}$.\\
            (No new equalities between characters are created.)
            \label{it:hd_subhash:ii}
        \item If $g \coloneqq \gcd(S)$ is a divisor of $x\in \fragment{0}{n-m}$, then $\MI(P, T\fragmentco{x}{x +
            m}) = \MI(P^\#_S, T^\#_S\fragmentco{x}{x + m})$. \\
            (For integer multiples of \(g\), the mismatch information and the Hamming distance are preserved.)
            \label{it:hd_subhash:i}
            \qedhere
    \end{enumerate}
\end{restatable*}
\medskip

As an immediate consequence of \cref{thm:hd_subhash}, we obtain that
if $\{0,|T|-|P|\}\subseteq S$ and \(\gcd(S) = \gcd(\OccH_k(P,
T))\), then \(\OccH_k(P, T) = \OccH_k(P^\#_S, T^\#_S)\).

Recalling the set \(C\) of candidates for occurrences (that is, \(\OccH_k(P, T) \subseteq
C \subseteq \OccH_{10k}(P, T)\)), we observe that \(\gcd(C)\) is a divisor of
\(\gcd(\OccH_{k}(P, T))\) and thus applying  \cref{thm:hd_subhash} to any set \(C'\) with
\(\gcd(C') = \gcd(C)\) yields strings \(P^\#_{C'}\) and \(T^\#_{C'}\)
that correctly capture the set \(\OccH_{k}(P, T)\)
(and even something more).

Further, we obtain a method to compute such a set \(C'\).
\begin{enumerate}
    \item We start with \(C' \coloneqq \left\{ 0, n-m \right\}\) and compute \(g \coloneqq \gcd(C')\).
    \item For each \(c \in C\), we check if \(g\) divides \(c\) and skip \(c\) if that is
        the case.
        Otherwise, we add \(c\) to \(C'\) and update \(g\).
\end{enumerate}
As before, whenever we add a position \(c\) to \(C'\), the greatest common divisor \(g\) decreases by a factor
of at least 2. Further, as we try all \(c \in C\), in the end we obtain a set \(C'\) of
size \(\Oh(\log m)\) with \(\gcd(C) = \gcd(C')\).

Now for each \(c \in C'\), it is efficient enough to compute the mismatch information
using \(\Ohtilde(\sqrt{km})\) quantum time. This in turn would allow us to compute the
strings  \(P^\#_{C'} = P^\#_{C}\) and \(T^\#_{C'} = T^\#_C\).
However, naively representing the strings \(P^\#_C\) and
\(T^\#_C\)
still yields strings of length \(\Oh(m)\), destroying any efforts to obtain a sublinear
running time.

Hence, we instead represent \(P^\#_C\) and \(T^\#_C\) as a set of \emph{substring equations}
(whose most general solution is \(P^\#_C\) and \(T^\#_C\)).

\begin{restatable*}{definition}{nineeighteenthree}\label{def:eq_system_hd}
    Consider a string \(P\) of length \(m\), a string \(T\) of length \(n \le 3/2 \cdot m\),
    and an enhanced set of occurrences \(\eso{C'}\).
    Further, write \(D \coloneqq \sigma_1 \cdots \sigma_c\) for a string of all of the \(c\) different
    characters appearing in \(\eso{C'}\).

    We write \(\exw{P}{T}{C'}\) for the set
    of substring equalities constructed in the following manner.
    For each (enhanced) position $(x, \MI(P, T\fragmentco{x}{x + m})) \in \eso{C'}$, we
    add the following substring equations to $\exw{P}{T}{C'}$:
    \begin{enumerate}[(a)]
        \item For each \((j, \sigma_a, \sigma_b) \in \MI(P, T\fragmentco{x}{x + m})\), we add $e' :
            {P}\position{j} = {D}\position{a}$ and $e :
            {T}\position{x+j} = {D}\position{b}$.
        \label{lem:qhd_proxy:rule:a}
        \item For each maximal interval \(\fragmentco{y}{y'}\subseteq \fragmentco{0}{m}
            \setminus \Mis(P,
            T\fragmentco{x}{x + m})\), we add $e : {P}\fragmentco{y}{y'} = {T}\fragmentco{x+y}{x+y'}$.
            \qedhere
        \label{lem:qhd_proxy:rule:b}
    \end{enumerate}
\end{restatable*}
\medskip

We justify our definition by showing that its most general solution behaves as intended.

\begin{restatable*}{lemma}{nineeighteenfour} \label{lem:str_eq_to_xslp_hd}
    Consider a string \(P\) of length \(m\), a string \(T\) of length \(n \le 3/2 \cdot m\),
    and an enhanced set of occurrences \(\eso{C'}\).
    Further, write
    \(D \coloneqq \sigma_1 \cdots \sigma_c\) for a string of all of the \(c\) different
    characters appearing in \(\eso{C'}\).

    The strings \( (T^\#_{C'},P^\#_{C'},D) \) are the unique solution of \( \exw{P}{T}{C'} \) (up to
    renaming of placeholder characters).
\end{restatable*}

Summarizing the step, we obtain the following result.

\begin{lemmaq}\label{lem:twothree}
    Consider a string \(P\) of length \(m\), a string \(T\) of length \(n \le 3/2 \cdot m\),
    a threshold \(k\), and a set
    $C$ such that $\OccH_k(P,T) \subseteq C \subseteq \fragment{0}{n-m}$
    and
    \begin{itemize}
        \item $|C|=\Ohtilde(k)$, \textbf{or}
        \item $C$ forms an arithmetic progression and $\OccH_k(P,T) \subseteq C \subseteq \OccH_{10k}(P,T)$.
    \end{itemize}
    Then, using \(\Ohtilde(\sqrt{km})\) quantum time, we can compute a
    subset \(C' \subseteq C\) and a
    system of substring
    equalities \(\exw{P}{T}{C'}\) of size \(\Ohtilde(k)\) such that
    \begin{itemize}
        \item the strings \( (T^\#_{C},P^\#_{C},D) \) are the unique solution of
            \(\exw{P}{T}{C'}\) and
        \item we have \(\OccH_k(P, T) = \OccH_k(P^\#_C, T^\#_C)\);
    \end{itemize}
    where
    \(D \coloneqq \sigma_1 \cdots \sigma_c\) is a string of all of the \(c\) different
    characters that appear in \(\eso{C'}\).
\end{lemmaq}

\begin{remark}
    Observe that while the preconditions of \cref{lem:twothree} seem complicated, they are
    exactly what we received as the result of Step 1.
\end{remark}

\subsection{Step 3: From Substring Equations to (eXtended) Straight-Line Programs}

Toward our goal of computing \(k\)-mismatch occurrences of \(P\) in \(T\), we still need
to filter out positions from \(C\) that result in more than \(k\) mismatches.
So far, we learned a compressed representation of the input \(P\) and \(T\) in the form of
a small set of substring equations.

Recalling that the algorithms from \cite{CKW20,CKW22} work in a different but similar
fully-compressed setting, we hope to convert our compressed input to a form that can then
be input to \cite{CKW20,CKW22}. In that way, we are able to elegantly employ their
algorithms, without the need to devise a specific algorithm to filter \(C\).%
\footnote{Observe that this is somewhat wasteful: the algorithms from \cite{CKW20,CKW22}
start by re-doing a lot of Step~1, only to then run deterministic exact algorithms---which
are now fast enough as the input is compressed and compressed inputs allow for efficient
processing. We leave as an open question if a more direct approach to filtering \(C\)
yields more efficient and/or simpler algorithms.}

Hence, the main goal of our next step is to transform systems of substring equations into
compressed formats that can be input to the algorithms of \cite{CKW20,CKW22}; in
particular into (eXtended) Straight-Line Programs.

\begin{restatable*}[Extended SLP]{definition}{xSLP}
    We say that a straight-line grammar $\G$ is an \emph{extended straight-line program}
    (xSLP) if one can distinguish a set $\mP_\G\subseteq \N_\G$ of \emph{pseudo-terminals}
    such that:
    \begin{enumerate}
        \item For every pseudo-terminal $A\in \mP_\G$, we have
            $\rhs_\G(A)=\#^A_0\#^A_1\cdots \#^A_{a-1}$ for distinct terminals
            $\#^A_0,\#^A_1,\ldots,\#^A_{a-1}$ \emph{associated with} $A$ that do not occur
            in any other production.
            We define the \emph{length} of a pseudo-terminal $A\in \mP_\G$ as the length
            of its right-hand side (which equals the length of its expansion), that is,
            $|A|=|\rhs_\G(A)|=|\val_\G(A)|$.
        \item For every other non-terminal $A\in \N_\G\setminus \mP_\G$, we have
            $\rhs(A)=BC$ for symbols $B,C\in \mS_\G$.
    \end{enumerate}
    We call an xSLP $\G$ \emph{pure} if every terminal is associated with a
    pseudo-terminal.
    The \emph{size} of an xSLP $\G$ is defined as the number of symbols excluding
    terminals associated with a pseudo-terminal, that is,  $|\G|=|\mS_\G|-\sum_{A\in
    \mP_G}|A|$.
\end{restatable*}

With the aim of broader applicability, we in fact give a classical and deterministic
algorithm to covert general sets of substring equations into (eXtended) Straight-Line
Programs.

\begin{restatable*}{mtheorem}{solvesubstringequations}\label{prp:solve_substring_equations}
    Given a system $E$ of substring equations on length-$n$ strings, in
    $\Oh(|E|^2\log^7 n)$ time, we can construct an xSLP of size $\Oh(|E|\log^4 n)$ that represents a
    universal solution of $E$.
\end{restatable*}

\begin{remark}
    Upon closer inspection one may notice that our xSLPs differ from the (standard) SLPs
    used in \cite{CKW20,CKW22}. Thus not obvious that xSLPs support \modelname operations,
    as required by the algorithms of \cite{CKW20,CKW22}. Hence, we prove that xSLP do
    indeed support \modelname operations. For details, consider \cref{lem:pillar_on_xslp}
    and its proof.
\end{remark}

The high-level idea for proving \cref{prp:solve_substring_equations} is to process $E$
equation by equation.
We begin with an empty system $E=\emptyset$ and an xSLP $\G$ that consists of a single
pseudo-terminal of length $n$ (observe that a string with $n$ distinct characters
constitutes a universal solution of $E=\emptyset$).
We then incrementally add substring equations to system $E$ and update $\G$ so that it still represents a universal solution of $E$.
To that end, it is instructive to introduce the following data structure problem.

\problembox{%
    \textsc{Dynamic xSLP Substring Equation}\\
    {\bf{Maintained Object:}}
    a set \(E\) of substring equations on length-\(n\) strings
    and a pure xSLP $\G$ that generates a string $S$ of fixed length $n$ that is a
    universal solution to \(E\).\\
    {{\bf{Updates}} that modify \(\G\)}:
    \begin{itemize}
        \item \textsc{Init}$(n)$: Given \( n \in \mathbb{Z}_+ \), set \(E \coloneqq
            \emptyset\) and initialize the xSLP \(\G\) such
            that it contains a single pseudo-terminal of length $n$ (thereby being a
            universal solution to \(E = \emptyset\)).\\
            We assume that this operation is always called exactly once before any other
            operations are called.
        \item \textsc{SetSubstringsEqual}$(x, x', y, y')$:
            Given \(x, x', y, y' \in \fragment{0}{n}\) such that \(x < x'\), \(y < y'\) and \(x' - x = y' - y\),
            add to \(E\) the substring equation \(e: T\fragmentco{x}{x'} =
            T\fragmentco{y}{y'}\), where \(T\) is a formal variable representing a
            length-\(n\) string.
            Further, update \(\G\) such that the represented string $S$ is a
            universal solution to \(E \cup \{e\}\).
    \end{itemize}
    {{\bf{Queries}} that do not modify \(\G\):}
    \begin{itemize}
        \item \textsc{Export}: Return \(\G\).
    \end{itemize}
}%

To obtain \cref{prp:solve_substring_equations}, we give an implementation for
\textsc{Dynamic xSLP Substring Equation} that satisfies the following bounds.

\begin{restatable*}{lemma}{dynslpmain}\label{lem:dynslpmain}
    There is a data structure for \textsc{Dynamic xSLP Substring Equation} such that
    after initialization with $\textsc{Init}(n)$ followed by a sequence of $m$ updates
    ($\textsc{SetStringEqual}$)
    \begin{itemize}
        \item the xSLP $\G$ is of size $\Oh(m \log^4 n)$;
        \item the whole sequence of operations takes time \(\Oh(m^2 \log^7 n)\).
            \qedhere
    \end{itemize}
\end{restatable*}

On a high level, we prove \cref{lem:dynslpmain} as follows.
Given a substring equality, we first extract the corresponding fragments from \(\G\); in
particular, we are interested in pointers to the beginnings and ends of said fragments.
Observe that this extraction operation might require us to \emph{split} some of the pseudo-terminals of
\(\G\). Whenever we do indeed split a pseudo-terminal, we also have to ensure that all of
its other occurrences in \(\G\) are transformed accordingly.

Once we have extracted (pointers to) the substrings to be set equal, we then proceed to
actually replace one such substring with the other.
To that end, we repeatedly compute longest common extensions (to skip over parts that are
already equal), and then set to be equal the next symbol in the grammar (this involves
first splitting off the equal parts).

Using a potential function argument, we show that amortized, in each call to
\textsc{SetSubstringEqual} we ``split'' or ``set equal symbols'' at most \(\Oh(\log n)\) times.

To support the required queries on \(\G\), we employ the known literature on AVL
grammars~\cite{Ryt03,KK20,KL21}. In particular, we show that after in total \(m\) ``set equal'' and
``split'' operations on \(\G\), the string represented by \(\G\) has a compressed size of
\(\Oh(m \log n)\).
This in turn allows us to efficiently recompute a new AVL grammar for \(\G\) after each
such operation, which in turn results in the \(i\)-th ``set equal'' or ``split'' operation
taking time \(\Oh(i \log^5 n)\) and the resulting xSLP \(\G\) having size \(\Oh(i \log^3 n)\).

In total, after \(|E|\) calls to \textsc{SetSubstringEqual}, we thus incur a running time of
\[
    \sum_{i = 1}^{|E| \log n} \Oh(i \log^5 n) = \Oh(|E|^2 \log^7 n);
\]
and obtain an xSLP of size \(\Oh(|E| \log^4 n)\).

Summarizing the step, we obtain the following result.

\begin{lemmaq}\label{lem:twosix}
    Consider a string \(P\) of length \(m\), a string \(T\) of length \(n \le 3/2 \cdot m\),
    a threshold \(k\), and a set
    $C$ such that $\OccH_k(P,T) \subseteq C \subseteq \fragment{0}{n-m}$
    and
    \begin{itemize}
        \item $|C|=\Ohtilde(k)$, \textbf{or}
        \item $C$ forms an arithmetic progression and $\OccH_k(P,T) \subseteq C \subseteq \OccH_{10k}(P,T)$.
    \end{itemize}
    Further, consider a subset \(C' \subseteq C\) and a
    system of substring
    equalities \(\exw{P}{T}{C'}\) of size \(\Ohtilde(k)\) such that
    \begin{itemize}
        \item the strings \( (T^\#_{C},P^\#_{C},D) \) are the unique solution of
            \(\exw{P}{T}{C'}\) and
        \item we have \(\OccH_k(P, T) = \OccH_k(P^\#_C, T^\#_C)\);
    \end{itemize}
    where
    \(D \coloneqq \sigma_1 \cdots \sigma_c\) is a string of all of the \(c\) different
    characters that appear in \(\eso{C'}\).

    There is a classical algorithm that given \(\exw{P}{T}{C'}\),
    in time \(\Ohtilde(k^2)\) computes
    an xSLP  \(\G\) of size \(\Ohtilde(k)\) that represents a universal solution to
    \(\exw{P}{T}{C'}\).
\end{lemmaq}

\begin{remark}
    This preconditions of \cref{lem:twosix} are now even more complicated compared to
    \cref{lem:twothree}; however, as before, they closely mirror the result of
    Step~2.
\end{remark}

\subsection{Step 4: Running the \modelname Algorithms for the Fully-Compressed Setting}

In our final step, we now run the existing \modelname algorithm (\cref{thm:hdalg}) on the xSLP from the
previous step. As mentioned before, for technical reasons, we have to do some lightweight
post-processing on  \(\G\) to allow for \modelname operations.

\begin{restatable*}{lemma}{tentwo} \label{lem:pillar_on_xslp}
    An xSLP $\G$ representing a string $S$ of length $n$ can be preprocessed in
    $\Oh(|\G|\log^2 n)$ time so that the \modelname operations on $S$ can be supported in
    $\Oh(\log n)$ time.
\end{restatable*}
\begin{remark}
    Recall that \(\G\) from the previous step has a size of \(\Ohtilde(k)\). Hence,
    \cref{lem:pillar_on_xslp} runs in time \(\Ohtilde(k)\).
\end{remark}

Summarizing and combining the previous steps, we obtain that
\begin{itemize}
    \item (optionally) using the Standard Trick incurs a multiplicative factor of \(\Oh(n/m)\) in the
        running time;
    \item Step~1~and~2 run in quantum time \(\Ohtilde(\sqrt{km})\) each;
    \item Step~3 runs in time \(\Ohtilde(k^2)\); and
    \item finally running the \modelname algorithm takes again time \(\Ohtilde(k^2)\) (as
        every \modelname operation costs \(\Ohtilde(1)\) time and we have \(n = \Oh(m)\)
        due to using the Standard Trick).
\end{itemize}

All together, we then obtain our claimed main result.

\qpmwm*

\subsection{Additional Challenges for Pattern Matching with Edits}

Our quantum algorithm for \PMwE follows the same high-level approach
as its counterpart for Mismatches, except that combinatorial insight from~\cite{KNW24} is
used instead of the one originating from~\cite{CKP19}.
This makes the resulting procedures much more complicated, at the same time,
since~\cite{KNW24} already provides a quantum algorithm (unlike~\cite{CKP19}) many of the
original lemmas can be re-used in a black-box fashion.
The biggest innovation on top of~\cite{KNW24}, which is shared with Pattern Matching with
Mismatches, is that we represent $P^\#$ and $T^\#$ in compressed space $\Ohtilde(k)$; we
use our novel tool (solving system of substring equations or, equivalently, macro schemes)
to build such a compressed representation in $\Ohtilde(k^2)$ time whenever a new alignment
is added to the set $S$ (based on which the strings $P^\#$ and $T^\#$) are defined.

One step that requires significantly more attention, however, is finding new (approximate)
alignments that could be used to extend $S$.
Recall that, in that setting, we have a subset of positions in $T$ (encoded with a very
efficient oracle deciding if a given position belongs to the set), and our task is to
either find an $\Ohhat(k)$-edit occurrence of $P$ at one of the selected position or
report that no $k$-edit occurrence is present at any of these positions.

A natural way of addressing this task, already present in~\cite{KNW24}, is to use \GS
(over the set of positions) with an oracle that can efficiently tell whether a
given position is (up to some approximation factor) a $k$-edit occurrence of $P$ in $T$.
Unfortunately, we can only use a probabilistic oracle for that task. In particular,
independently of the thresholds we choose, there are positions for which the oracle
returns \yes and \no with essentially equal probability (up to $\pm 1/k$); this is because
we have no guarantee the oracle's behavior when the edit distance at the considered
position is somewhere between $k$ and $\Ohhat(k)$.
The workaround of~\cite{KNW24} was to fix the random bits and pretend that the oracle is
deterministic (up to inverse polynomial errors originating from the fact that it is
implemented as a quantum algorithm simulating a deterministic classical procedure).
Unfortunately, the number of random bits is $\Theta(m)$, which resulted in the algorithm
of~\cite{KNW24} being slower than its classical counterparts (despite achieving
near-optimal quantum query complexity).

A trick that we employ here instead is to generalize \GS so that the
underlying function, for every argument $i$, is allowed to return \yes and \no with an
arbitrary (fixed) probability between $0$ and $1$.
To the best of our knowledge, this setting has not been considered before in the quantum
algorithms literature.
Formally, our setting is as follows.

\problembox{
    \textsc{Search on Bounded-Error and Neutral Inputs}\\
    {\bf{Setting:}} suppose $I \coloneqq \fragmentco{0}{n}$ is partitioned into $I = I^{\bm{+}} \uplus I^{\bm{\sim}} \uplus I^{\bm{-}}$.\\
    {\bf{Input:}} a randomized function $f : I\to \{0,1\}$ such that $p_i \coloneqq
    \pr{f(i)=1}$ satisfies the following:
    \begin{itemize}
        \item if $i \in  I^{\bm{+}}$, then $p_i \geq 9/10$;
        \item if $i \in I^{\bm{-}}$, then $p_i < 1/10$; and
        \item if $i \in I^{\bm{\sim}}$, then $p_i$ can be arbitrary.
    \end{itemize}
    The function $f$ is accessible through a black-box unitary operator (oracle) $U_f$ such that \[U_f\ket{i}\ket{b}=\sqrt{p_i}\ket{i}\ket{b\oplus 1} + \sqrt{1-p_i}\ket{i}\ket{b}.\]
    {\bf{Output:}} an arbitrary $i \in I^{\bm{+}} \cup I^{\bm{\sim}}$ or a value $\bot$ indicating  $I^{\bm{+}} = \emptyset$.
}

In other words, except for the usual ``good'' entries (for which the oracle predominantly
returns \yes) and ``bad'' entries (for which the oracle predominantly returns \no), we
also have ``neutral'' entries, for which we have no guarantee about the probability of
\yes/\no answers.
In principle, the algorithm is allowed to report a neutral element (if it encounters one),
ignore all neutral elements, or do anything in between.
Classically, dealing with such elements is easy, but in the quantum model their presence
significantly affects the analysis of \GS.
Nevertheless, we were able to adapt the analysis of the Høyer, Mosca, and de Wolf
\cite{groverext} search algorithm, separately keeping track of the contribution of all
three kinds of elements to the maintained quantum state.
Formally, we achieve the following result.%
\footnote{The dependency on the error probability $\delta$ has been optimized.}
\begin{restatable*}{lemma}{groverext}\label{lem:groverext}
    There exists a quantum algorithm that, in $\Ohtilde(\sqrt{n})$ time using $\Oh(\sqrt{n})$
    applications of $U_f$, solves \pn{Search on Bounded-Error and Neutral Inputs}
    correctly with probability at least $2/3$.
    For any $\delta \in (0,1/3)$, the error probability can be reduced to at most
    $\delta$ at the cost of increasing the complexity to $\Oh(\sqrt{n}\log
    \delta^{-1}+\log^3\delta^{-1})$.
\end{restatable*}

We use \cref{lem:groverext} not only to find positions where $\Ohhat(k)$-edit occurrences
of $P$ are located but also within the implementation of our approximation algorithm that
tells apart positions with edit distance less than from positions with edit distance more
than $\Ohhat(k)$.
This is possible because the underlying classical procedure~\cite{gapED} makes a few
recursive calls and decides on the answer depending on how many of these calls return
\yes.
In the original version~\cite{gapED}, the threshold is set to a small constant ($5$).
Here, we instead set the threshold to $1$ and repeat the procedure multiple times, which
is more compatible with~\cref{lem:groverext} yet requires redoing parts of the analysis
(note that the standard approach of listing \yes answers from left to right, akin to the
computation of Hamming distance, does not work because the same oracle call may return
different answers when applied multiple times).

%% file: s2_prelim.tex
\section{Preliminaries}
\label{sec:prelims}

For integers $i,j \in \mathbb{Z}$, we use the notation $\fragment{i}{j}$ to denote the set $\{i, \dots, j\}$,
while $\fragmentco{i}{j}$ represents the set $\{i ,\dots, j - 1\}$.
Similarly, we define $\fragmentoc{i}{j}$ and $\fragmentoo{i}{j}$.
When considering a set $S$, we denote by $kS$ the set obtained by multiplying every element in $S$ by $k$,
that is, $kS \coloneqq \{ k\cdot s \mid s \in S\}$.
Furthermore, we define $\floor{S/k}$ as $\{ \floor{s/k} \mid s \in S \}$
and $(S \bmod p)$ as $\{s \bmod p \mid s \in S\}$.

\subsection{Strings, Compressibility, and Approximate Pattern Matching}

\paragraph*{Strings}

An \emph{alphabet} \(\Sigma\) is a set of characters.
We write \(X=X\position{0}\, X\position{1}\cdots X\position{n-1} \in \Sigma^{n}\)
to denote a \textit{string} of length \(|X|=n\) over \(\Sigma\).
For a \emph{position} \(i \in \fragmentco{0}{n}\) we use \(X\position{i}\)
to denote the \(i\)-th character of \(X\).
For indices \(0 \leq i < j \leq |X|\), we write  \(X\fragmentco{i}{j}\)
for the (contiguous) subsequence of characters \(X\position{i} \cdots X\position{j-1}\);
we say that \(X\fragmentco{i}{j}\) is the \emph{fragment} of \(X\) that starts at  \(i\)
and ends at \((j-1)\).
We may also write
\(X\fragment{i}{j-1}\), \(X\fragmentoc{i-1}{j-1}\), or \(X\fragmentoo{i-1}{j}\)
for the fragment \(X\fragmentco{i}{j}\).

A string \(Y\) of length \(m\) with \(0<m\leq n\)
is a \emph{substring} of another string \(X\) (of length \(n\)),
denoted by \(Y \substr X\),
if for (at least) one \(i \in\fragment{0}{n-m}\), the fragment \(X\fragmentco{i}{i + m}\)
is identical to \(Y\); that is, there is an \emph{exact occurrence}
of \(Y\) at position \(i\) in \(X\).
Further,  \(\Occ(Y,X) \coloneqq \{i\in \fragment{0}{n-m} \mid Y = X\fragmentco{i}{i+m}\}\)
represents the set of starting positions of all (exact) occurrences of \(Y\) in \(X\).
Combined, we have \( Y \substr X \Leftrightarrow  \Occ(Y,X) \neq \emptyset\).

A \emph{prefix} of a string \(X\) is a fragment that starts at position \(0\),
while a \emph{suffix} of a string \(T\) is a fragment that ends at position \(|T|-1\).

We denote by \(AB\) the concatenation of strings \(A\) and \(B\).
Further, we denote by \(A^k\) the concatenation of \(k\) copies of the string \(A\),
while we denote by \(A^\infty\) the concatenation of an infinite number of copies of \(A\).
A \textit{primitive} string is a string that cannot be expressed as \(X^k\) for any string \(X\)
and any integer \(k > 1\).

A positive integer \(p\) is \emph{a period} of a string \(X\) if \(X\position{i} = X\position{i + p}\)
for all \(i \in \fragmentco{0}{n-p}\).
\emph{The period} of a string \(X\), denoted by \(\per(X)\), is the smallest period \(p\) of \(X\).
A string \(X\) is \textit{periodic} if \(\per(X) \le |X| / 2\), and \emph{aperiodic} if it is not periodic.

For a string \(X\), we define the following \emph{rotation} operations.
The operation \(\rot(\cdot)\) takes as input a string and moves its last character to the
front; that is, $\rot(X) \coloneqq X\position{n-1}X\fragment{0}{n-2}$. The inverse
operation \(\rot^{-1}(\cdot)\) takes as input a string and moves its initial character to
the end; that is,
$\rot^{-1}(X) \coloneqq X\fragment{1}{n-1}X\position{0}$.

A primitive string \(X\) does not match any of its non-trivial rotations.
Specifically, we have \(X = \rot^j(X)\) if and only if \(j \equiv 0 \pmod{|X|}\).

In this paper we utilize the following lemma, which is a consequence of the Periodicity Lemma~\cite{FW65}.

\begin{lemmaq}[{\cite[Lemma~3.2]{KNW24}}]\label{fct:periodicity}
    Consider strings $P$ and $T$ with $|T| \le 2|P| + 1$.
    If $\{0,|T|-|P|\}\subseteq \Occ(P,T)$, that is, $P$ occurs both as a prefix and as a
    suffix of $T$, then $\gcd(\Occ(P,T))$ is a period of $T$.%
    \footnote{Observe that this lemma makes a statement only about the periodicity of
    \(T\) (and not about \(P\)). In particular, the statement is trivial whenever
    $\gcd(\Occ(P,T)) \ge |P|$.}
\end{lemmaq}

\paragraph*{String Compression}
Throughout this work, we utilize several compression methods briefly introduced next.
For a comprehensive overview, we refer to the survey by Navarro \cite{N21}.

\begin{description}
    \item[Lempel--Ziv Factorizations.]
        A fragment $X\fragmentco{i}{i+\ell}$ is a \emph{previous factor} in \(X\) if
        $X\fragmentco{i}{i+\ell}$ has an earlier occurrence in $X$, that is,
        $X\fragmentco{i}{i+\ell}=X\fragmentco{i'}{i'+\ell}$ holds for some $i'\in
        \fragmentco{0}{i}$.
        An \emph{LZ77-like factorization} of $X$ is a factorization $X = F_1\cdots F_f$
        into non-empty \emph{phrases} such that each phrase $F_j$ with $|F_j|>1$ is a
        previous factor.
        In the underlying \emph{LZ77-like representation}, phrases are encoded as follows.
        \begin{itemize}
            \item A previous factor phrase $F_j=X\fragmentco{i}{i+\ell}$ is encoded as
                $(i',\ell)$, where $i'\in \fragmentco{0}{i}$ satisfies
                $X\fragmentco{i}{i+\ell}=X\fragmentco{i'}{i'+\ell}$.
                The position \(i'\) is chosen arbitrarily in case of ambiguities.
            \item Any other phrase $F_j=X\position{i}$ is encoded as $(X\position{i},0)$.
        \end{itemize}

        The LZ77 factorization~\cite{DBLP:journals/tit/ZivL77} (or the LZ77 parsing) of a string
        $X$, denoted by $\LZ(X)$ is an LZ77-like factorization that is constructed by greedily
        parsing $X$ from left to right into the longest possible phrases.
        More precisely, the $j$-th phrase $F_j=X\fragmentco{i}{i+\ell}$ is the longest previous
        factor starting at position $i$; if no previous factor starts at position \(i\), then
        $F_j$ is the single character \(X\position{i}\).
        It is known that the aforementioned greedy approach produces the shortest possible
        LZ77-like factorization~\cite{DBLP:journals/tit/ZivL77}.

    \item[Straight-Line Grammars and Straight-Line Programs.]
        For a context-free grammar $\G$, we denote by $\Sigma_\G$ and $\N_\G$ the sets of
        terminals and non-terminals of $\G$, respectively.
        Moreover, $\S_\G=\Sigma_\G\cup\N_\G$ is the set of \emph{symbols} of $\G$.
        We say that $\G$ is a~\emph{straight-line grammar} (SLG) if:
        \begin{itemize}
            \item each non-terminal $A\in \N_\G$ has a unique production $A\to
                \rhs_\G(A)$, whose right-hand side is a non-empty sequence of symbols,
                that is, $\rhs_\G(A)\in \S_\G^+$, and
            \item there is a partial order $\prec$ on $\S_\G$ such that $B \prec A$ if $B$
                occurs in $\rhs_\G(A)$.
        \end{itemize}
        A straight-line grammar is a \emph{straight-line program} (SLP) if
        $|\rhs_\G(A)|=2$ holds for each $A\in \N_\G$.

        Every straight-line grammar $\G$ yields an \emph{expansion} function $\val_\G:
        \S_\G \to \Sigma_\G^*$:
        \[\val_\G(A) = \begin{cases}
            A & \text{if $A\in \Sigma_\G$},\\
            \val_\G(B_0)\val_\G(B_1)\cdots \val_\G(B_{a-1}) & \text{if $A\in \N_\G$ and $\rhs_\G(A)=B_0B_1\cdots B_{a-1}$.}
        \end{cases}\]
        The string represented by $\G$, denoted $\val(\G)$, is the expansion $\val_\G(A)$
        of the starting symbol $A$ of~$\G$.
        We sometimes extend the notion of a straight-line grammar $\G$ so that there are
        multiple starting symbols $A_0,\ldots,A_{s-1}$
        each representing a separate string $\val_\G(A_i)$.
\end{description}

\paragraph*{Hamming Distance and Pattern Matching with Mismatches}

For two strings $S$ and $T$ of the same length $n$, the set of \emph{mismatches} between
$S$ and $T$ is defined as
\[\Mis(S,T) \coloneqq\{i\in \fragmentco{0}{n} \mid S\position{i}\ne T\position{i}\}.\]
The \emph{Hamming distance} of $S$ and $T$ is then the number of mismatches between them,
denoted as $\hd(S,
T) \coloneqq |\Mis(S,T)|$.

Following~\cite{CKP19}, we also define \emph{mismatch information}
\[\MI(S,T) \coloneqq \{(i,S\position{i},T\position{i}) : i\in \Mis(S,T)\}.\]
It is easy to verify that the Hamming distance satisfies the triangle inequality. That is,
for any strings $A$, $B$, and $C$ of~the same length satisfy
\[\hd(A, C) + \hd(C, B) \ge \hd(A, B) \ge |\hd(A, C) - \hd(C, B)|.\]

Since we are often interested in the Hamming distance between a string \( S \) and a
prefix of \( T^{\infty} \) for a given string \( T \), we write
\[\Mis(S, T^*) \coloneqq \Mis(S, T^{\infty}\fragmentco{0}{|S|}) \quad\text{and}\quad
\hd(S, T^*) \coloneqq |\Mis(S, T^*)|.\]
Here, \(\Mis(S, T^*)\) denotes the set of positions where \( S \) and the prefix of \(
T^{\infty} \) of the same length as \( S \) differ, and \(\hd(S, T^*)\) represents the
Hamming distance, which is the cardinality of this set.

For a string $P$ (also referred to as a \emph{pattern}), a string $T$ (also referred to as
a \emph{text}),
and an integer $k\ge 0$ (also referred to as a \emph{threshold}), a \emph{$k$-mismatch
occurrence} of $P$ in $T$
is a position $i \in \fragment{0}{|T|-|P|}$ such that $\hd(P, T\fragmentco{i}{i+|P|})\leq
k$.
We denote by $\OccH_k (P,T)$ the set of all positions of $k$-mismatch occurrences of $P$
in $T$ as $\OccH_k(P,T)$, that is,
\[\OccH_k (P,T)\coloneqq \{i \mid \hd(P,T\fragmentco{i}{i+|P|}\leq k)\}.\]

\paragraph*{Edit Distance and Pattern Matching with Edits}
The \emph{edit distance}, also known as the \emph{Levenshtein
distance}~\cite{Levenshtein66}, between two strings \(X\) and \(Y\), denoted by
\(\ed(X,Y)\), is the minimum number of character insertions, deletions, and substitutions
required to transform \(X\) into \(Y\). To formalize this, we define the notion of an
\emph{alignment}.

\begin{definition}[{\cite[Definition 2.1]{CKW22}}]\label{def:alignment}
    A sequence $\mA=(x_i,y_i)_{i=0}^{m}$ is an \emph{alignment} of $X\fragmentco{x}{x'}$
    onto $Y\fragmentco{y}{y'}$, denoted by \(\mA: X\fragmentco{x}{x'} \onto
    Y\fragmentco{y}{y'}\), if it satisfies $(x_0,y_0)=(x,y)$, $(x_{i+1},y_{i+1})\in
    \{(x_{i}+1,\allowbreak y_{i}+1),(x_{i}+1,\allowbreak y_{i}),(x_{i},y_{i}+1)\}$ for
    $i\in \fragmentco{0}{m}$, and $(x_m,y_m) =(x',y')$. Moreover, for $i\in
    \fragmentco{0}{m}$:
    \begin{itemize}
        \item If $(x_{i+1},y_{i+1})=(x_{i}+1,y_{i})$, we say that $\mA$ \emph{deletes}
            $X\position{x_i}$.
        \item If $(x_{i+1},y_{i+1})=(x_{i},y_{i}+1)$, we say that $\mA$ \emph{inserts}
            $Y\position{y_i}$.
        \item If $(x_{i+1},y_{i+1})=(x_{i}+1,y_{i}+1)$, we say that $\mA$ \emph{aligns}
            $X\position{x_i}$ to $Y\position{y_i}$.
        If~additionally $X\position{x_i}=Y\position{y_i}$, we say that $\mA$
        \emph{matches} $X\position{x_i}$ and $Y\position{y_i}$; otherwise, $\mA$
        \emph{substitutes} $X\position{x_i}$ with $Y\position{y_i}$. \qedhere
    \end{itemize}
\end{definition}

The \emph{cost} of an alignment $\mA$ of $X\fragmentco{x}{x'}$ onto $Y\fragmentco{y}{y'}$,
denoted by $\edal{\mA}(X\fragmentco{x}{x'},Y\fragmentco{y}{y'})$, is the total number of
characters that $\mA$ inserts, deletes, or substitutes.
The edit distance $\ed(X,Y)$ is the minimum cost of an alignment of $X\fragmentco{0}{|X|}$
onto~$Y\fragmentco{0}{|Y|}$.
An alignment of $X$ onto $Y$ is \emph{optimal} if its cost is equal to $\ed(X, Y)$.

The following notion of \emph{edit information} can be used to encode alignments in space
proportional to their costs.

\begin{definition}[{Edit information}, {\cite[Definition~3.5]{KNW24}}]\label{def:edinfo}
    For an alignment $\mA=(x_i,y_i)_{i=0}^m$ of $X\fragmentco{x}{x'}$ onto $Y\fragmentco{y}{y'}$,
    the \emph{edit information} is defined as the set of 4-tuples $\sE_{X,
        Y}(\mA)=\{(x_i,\mathsf{cx}_i \mid y_i,\mathsf{cy}_i) : i\in
    \fragmentco{0}{m}\text{ and }\mathsf{cx}_i\ne \mathsf{cy}_i\}$, where
    \[\mathsf{cx}_i = \begin{cases}
        X\position{x_i} & \text{if }x_{i+1}=x_i+1,\\
        \varepsilon & \text{otherwise};
    \end{cases}\qquad\text{and}\qquad
    \mathsf{cy}_i = \begin{cases}
        Y\position{y_i} & \text{if }y_{i+1}=y_i+1,\\
        \varepsilon & \text{otherwise}.
    \end{cases}
    \]
\end{definition}

Further, we denote the minimum edit distance between a string $S$ and any prefix of a
string $T^\infty$ by
\[\ed(S, T^\infty) \coloneqq \min\{\ed(S,T^\infty\fragmentco{0}{j}) \mid j \in \Zz\}.\]
Similarly, we denote the minimum edit distance between a string $S$ and any substring of
$T^\infty$ by
\[\edl{S}{T} \coloneqq \min\{\ed(S,T^\infty\fragmentco{i}{j}) \mid i, j \in \Zz\text{ and
}i \le j\}.\]

For a pattern $P$, a text $T$,
and a threshold $k\ge 0$, we say that there is a $k$-error or $k$-edits occurrence of $P$
in $T$ at position $i\in \fragment{0}{|T|}$ if $\ed(P, T\fragmentco{i}{j})\leq k$ holds
for some position $j\in \fragment{i}{|T|}$.
The set of all starting positions of $k$-error occurrences of $P$ in $T$ is denoted by
$\OccE_k(P,T)$; formally, we set
\[
    \OccE_k (P,T)\coloneqq \{i\in \fragment{0}{|T|} \mid \exists_{j\in \fragment{i}{|T|}}, \ed(P,T\fragmentco{i}{j}\leq k)\}.
\]

Lastly, we write $\mA_{P,T}$ for the set of all (not necessarily optimal) alignments of $P$
onto fragments of $T$ of cost at most $k$.
The subset $\mA_{P,T}^{\leq k} \subseteq \mA_{P,T}$ contains all alignments of cost at
most $k$.

\subsection{Quantum Algorithms}
\label{7.6-2}

\paragraph*{The Quantum Model}

We think of the input string $S \in \Sigma^n$ being accessible within a quantum query
model \cite{ambainis2004quantum,DBLP:journals/tcs/BuhrmanW02}: an input oracle $O_S$
performs the unitary mapping $O_S$ such that $O_S \ket{i}\ket{b} = \ket{i}\ket{b \oplus
S\position{i}}$ for any index $i \in \fragmentco{0}{n}$ and any character $b \in \Sigma$.

The \emph{query complexity} of a quantum algorithm is the number of queries it makes to
the input oracles.
Additionally, the \emph{time complexity} of the quantum algorithm accounts for the
elementary gates \cite{PhysRevA.52.3457} needed to implement the unitary operators
independent of the input.
We assume quantum random access quantum memory (the so-called QRAG model), as in prior works \cite{GS22,AJ22,ambainis2004quantum}.

An algorithm succeeds \emph{with high probability (w.h.p.)} if the success probability can
be made at least $1 - 1/n^c$ for any desired constant $c > 1$.
A bounded-error algorithm (with success probability $2/3$) can be boosted to succeed
w.h.p.\ by $O(\log n)$ repetitions.
In this paper, we do not optimize poly-logarithmic factors of the quantum query complexity
(and time complexity) of our algorithms.

For our quantum algorithms, we rely on the following primitive quantum operation.

\grover

\paragraph*{Quantum Search on Bounded-Error Inputs and Neutral Inputs}

In the existing literature, there are quantum search algorithms that handle
bounded-error inputs, such as the search algorithm by Høyer, Mosca, and
de~Wolf~\cite{groverext}. Their algorithm considers cases where evaluations of \( f(i) \)
have two-sided errors; that is, if \( f(i) = 1 \), the algorithm outputs \( 1 \) with
probability at least \( 9/10 \), and if \( f(i) = 0 \), it outputs \( 0 \) with
probability at least \( 9/10 \).

However, to the best of our knowledge, there is no prior
work addressing the addition of ``neutral'' elements for which we lack any guarantee
(which we need in our setting).
Consequently, we modify the search algorithm from \cite{groverext} to accommodate these
neutral elements as well.
To that end, first recall the \pn{Search on Bounded-Error and Neutral Inputs} problem from the
Technical Overview.

\problembox{
    \textsc{Search on Bounded-Error and Neutral Inputs}\\
    {\bf{Setting:}} suppose $I \coloneqq \fragmentco{0}{n}$ is partitioned into $I =
    I^{\bm{+}} \uplus I^{\bm{\sim}} \uplus I^{\bm{-}}$.\\
    {\bf{Input:}} a randomized function $f : I\to \{0,1\}$ such that $p_i \coloneqq
    \pr{f(i)=1}$ satisfies the following:
    \begin{itemize}
        \item if $i \in  I^{\bm{+}}$, then $p_i \geq 9/10$;
        \item if $i \in I^{\bm{-}}$, then $p_i < 1/10$; and
        \item if $i \in I^{\bm{\sim}}$, then $p_i$ can be arbitrary.
    \end{itemize}
    The function $f$ is accessible through a black-box unitary operator (oracle) $U_f$
    such that \[U_f\ket{i}\ket{b}=\sqrt{p_i}\ket{i}\ket{b\oplus 1} +
    \sqrt{1-p_i}\ket{i}\ket{b}.\]
    {\bf{Output:}} an arbitrary $i\in I^{\bm{+}} \cup I^{\bm{\sim}}$ or a sentinel value
    $\bot$, indicating that $I^{\bm{+}}=\emptyset$.
}

For sake of convenience, we say $i\in I$ is \emph{good} if $i \in I^{\bm{+}}$, \emph{bad}
if $i \in I^{\bm{-}}$, and \emph{neutral} otherwise.

\groverext
\begin{proof}
    To create the necessary quantum subroutine, we can use a nearly identical algorithm as
    presented in~\cite{groverext}.
    The key differences lie in the analysis, while the notation and most of the algorithm
    remain the same.
    For a comprehensive understanding and intuition of the algorithm,
    we refer directly to~\cite{groverext}.
    Here, we maintain the same notation and highlight only the sections where the analysis
    and algorithm differ.

    First, set
    \begin{align*}
        \Gamma^{\bm{+}} &\coloneqq \{i\in I : p_i \geq 9/10\}, \qquad
        \Gamma^{\bm{-}} \coloneqq \{i\in I : p_i < 1/5\},\qquad\text{and}\qquad
        \Gamma^{\bm{\sim}}\coloneqq \{i\in I : 1/5 \le p_i < 9/10\}.
    \end{align*}
    Now, for $k \in \mathbb{Z}_+$ and as outlined in~\cite{groverext}, a unitary operator
    $A_k$ is constructed so that
    \[
        A_k \ket{\mathbf{0}} =
        \alpha_k^{\bm{+}} \ket{\Gamma_k^{\bm{+}}} \ket{1} +
        \alpha_k^{\bm{\sim}} \ket{\Gamma_k^{\bm{\sim}}} \ket{1} +
        \alpha_k^{\bm{-}} \ket{\Gamma_k^{\bm{-}}} \ket{1} +
        \sqrt{1 - \left({\alpha_k^{\bm{+}}}\right)^2 -\left({\alpha_k^{\bm{\sim}}}\right)^2 - \left({\alpha_k^{\bm{-}}}\right)^2} \cdot \ket{H_k} \ket{0},
    \]
    where $\alpha_k^{\bm{+}},  \alpha_k^{\bm{\sim}},  \alpha_k^{\bm{-}}$ are non-negative
    real numbers,
    $\ket{\Gamma_k^{\bm{+}}}$ (respectively $\ket{\Gamma_k^{\bm{\sim}}}$ and
    $\ket{\Gamma_k^{\bm{-}}}$) is a unit vector whose first register is spanned by
    $\ket{i}$
    for $i \in \Gamma^{\bm{+}}$  (respectively $\Gamma^{\bm{\sim}}$, $\Gamma^{\bm{-}}$),
    and
    $\ket{H_k}$ is another unit vector.

    Initially, $A_1$ is constructed by running $U_f$ on the uniform superposition of all
    inputs:
    \[A_1 \ket{\mathbf{0}} = \frac{1}{\sqrt{n}} \sum_{i\in I}
        \ket{i}\left(\sqrt{p_i}\ket{\psi_{i,1}}\ket{1} +
    \sqrt{1-p_i}\ket{\psi_{i,0}}\ket{0}\right),\]
    where the states $\ket{\psi_{i,1}}$ and $\ket{\psi_{i,0}}$ describe the workspace of
    $U_f$.
    As a result, we have
    \[(\alpha_1^{\bm{+}})^2 = \frac1n\sum_{i\in \Gamma^{\bm{+}}} p_i,\qquad
        (\alpha_1^{\bm{\sim}})^2 = \frac1n\sum_{i\in \Gamma^{\bm{\sim}}} p_i,\qquad
        \text{and}\qquad (\alpha_1^{\bm{-}})^2 = \frac1n\sum_{i\in \Gamma^{\bm{-}}} p_i.
    \]

    For $k \in \mathbb{Z}_+$, the unitary operator $A_{k+1}$ is constructed
    in~\cite{groverext} from $A_{k}$ using an amplitude amplification operator $G_k$ and
    error-reduction operator $E_k$.
    Applying the amplification operator yields

    \begin{align*}
        G_kA_k \ket{\mathbf{0}} &=
        \frac{\sin(3\theta_k)}{\sin(\theta_k)} \Big(\alpha_k^{\bm{+}} \ket{\Gamma_k^{\bm{+}}} \ket{1} +
            \alpha_k^{\bm{\sim}} \ket{\Gamma_k^{\bm{\sim}}} \ket{1} +
        \alpha_k^{\bm{-}} \ket{\Gamma_k^{\bm{-}}} \ket{1}\Big)
                              \\&\qquad+ \sqrt{1 -  \left(\frac{\sin(3\theta_k)}{\sin(\theta_k)}\right)^2
                                  \left(\left({\alpha_k^{\bm{+}}}\right)^2 +\left({\alpha_k^{\bm{\sim}}}\right)^2 +
                              \left({\alpha_k^{\bm{-}}}\right)^2\right)} \cdot \ket{H_k} \ket{0},
    \end{align*}
    where $\theta_k \in \intvl{0}{\pi/2}$ is chosen such that $\sin^2 (\theta_k) = ({\alpha_k^{\bm{+}}})^2+ ({\alpha_k^{\bm{\sim}}})^2 + ({\alpha_k^{\bm{-}}})^2$,
    and we assume $\sin(3\theta_k) = \sin(\theta_k)$ for $\theta_k=0$.

    Next, in the error reduction step, we apply a majority voting on $r=\Theta(k)$ runs of
    the $U_f$.
    Technically, the error reduction operator $E_k$ is defined as follows (ignoring the
    workspace, which is added to the second register):
    \begin{align*}
        E_k\ket{i}\ket{0}\ket{0}&=\ket{i}\ket{0}\ket{0},\\
        E_k\ket{i}\ket{1}\ket{0}&=\sqrt{\pr{B(r,p_i)\ge r/2}}\cdot
        \ket{i}\ket{1}\ket{1}+\sqrt{\pr{B(r,p_i)< r/2}}\cdot \ket{i}\ket{1}\ket{0},
    \end{align*}
    where $B(r,p_i)$ denotes the number of obtained 1s on $r$ samples of $f_i$.
    If $i\in \Gamma^{\bm{+}}$, then $\ex{B(r,p_i)}=p_i\cdot r \ge 0.9r$.
    Hence, by a suitable Chernoff bound,
    \[\pr{B(r,p_i) < 0.5r} \le \exp(-\Omega(r))\] does not
    exceed
    ${1}/{2^{k+5}}$ for sufficiently large $r=\Oh(k)$.

    Similarly, if $i\in \Gamma^{\bm{-}}$, then $\ex{B(r,p_i)}=p_i\cdot r < 0.2r$.
    Hence, by a suitable Chernoff bound,
    \[\pr{B(r,p_i) \ge 0.5r} \le \exp(-\Omega(r))\] does not
    exceed ${1}/{2^{k+5}}$ for sufficiently large $r=\Oh(k)$.

    This gives us the following relations (in which we assume $\theta_k \in
    \intvl{0}{{\pi}/{3}}$ so that $\sin(3\theta_k)\ge 0$):
    \begin{equation}\label{eq:groverext:1}
        \alpha_{k+1}^{\bm{+}} \geq  \alpha_{k}^{\bm{+}} \cdot
        \frac{\sin(3\theta_k)}{\sin(\theta_k)} \sqrt{1 - \frac{1}{2^{k+5}}}
        = \alpha_{k}^{\bm{+}} \cdot 3\cdot \left(1-\frac{4}{3}\sin^2(\theta_k)\right)\sqrt{1 - \frac{1}{2^{k+5}}}
    \end{equation} and
    \begin{equation}\label{eq:groverext:2}
        \alpha_{k+1}^{\bm{-}} \leq \alpha_{k}^{\bm{-}} \cdot
        \frac{\sin(3\theta_k)}{\sin(\theta_k)} \cdot \frac{1}{\sqrt{2^{k+5}}} \le
        \alpha_{k}^{\bm{-}}\cdot 3 \cdot\frac{1}{\sqrt{2^{k+5}}}.
    \end{equation}

    Next, we proceed with a useful intermediate claim.

    \begin{claim}\label{claim:groverext}
        If $\Gamma_k^{\bm{+}} \neq \emptyset$ and $m= \lceil\log_9 n\rceil$, then
        \[ \sum_{k=1}^m \left({\left(\alpha_k^{\bm{+}}\right)}^2 +
        {\left(\alpha_k^{\bm{\sim}}\right)}^2\right) > {1}/{40}.\]
    \end{claim}
    \begin{claimproof}
        For a proof by contradiction, suppose that $\sum_{k=1}^{m} ((\alpha_k^{\bm{+}})^2
        + (\alpha_k^{\bm{\sim}})^2) \le {1}/{40}$.
        We start by inductively proving that $(\alpha_{k}^{\bm{-}})^2\le {1}/{5}\cdot
        ({3}/{8})^{k-1}$ and $\theta_k < {\pi}/{3}$ hold for $k\in \fragment{1}{m}$.

        In the base case of $k=1$, we have
        \[
            \left(\alpha_{1}^{\bm{-}}\right)^2
            =\frac{1}{n}\sum_{i\in \Gamma^{\bm{-}}}p_i
            \le \frac{1}{n}\sum_{i\in \Gamma^{\bm{-}}}\frac{1}{5}
            \le \frac{1}{n}\sum_{i\in
            I}\frac{1}{5}=\frac{1}{5}.
        \]
        For each $k\in \fragment{1}{m}$, we further have
        \[
            \sin^2(\theta_k)
            = ({\alpha_k^{\bm{+}}})^2+ ({\alpha_k^{\bm{\sim}}})^2 + ({\alpha_k^{\bm{-}}})^2
            \le \frac{1}{40} + \frac{1}{5}\cdot \left(\frac{3}{8}\right)^{k-1}
            \le \frac{1}{40}+\frac{1}{5}
            =\frac{9}{40} < \frac{3}{4},
        \]
        and thus $\theta_k < {\pi}/{3}$.

        Consequently, we have
        \[
            \alpha_{k+1}^{\bm{-}}
            \leq \alpha_{k}^{\bm{-}} \cdot 3\cdot \frac{1}{\sqrt{2^{k+5}}}
            \le \frac{1}{5}\cdot \left(\frac{3}{8}\right)^{k-1}\cdot 3 \cdot \frac{1}{\sqrt{2^6}}
            = \frac{1}{5}\cdot \left(\frac{3}{8}\right)^{k},
        \]
        as claimed for $k\in \fragmentco{1}{m}$.

        Therefore,
        \[
            \sum_{k=1}^{m} \sin^2(\theta_k)
            \leq \sum_{k=1}^{m}  \left(({\alpha_k^{\bm{+}}})^2+ ({\alpha_k^{\bm{\sim}}})^2 + ({\alpha_k^{\bm{-}}})^2\right)
            \leq \frac{1}{40} + \sum_{k=1}^{m} \frac{1}{5}\cdot \left(\frac{9}{64}\right)^{k-1}
            \leq \frac{1}{40}+\frac{64}{275}=\frac{567}{2200}.
        \]
        By repeatedly applying \eqref{eq:groverext:1}, we obtain
        \begin{align*}
            \alpha_m^{\bm{+}}
            &\geq \alpha_1^{\bm{+}} 3^{m-1}\prod_{k=1}^{m-1}
            \left(1-\frac{4}{3}\sin^2(\theta_k)\right) \sqrt {1 - \frac{1}{2^{k+5}}}\\
            &\geq \alpha_1^{\bm{+}} 3^{m-1} \left(1 - \frac43\sum_{k=1}^m\sin^2(\theta_k)-\frac{1}{2^{k+5}}\right)
            \ge \alpha_1^{\bm{+}} 3^{m-1}\left(1-\frac{189}{550}-\frac{1}{32}\right)\\
            &> \alpha_1^{\bm{+}} \frac{3^{m-1}}{2}.
        \end{align*}
        In particular, if $\Gamma_k^{\bm{+}} \neq \emptyset$,
        then $(\alpha_1^{\bm{+}})^2 \ge {9}/{10n}$.
        Hence,
        \[  \left(\alpha_m^{\bm{+}}\right)^2
            > (\alpha_1^{\bm{+}})^2 \cdot \frac{9^{m-1}}{4}
            \ge \frac{9^m}{40n}
            \ge \frac{1}{40}.
        \]
        This contradicts the assumption that $\sum_{k=1}^{m} ((\alpha_k^{\bm{+}})^2 +
        (\alpha_k^{\bm{\sim}})^2) \le {1}/{40}$.
    \end{claimproof}
    We now proceed to construct the final algorithm using the $A_k$ as follows
    (this is the only part where the actual algorithm differs from that of~\cite{groverext}).

    \begin{enumerate}
        \item For \(k = 1\) to \(\ceil{\log_9(n)}\) the algorithm does the following.
            \begin{itemize}
                \item Run the operator \(A_k\) independently \(\lceil80\ln \delta^{-1}\rceil\) times and measure each time the first register.
                \item For each measured \(\ket{i}\):
                    \begin{itemize}
                        \item Run \(U_f\) independently \(t_2 = \Theta(\log (n/\delta))\) times.
                        \item If we obtain $f(i)=1$ at least \(0.15 \cdot t_2\) times out of the \(t_2\) runs of \(U_f\), then report \(i\).
                    \end{itemize}
            \end{itemize}
        \item Return $\bot$.
    \end{enumerate}

    Regarding the correctness proof, suppose the hidden constant in \(t_2\) is large enough.

    First, assume \(I^{\bm{+}} \neq \emptyset\).
    A single measurement at step $k$ results in $\ket{i}$ for $i\in \Gamma^{\bm{-}}$ with
    probability at most
    \[1-(\alpha_{k}^{\bm{+}})^2 - (\alpha_{k}^{\bm{\sim}})^2 \le
    \exp(-(\alpha_{k}^{\bm{+}})^2-(\alpha_{k}^{\bm{\sim}})^2).\]

    Since all measurements are independent and repeated $\lceil80\ln \delta^{-1}\rceil$
    times, the probability that we
    only measure $\ket{i}$ for $i\in \Gamma^{\bm{-}}$ does not exceed
    \[\exp\left(-80\ln \delta^{-1} \sum_{k=1}^{m}((\alpha_{k}^{\bm{+}})^2+(\alpha_{k}^{\bm{\sim}})^2) \right)\le
    \exp(-2\ln\delta^{-1})=o(\delta),\]
    where the inequality follows from \cref{claim:groverext}.

    By standard concentration bounds, whenever $\ket{i}$ for $i\in \Gamma^{\bm{+}} \cup
    \Gamma^{\bm{\sim}}$ is measured,
    we output $i$ with probability at least $1-(\delta/n)^2=1-o(\delta)$ because $p_i \ge
    0.2 > 0.15$.
    Thus, if $I^{\bm{+}}\ne \emptyset$, the algorithm outputs $\bot$ with probability
    $o(\delta)$.

    As far as the correctness is concerned, it remains to prove that the algorithm is unlikely
    to report an element of $I^{\bm{-}}$.
    If $\ket{i}$ is measured for some $i\in I^{\bm{-}}$, then $p_i \leq 0.1 < 0.15$, so $i$ is
    rejected during the verification with probability at least $1-(\delta/n)^2$.

    Overall, the algorithm verifies the outcomes of $\Oh(\log n \cdot \log \delta^{-1})$
    measurements.

    The probability that any element of $I^{\bm{-}}$ passes this verification is
    $\Oh((\delta/n)^2\cdot \log n \cdot \log \delta^{-1}) = o(\delta)$.
    Thus, the algorithm outputs an element of $I^{\bm{-}}$ with probability $o(\delta)$.

    Lastly, observe that, similar to~\cite{groverext}, the time and query complexity of \(A_k\) is \(\Oh(3^k)\).
    Therefore, the overall complexity becomes
    \begin{align*}
        \sum_{k=1}^{\ceil{\log_9(n)}} \Oh\left(\log\delta^{-1}\cdot
        \left(3^k+\log{\delta^{-1}n}\right)\right)
        &=
        \Oh\left(3^{\log_9(n)}\log{\delta^{-1}}+\log n\cdot \log{\delta^{-1}n}\cdot
        \log{\delta^{-1}}\right)\\
        &=\Oh(\sqrt{n}\log{\delta^{-1}}+\log^3
        {\delta^{-1}n}).
    \end{align*}
    Since $\log^3 ({n}/{\delta})\le 8\cdot (\log^3 n + \log^3 \delta^{-1})$, the
    complexity simplifies to $\Oh(\sqrt{n}\log\delta^{-1}+\log^3 \delta^{-1})$.
\end{proof}

\begin{corollary}\label{cor:groverext}
    At the cost of increasing the query complexity to
    $\Oh(\sqrt{n}\log\delta^{-1}+\log^4\delta^{-1})$ (and keeping the time complexity at
    \(\Ohtilde(\sqrt{n})\)), the algorithm of
    \cref{lem:groverext} can be extended so that the reported position $i\in
    I^{\bm{+}}\cup I^{\bm{\sim}}$ satisfies $i \le \min I^{\bm{+}}$ or, alternatively, $i
    \ge \max I^{\bm{+}}$.
\end{corollary}
\begin{proof}
    We focus on the case when $i \le \min I^{\bm{+}}$ is desired; the alternative version
    is symmetric.

    For an interval $J\subseteq I$, we denote by $\mathsf{ALG}(J,\delta)$ an invocation of
    the algorithm of \cref{lem:groverext}  with the error probability at most $\delta$ and
    the domain restricted to $J$.

    Our algorithm proceeds as follows:
    \begin{enumerate}
        \item $i\coloneqq \mathsf{ALG}(I,\delta/2)$.
        \item If $i= \bot$, return $\bot$.
        \item Otherwise, set $J\coloneqq I$ and perform the following steps for
            $k\coloneqq 1$ to $\ceil{\log n}$:
            \begin{enumerate}
                \item Partition $J=J_L\cup J_R$ into two halves of the same size (up to
                    $\pm 1$) with $\max J_L < \min J_R$.
                \item If $i\in J_L$, set $J\coloneqq J_L$.
                \item Otherwise:
                    \begin{itemize}
                        \item $i'\coloneqq \mathsf{ALG}(J_L,\delta/2^{k+1})$.
                        \item If $i'=\bot$, set $J\coloneqq J_R$.
                        \item Otherwise, set $i\coloneqq i'$ and $J\coloneqq J_L$.
                    \end{itemize}
            \end{enumerate}
        \item Return $i$.
    \end{enumerate}

    To analyze the algorithm, let us first observe that error bounds in all applications
    of $\mathsf{ALG}$ sum up to at most $\delta$, so we may assume without loss of
    generality that these applications all return correct answers.

    If the initial application returns $\bot$, then $I^{\bm{+}}=\emptyset$ and the
    algorithm correctly returns $\bot$.

    Otherwise, we maintain an interval $J\subseteq I$ so that $|J|\le \ceil{n/{2^{k}}}$
    holds by the end the $k$-th iteration and $J$ satisfies the following two invariants:
    \[i\in J\cap (I^{\bm^{+}}\cup I^{\bm{\sim}})\qquad\text{and}\qquad \min J \le \min I^{\bm^{+}}.\]

    Initially, the invariant is trivially satisfied for $J\coloneqq I$ and $k\coloneqq 0$.
    In the $k$-th iteration, we partition $J=J_L\cup J_R$ into two halves, both of size at
    most $\ceil{n/{2^{k}}}$.

    If $i\in J_L$, then invariant remains satisfied for $J\coloneqq J_L$.
    Otherwise, we use $\mathsf{ALG}$ for $J_L$.

    If the algorithm reports $\bot$, then $J_L\cap I^{\bm^{+}}=\emptyset$, and thus $\min
    J_R \le \min I^{\bm^{+}}$, which means that the invariant remains satisfied for
    $J\coloneqq J_R$.

    Otherwise, the algorithm reports $i'\in J_L\cap (I^{\bm^{+}}\cup I^{\bm{\sim}})$.
    We then set $i\coloneqq i'$ so that the invariant remains satisfied for $J\coloneqq
    J_L$.
    After $\ceil{\log n}$ iterations, we arrive at a singleton interval $J=\{i\}$,
    concluding that $i\in I^{\bm^{+}}\cup I^{\bm{\sim}}$ and $i \le \min I^{\bm^{+}}$
    satisfies the requirements.

    As we have only up to $\Oh(\log n)$ control instructions, the complexity of the algorithm is dominated
    by
    the calls to $\mathsf{ALG}$.
    The $k$-th of these calls has complexity
    \[\Oh(\sqrt{n/2^k}\log(2^k/\delta)+\log^3(2^k/\delta)).\]
    In total, we thus have a complexity of $\Oh(\sqrt{n}\log\delta^{-1}+\log^4(n/\delta))$. Since $\log^4
    (n/\delta) \le 16\cdot (\log^4 n + \log^4\delta^{-1})$, this bound reduces to
    $\Oh(\sqrt{n}\log\delta^{-1}+\log^4\delta^{-1})$.
\end{proof}

\paragraph*{Quantum Subroutines on Strings}

In this paper, we utilize Hariharan and Vinay's quantum algorithm for finding exact
occurrences of a pattern in a text as a subroutine.

\begin{theoremq}[Quantum Exact Pattern Matching \cite{HV03}] \label{thm:quantum_matching}
    There is an $\Ohtilde(\sqrt{n})$-time quantum algorithm that, given a pattern $P$ of
    length $m$ and a text $T$ of length~$n$ with $n\ge m$, finds an occurrence of $P$ in
    $T$ (or reports that \(P\) does not occur in \(T\)).
\end{theoremq}

One first application of the algorithm of \cite{HV03} is to find the cyclic rotation of a
primitive string.

\begin{lemma}\label{claim:hd_compstructure_rot}
    Given a string \( X \) and a primitive string \( Q \) of the same length $n$, we can
    determine in \(\Ohtilde(\sqrt{n})\) quantum time whether there exists an index \( a
    \in
    \fragmentco{0}{n} \) such that \(\rot^a(Q) = X\). If such an index exists, we
    can find the unique \( a \) that satisfies this condition.
\end{lemma}
\begin{proof}
    It suffices to use \cref{thm:quantum_matching} to check whether \( Q \) occurs in \(
    XX\fragmentco{0}{n-1} \).
    If an occurrence is found, return its starting position.
    Since \( Q \) is primitive, the occurrence must be unique.
\end{proof}

Next, we modify the algorithm of \cite{HV03} such that it outputs the whole set $\Occ(P,
T)$.

\begin{lemma}\label{thm:quantum_matching_all}
    There is an $\Ohtilde(\sqrt{\mathsf{occ} \cdot n})$-time quantum algorithm that, given
    a pattern $P$ of length $m$ and a text $T$ of length $n$ with $n\ge m$, outputs
    $\Occ(P, T)$, where $\mathsf{occ} = |\Occ(P, T)|$.
\end{lemma}
\begin{proof}
    Let $t_1 < t_2 < \cdots < t_{\mathsf{occ}}$ denote the positions that are contained in
    $\Occ(P, T)$.

    Suppose that we have just determined $t_i$  for some $0 \leq i < \mathsf{occ}$ (if $i
    = 0$, then $t_0$ is the starting position of the string).
    To find $t_{i+1}$ we use \cref{thm:quantum_matching} combined with an exponential
    search.
    More specifically, at the $j$-th jump of the first stage of the exponential search, we
    use \cref{thm:quantum_matching} to check whether
    \[\Occ(P, T\fragmentco{t_{i}+1}{\min(\sigma_i + 2^{j - 1},|T|)}) \neq \emptyset.\]
    Once we have found the first $j$, this set is non-empty, we use a binary search
    combined with  \cref{thm:quantum_matching} to find the exact position of $t_{i+1}$.

    By doing so, we use at most $\Oh(\log (t_{i+1} - t_{i})) \leq \Oh(\log n)$ times the
    algorithm from \cref{thm:quantum_matching} as subroutine, each of them requiring at
    most $\Ohtilde(\sqrt{t_{i+1} - t_i})$ time.
    Consequently, finding all $t_1, \ldots, t_{\mathsf{occ}}$ means spending
    $\Ohtilde(\sum_{i=0}^{\mathsf{occ}-1} \sqrt{t_{i+1} - t_i})$ time.

    This expression is maximized if the positions in $\Occ(P, T)$ are equally spaced from
    each other, meaning that $t_{i+1} - t_{i} = \Oh(n/\mathsf{occ})$ for all $0 \leq i <
    \mathsf{occ}$.
    We conclude that $\Ohtilde(\sum_{i=0}^{\mathsf{occ}-1} \sqrt{t_{i+1} - t_i}) \leq
    \Ohtilde(\mathsf{occ} \cdot \sqrt{n/\mathsf{occ}}) = \Ohtilde(\sqrt{\mathsf{occ} \cdot
    n})$.
\end{proof}

Using a similar proof, it is possible to devise the following quantum subroutine for
verifying whether the Hamming distance between two strings is less than or equal to a
threshold $k$ and, if this condition is met, to report the positions where the mismatches
occur.

\begin{lemma}\label{lem:find_mismatches}
    There is an $\Ohtilde(\sqrt{kn})$-time quantum algorithm that, given strings $X$ and
    $Y$ of length $n$, and a threshold $k$, verifies whether $\hd(X, Y) \leq k$.
    If this condition is met, the algorithm returns the set $\Mis(X, Y) \coloneqq \{i \in
    \fragmentco{0}{n} : X\position{i} \neq Y\position{i}\}$.
\end{lemma}
\begin{proof}
    Let $j_1 < j_2 < \cdots < j_{\ell}$ represent the first $\ell = \min(k+1, \hd(X, Y))$
    positions contained in $\Mis(X, Y)$.

    Similar to the approach in \cref{thm:quantum_matching_all}, we employ an exponential
    search technique to find $j_{i+1}$ given $j_{i}$.
    However, in this case, the search algorithm, instead of being intertwined with
    \cref{thm:quantum_matching}, utilizes a straightforward \GS to check
    whether there is a mismatch between two substrings of $X$ and $Y$.
    If $\ell = k + 1$, then the algorithm reports that $\hd(X, Y) > k$; otherwise, it
    returns the set $\{j_{i} : i \in \fragment{1}{\ell}\}$.

    The time analysis for this algorithm closely resembles the one presented in
    \cref{thm:quantum_matching_all}.
\end{proof}

If, instead of verifying whether $\hd(X, Y) \leq k$, it is sufficient to distinguish
between the two cases, $\hd(X, Y) \leq k$ and $\hd(X, Y) > 2k$, a speedup by a factor of
$k$ can be achieved.

\begin{lemma}\label{lem:quantum_gap_hd}
    There is a quantum algorithm $\mathcal{A}(X,Y,k)$ that, given two strings $X$ and $Y$
    of equal length $|X| = |Y| = m$ and a threshold $0 < k \leq m$,
    outputs \yes or \no so that:
    \begin{itemize}
        \item If $\hd(X,Y)\le k$, then $\pr{\mathcal{A}(X,Y,k)=\text{\yes}} \ge 9/10$.
        \item If $\hd(X,Y) > 2k$, then $\pr{\mathcal{A}(X,Y,k)=\text{\yes}} \le 1/10$.
    \end{itemize}
   The subroutine achieves this within $\Ohtilde(\sqrt{m/k})$ quantum time and queries.
\end{lemma}

\begin{proof}
    Consider a function $f : \fragmentco{0}{m}\to \{0,1\}$ such that $f(i)=1$ if and only
    if $X\position{i} \ne Y\position{i}$.
    Observe that $h\coloneqq\hd(X,Y)=|f^{-1}(1)|$.

    We apply the (approximate) counting algorithm of \cite[Theorem
    13]{brassard2002quantum} with parameters $(M,k)=(\big\lceil 48 \pi
    \sqrt{m/k}\big\rceil,6)$.
    Said algorithm uses $\Ohtilde(M)=\Ohtilde(\sqrt{m/k})$ quantum time and queries
    (evaluations of $f$) and, with probability at least $1- {1}/{(2(k-1))}= {9}/{10}$
    outputs a value $h'$ such that
    \[|h'-h|\le 12\pi ({\sqrt{h(m-h)}}) / {M}+36\pi^2 {m}/{M^2} \le 1/4\cdot\sqrt{hk}+
    1/{64} \cdot k.\]
    If $h \le k$, then
    \[h' \le h + 1/4\cdot \sqrt{hk} + {1}/{64} \cdot k \le \left(1+1/4+1/{64}\right)k \le
    3/2 \cdot k.\]
    If $h > 2k$, then
    \[h' \ge h - 1/4\cdot \sqrt{hk} - {1}/{64} \cdot k > \left(\sqrt{2}-
    1/4\right)\sqrt{hk}-{1}/{64}\cdot k > \left(2-{\sqrt{2}}/{4}-1/{64}\right)k > 3/2\cdot k.\]
    Thus, it suffices to return \yes if the computed estimate $h'$ is at most $3/2 \cdot k$ and \no otherwise.
\end{proof}

We conclude this (sub)section by presenting an edit counterpart to
\cref{lem:find_mismatches}.

\begin{lemma}[{\cite[Corollary 5.8]{KNW24}} adapted from {\cite[Theorem 1.1]{GJKT24}}]\label{prp:quantumed_w_info}
    There is a quantum algorithm that, given quantum oracle access to strings $X,Y$ of
    length at most $n$,
    computes their edit distance $k \coloneqq \ed(X,Y)$ and the edit information $\sE_{X,
    Y}(\mA)$ of some optimal alignment $\mA : X \onto Y$.
    The algorithm requires $\Ohtilde(\sqrt{n + kn})$ queries and $\Ohtilde(\sqrt{n + kn} +
    k^2)$ time. \lipicsEnd
\end{lemma}

\subsection{The \modelname Model}\label{sec:pillar}

Introduced in \cite{CKW20}, the \modelname model provides a useful abstraction of common
basic building blocks of string matching algorithms and makes them independent of the
concrete
representation of the input (strings).
In particular, the \modelname comprises the following basic operations.

\begin{itemize}
    \item ${\tt Extract}(S, \ell, r)$: Compute and return $S\fragment{\ell}{r}$.
    \item $\lceOp{S}{T}$: Compute the length of~the longest common prefix of~$S$ and $T$.
    \item $\lcbOp{S}{T}$: Compute the length of~the longest common suffix of~$S$ and $T$.
    \item $\ipmOp{S}{T}$: Assuming that $|T|\le 2|S|$, compute the starting positions of
        all exact occurrences of~$S$ in~$T$.
    \item $\accOpName(S,i)$: Retrieve the character $\accOp{S}{i}$.
    \item $\lenOpName(S)$: Compute the length $|S|$ of~the string $S$.
\end{itemize}

Now, a \modelname algorithm is an algorithm that interacts with the input only via the
above operations; the complexity of a \modelname algorithm is measured in the number of
calls to the basic operations it makes.

A \modelname algorithm works on any input
representation that \emph{implements} the basic operations.
Classically, this includes the ``standard setting'', where the input is given as ordinary
strings, as well as the fully-compressed setting, where the input is given as a
(run-length) straight-line program~\cite{CKW20}.

We use the known \modelname algorithms for Approximate Pattern Matching from
\cite{CKW20,CKW22}; which we list here for convenience.

\hdalg
\edalg

%% file: s3_gs_encoding.tex
\section{Step 3: Solving Systems of Substring Equations Classically and Efficiently}
\label{sec:encoding}

We discuss a classical algorithm to compute a universal solution to a
system of substring equations.

\nineeighteenone

As shown in~\cite{GKRR20}, a universal solution of a system of $b$ substring equations on
length-$n$ strings can be constructed in $\Oh(n+b)$ time.
Our goal is to achieve sublinear dependency on $n$, which requires a compressed
representation of the constructed universal solution.
To this end, we introduce the following variant of straight-line programs.

\xSLP

The main goal of this section is to prove \cref{prp:solve_substring_equations}.

\solvesubstringequations*

As outlined in the Technical Overview, we prove \cref{prp:solve_substring_equations} by
giving an efficient solution to \pn{Dynamic xSLP Substring Equation}.

\problembox{%
    \pn{Dynamic xSLP Substring Equation}\\
    {\bf{Maintained Object:}}
    a set \(E\) of substring equations on length-\(n\) strings
    and a pure xSLP $\G$ that generates a string $S$ of fixed length $n$ that is a
    universal solution to \(E\).\\
    {{\bf{Updates}} that modify \(\G\)}:
    \begin{itemize}
        \item \textsc{Init}$(n)$: Given \( n \in \mathbb{Z}_+ \), set \(E \coloneqq
            \emptyset\) and initialize the xSLP \(\G\) such
            that it contains a single pseudo-terminal of length $n$ (thereby being a
            universal solution to \(E = \emptyset\)).\\
            We assume that this operation is always called exactly once before any other
            operations are called.
        \item \textsc{SetSubstringsEqual}$(x, x', y, y')$:
            Given \(x, x', y, y' \in \fragment{0}{n}\) such that \(x < x'\), \(y < y'\) and \(x' - x = y' - y\),
            add to \(E\) the substring equation \(e: T\fragmentco{x}{x'} =
            T\fragmentco{y}{y'}\), where \(T\) is a formal variable representing a
            length-\(n\) string.
            Further, update \(\G\) such that the represented string $S$ is a
            universal solution to \(E \cup \{e\}\).
    \end{itemize}
    {{\bf{Queries}} that do not modify \(\G\):}
    \begin{itemize}
        \item \textsc{Export}: Return \(\G\).
    \end{itemize}
}%

\dynslpmain*

We proceed in two steps. First, we implement \textsc{SetSubstringsEqual} using a set of
few, conceptually simpler subroutines; second, we then show how to efficiently implement
said subroutines.

\subsection{Dynamic xSLP Substring Equation via Dynamic xSLP Split Substitute}

For the rest of this section, it is useful to define some extra notation.
For a pure xSLP \(\G\) we write
$\G_\mP$ for the SLP that is obtained from $\G$ by interpreting each
pseudo-terminal as a terminal.
Further, we often associate the string $S=\val(\G)$
with the underlying string of pseudo-terminals $S_\mP=\val(\G_\mP)$.
Observe that $S$ is obtained by expanding $S_\mP$, that is, we have
\[S=\bigodot_{i=0}^{|S_\mP|-1} \val_\G(S_\mP\position{i}).\]
Finally, we call $x\in \fragment{0}{|S|}$ a \emph{delimiter} if $S\fragmentco{0}{x}$ is obtained by
expanding a prefix of $S_{\mP}$.
In particular, $0$ and $|S|$ are always delimiters.
Consult \cref{fig:length_enhanced_string} for a visualization of an example.

Next, we introduce an intermediate data structure (problem) that allows us to implement
\pn{Dynamic xSLP Substring Equation}.

\problembox{%
    \pn{Dynamic xSLP Split Substitute}\\
    {\bf{Maintained Object:}}
    a pure xSLP $\G$ that generates a string $S$ of fixed length $n$ \\
    {{\bf{Updates}} that modify \(\G\)}:
    \begin{itemize}
        \item \textsc{Init}$(n)$: Given \( n \in \mathbb{Z}_+ \), initialize the xSLP so
            that it contains a single pseudo-terminal of length $n$.\\
            We assume that this operation is always called exactly once before any other
            operations are called.
        \item \textsc{Split}$(X, y)$: Given \( X\in \mP_\G\) and \( y \in
            \fragmentoo{0}{|X|} \), create two new pseudo-terminals \(Y\) of length
            $|Y|=y$ and \(Z\) of length $|Z|=|X|-y$, remove $X$ from $\mP_\G$, and update
            $\G$ so that every occurrence of $\val_\G(X)$ in $S$ is replaced with
            $\val_\G(Y)\cdot \val_\G(Z)$.\\
            We also allow \textsc{Split}$(X, y)$ for $y\in \{0,|X|\}$ with the convention
            that it does not do anything.
        \item \textsc{Split}$(x)$:
            Given \( x \in \fragment{0}{n} \),
            a shorthand for\\
            {$X \leftarrow $ \textsc{FindPseudoTerminal}$(x)$; $x' \leftarrow
            \text{\textsc{FindDelimiter}}(x)$;}
            \textsc{Split}$(X, x - x')$.%
            \footnote{
                From the parameters it is always clear which variant of \textsc{Split} we
                call.
            }
        \item \textsc{Substitute}$(X, x, y)$: Given \( X \in \mP_\G\) and two delimiters
            \( x, y \in \fragment{0}{n} \) such that \( e \coloneqq |X|/(y-x)\in
            \mathbb{Z}_+ \), remove $X$ from $\mP_\G$ and update $\G$ so that every
            occurrence of $\val_\G(X)$ in $S$ is replaced with \( S\fragmentco{x}{y}^e \).
    \end{itemize}
    {{\bf{Queries}} that do not modify \(\G\):}
    \begin{itemize}
        \item \textsc{Export}: Return \(\G\).
        \item \textsc{FindDelimiter}$(x)$: Given $x\in \fragment{0}{n}$, return the
            largest delimiter in $\fragment{0}{x}$.
        \item \textsc{FindPseudoTerminal}$(x)$: Given $x\in \fragmentco{0}{n}$, return the
            pseudo-terminal associated with  $X\position{x}$.
        \item \textsc{LCE}$(x, y)$: Given $x,y\in \fragment{0}{n}$, return the length of
            the longest common prefix of $S\fragmentco{x}{n}$ and $S\fragmentco{y}{n}$.
    \end{itemize}
}%

\begin{figure}[t]
    \centering
    \input{tikz/length_enhanced_string.tex}
    \caption{
        A string \( S \) of length \(n = 13\) obtained from a by expanding a string $S_\mP=XYZXY$ of pseudo-terminals
        with $|X|=4$, $|Y|=2$, and $|Z|=1$.
        The delimiters of \( S \) are $0, 4, 6, 7, 11, 13$.}
    \label{fig:length_enhanced_string}
\end{figure}

\pagebreak

\begin{lemma}\label{tenseven}
    Given a data structure for \pn{Dynamic xSLP Split Substitute}, there is a data
    structure for \pn{Dynamic xSLP Substring Equation} where
    \begin{itemize}
        \item the call to \textsc{Init}\((n)\) of \pn{Dynamic xSLP Substring Equation} makes
            exactly one call to \textsc{Init} of \pn{Dynamic xSLP Substring Equation};
        \item after the \(b\)-th call to \textsc{SetSubstringEqual}, the data structure
            made in total \(\Oh(b \log n)\) calls to (either variant of) \textsc{Split},
            \textsc{Substitute} (and any of the query operations); and
        \item any call to \textsc{Extract} of \pn{Dynamic xSLP Substring Equation} makes
            exactly one call to \textsc{Extract} of \pn{Dynamic xSLP Substring Equation}.
    \end{itemize}
\end{lemma}

\begin{algorithm}[t]
    \Fn{\textup{\textsc{SetSubstringsEqual}$(x, x', y, y')$}}{
        \textsc{Split}$(x)$, \textsc{Split}$(x')$\;
        \textsc{Split}$(y)$, \textsc{Split}$(y')$\;
        \While{$\textsc{LCE}(x,y)<x'-x$}{ \label{ln:copy:while}
            $\hx \leftarrow x + \textsc{LCE}(x, y) $, $\hy \leftarrow y + \textsc{LCE}(x, y) $\;
            $X \leftarrow \text{\textsc{FindPseudoTerminal}}(\hx)$, $Y \leftarrow \text{\textsc{FindPseudoTerminal}}(\hy)$\;
            \If{$|Y| > |X|$} {
                Swap $(X,x,\hx,x')$ with $(Y,y,\hy,y')$\;

            }
            $\ell \leftarrow |X|$\;
            $r \leftarrow \text{\textsc{FindDelimiter}}(\hy + \ell)-\hy$\;
            $Y' \leftarrow \text{\textsc{FindPseudoTerminal}}(\hy + \ell-1)$\; \label{ln:copy:name_x}
            \If{$\ell>r$ \KwSty{and} $X = Y'$} {  \label{ln:copy:per1}
                $e \leftarrow \floor{\ell/r}$\;
                \textsc{Split}$(\hx + er)$\; \label{ln:copy:make_delim_per}
                \textsc{Substitute}$(\text{\textsc{FindPseudoTerminal}}(\hx), \hy, \hy+r)$\; \label{ln:copy:sub_promote_per1}
            } \Else{  \label{ln:copy:notper}
                \textsc{Split}$(\hx + r)$\; \label{ln:copy:make_delim}
                \textsc{Substitute}$(\text{\textsc{FindPseudoTerminal}}(\hx), \hy,
                \hy+r)$\; \label{ln:copy:sub_promote}
            }
        }
    }
    \caption{Implementing \pn{Dynamic xSLP Substring Equation} using \pn{Dynamic xSLP
    Split Substitute}; in particular the \textsc{SetSubstringEqual} procedure.
    The subroutine splits and substitutes characters of \( S \) to ensure that  \(
    S\fragmentco{x}{x'} = S\fragmentco{y}{y'} \); which in turn guarantees that \(S\) is a
    universal solution to the system of substring equations maintained.}
    \label{alg:copysubstring}
\end{algorithm}

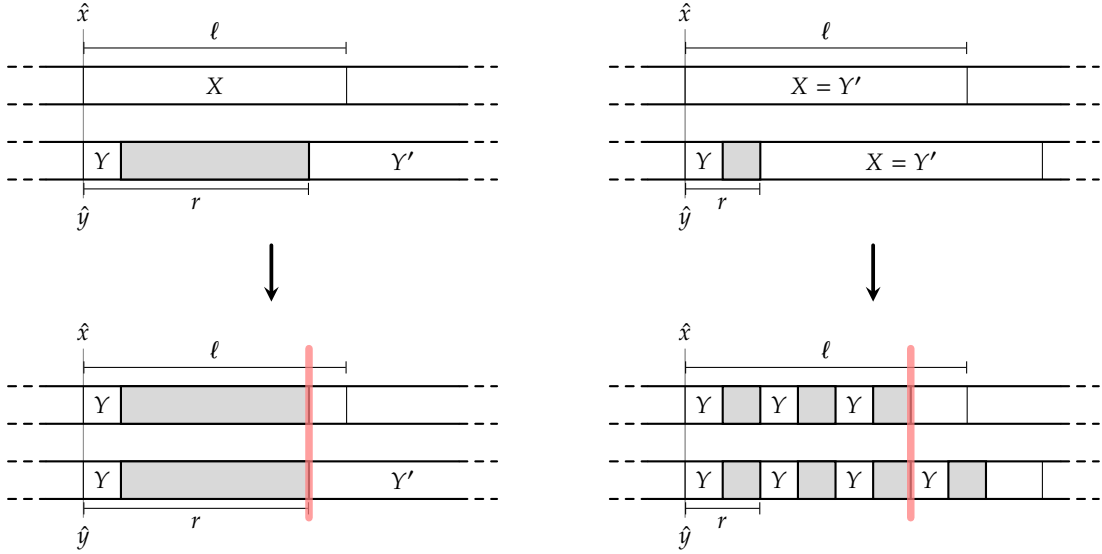
\begin{figure*}[t!]
    \centering
    \begin{tabularx}{\linewidth}{*{2}{>{\centering\arraybackslash}X}}
        \begin{subfigure}[t]{\linewidth}
            \centering
            \input{tikz/grammar_noper.tex}
            \caption{Pseudo-terminals before and after executing Line~\ref{ln:copy:make_delim}--\ref{ln:copy:sub_promote}.}\label{fig1a}
        \end{subfigure}
                &
        \begin{subfigure}[t]{\linewidth}
            \centering
            \input{tikz/grammar_per.tex}
            \caption{Pseudo-terminals before and after executing Line~\ref{ln:copy:make_delim_per}--\ref{ln:copy:sub_promote_per1}.}\label{fig1b}
        \end{subfigure}
    \end{tabularx}
    \caption{The two potential cases in an iteration  of \cref{alg:copysubstring} when
        pseudo-terminals of $S$ are split and substituted.
        The new starting points \(\hx\) and \(\hy\) (for the next iteration of the \KwSty{while} loop) are highlighted by a red line.
    }\label{fig:delim_and_promote}
\end{figure*}

\begin{proof}
    For \textsc{Init} and \textsc{Extract}, we directly forward the calls to the
    underlying \pn{Dynamic xSLP Split Substitute} data structure.

    For \textsc{SetSubstringsEqual}$(x, x', y, y')$, we proceed as follows.
    Given \(x, x', y, y' \in \fragment{0}{n}\) such that \(x < x'\), \(y < y'\) and
    \(x' - x = y' - y\),
    we first ensure that \(x, x', y, y'\) are delimiters.
    Then, as long as $S\fragmentco{x}{x'}\ne S\fragmentco{y}{y'}$, we identify indices
    \(\hx\) and \(\hy\) such that \(S\fragmentco{x}{\hx}=S\fragmentco{y}{\hy}\) and
    \(S\position{\hx}\ne S\position{\hy}\).
    In particular, we have $S\position{\hx}=\#_{0}^X\ne \#_0^{Y}=S\position{\hy}$ for some
    pseudo-terminals $X\ne Y$.
    We select the longer of \(X\) and \(Y\), without loss of generality \(X\), identify
    the maximum $r\in \fragment{0}{|X|}$ such that $\hy+r$ is a delimiter, and replace
    every occurrence of $\val_\G(X)$ with \(S\fragmentco{\hy}{\hy+r}\cdot \val_\G(X')\),
    where $X'$ is a new pseudo-terminal $X'$ of length $|X'|=|X|-r$ (omitted if $r=|X|$).

    If $|X|>r$ and $S\fragmentco{\hy+r}{\hy+r+|X|}=\val_\G(X)$, then, for efficiency
    reasons, we simulate $e\coloneqq\floor{|X|/r}$ aforementioned iterations with a single
    batched iteration that replaces every occurrence of $\val_\G(X)$ with
    \((S\fragmentco{\hy}{\hy+r})^e\cdot \val_\G(X')\), where $X'$ is a new pseudo-terminal
    $X'$ of length $|X'|=|X|-er$ (omitted if $er=|X|$).

    Consult \cref{alg:copysubstring} for detailed pseudo code and
    \cref{fig:delim_and_promote} for an illustration of the key part of the algorithm.

    We proceed to argue about the correctness of our algorithm.
    To that end,
    for a string $S$ of length $n$, let $P(S)$ denote the partition of $\fragmentco{0}{n}$
    according to where the same characters appear (so that $x,y\in \fragmentco{0}{n}$
    belong to the same class if and only if $S\position{x}=S\position{y}$).

    Observe that after  \textsc{Init}, we have
    $P(\val(\G)) = \left\{ \left\{ i \right\} \mid i \in \fragmentco{0}{n} \right\}$.
    In the next claim, we argue how $P(\val(\G))$ behaves after calls to
    \textsc{SetSubstringsEqual}.

    \begin{claim}\label{lem:sse}
        After a call of $\textsc{SetSubstringsEqual}(x,x',y,y')$ of
        \cref{alg:copysubstring},
        the classes of $x+i$ and $y+i$ are merged for $i\in \fragmentco{0}{x'-x}$ in
        $P(\val(\G))$.
    \end{claim}
    \begin{claimproof}
        Write \(S \coloneqq \val(\G)\).
        First, observe that a $\textsc{Split}$ operation does not affect the partition
        $P(S)$.
        Moreover, if $S\fragmentco{p}{q}=\val_\G(A)$, then the only effect of
        $\textsc{Substitute}(A,x,y)$ is that, for each $i\in \fragmentco{0}{q-p}$, the
        class of $p+i$ is merged with the class of $x+(i \bmod (y-x))$.
        Since $\textsc{SetSubstringsEqual}$ alters the string $S$ through these two
        operations only, the only effect on $P(S)$ is that some equivalence classes are
        merged.

        The \KwSty{while} loop of \cref{ln:copy:while} terminates only when
        $S\fragmentco{x}{x'}=S\fragmentco{y}{y'}$; at that time, for each $i\in
        \fragmentco{0}{x'-x}$, the positions $x+i$ and $y+i$ must be in the same class of
        $P(S)$.
        Thus, it suffices to prove that no other classes are merged throughout the
        execution of $\textsc{SetSubstringsEqual}(x,x',y,y')$.
        To show this, we analyze the applications of $\textsc{Substitute}$ in
        \cref{ln:copy:sub_promote,ln:copy:sub_promote_per1}.

        The call in \cref{ln:copy:sub_promote} is given a pseudo-terminal $A$ with
        $\val_\G(A)=S\fragmentco{\hx}{\hx+r}$ and delimiters $\hy,\hy+r$ such that $x \le
        \hx < \hx+r \le x'$ and $\hx-x=\hy-y$, and it replaces every occurrence of
        $\val_\G(A)$ with $S\fragmentco{\hy}{\hy+r}$.
        As a consequence, for each $i\in \fragmentco{0}{r}$, the class of $\hx+i$
        (representing the occurrences of $\#_i^A$) is merged with the class of $\hy+i$.
        This merger satisfies the requirements.

        The call in \cref{ln:copy:sub_promote_per1} is given a pseudo-terminal $A$ with
        $\val_\G(A)=S\fragmentco{\hx}{\hx+er}=S\fragmentco{\hy+r}{\hy+(e+1)r}$ and
        delimiters $\hy,\hy+r$ such that $x \le \hx < \hx+er \le x'$ and $\hx-x=\hy-y$.
        It replaces every occurrence of $\val_\G(A)$ with $(S\fragmentco{\hx}{\hx+r})^e$.
        As a consequence, for each $i\in \fragmentco{0}{er}$, the class of $\hx+i$
        (representing the occurrences of $\#_i^A$) is merged with the class of $\hy + (i
        \bmod r)$.

        While these mergers are not exactly of the required form, we claim that if we
        execute them from left to right, then the merging the classes of $\hx+i$ and
        $\hy+(i\bmod r)$ is equivalent to merging the classes of $\hx+i$ and $\hy+i$
        because $\hy+(i\bmod r)$ and $\hy+i$ are already in the same class.
        This statement is trivial in the base case of $i\in \fragmentco{0}{r}$.
        Otherwise, $\hy+i$ and $\hx+i-r$ are already in the same class since initially
        $S\position{\hy+i}=S\position{\hx+i-r}=\#^A_{i-r}$.

        Moreover, the inductive assumption guarantees that $\hy+i-r$ and $\hy+(i\bmod r)$
        are already in the same class and that we have already merged the classes of
        $\hx+i-r$ and $\hy+i-r$.

        Hence, all the positions $\hy+i$, $\hx+i-r$, $\hy+i-r$, and $\hy+(i\bmod r)$ are
        all already in the same class, and thus merging the classes of $\hx+i$ and
        $\hy+(i\bmod r)$ is indeed equivalent to merging the classes of $\hx+i$ and
        $\hy+i$.
    \end{claimproof}

    We conclude that calls to \textsc{SetSubstringEqual} correctly maintain a universal
    solution to the corresponding system of substring equations.

    \begin{claim} \label{lem:macro_scheme_to_set_equal}
        Consider a system $E$ of substring equations on length-$n$ strings.
        Construct a dynamic xSLP by calling \textsc{Init}$(n)$ and
        $\textsc{SetSubstringsEqual}(x,x',y,y')$ for all equations $E\ni e :
        T\fragmentco{x}{x'}=T\fragmentco{y}{y'}$ in arbitrary order.
        Then, the string $S=\val(\G)$ is a universal solution of $E$.
    \end{claim}

    \begin{claimproof}
        By \cref{lem:sse}, the partition $P(S)$ is obtained from a partition of
        $\fragmentco{0}{n}$ into singletons (obtained via the call $\textsc{Init}(n)$) by
        merging the classes of $x+i$ and $y+i$ for every $E\ni e :
        T\fragmentco{x}{x'}=T\fragmentco{y}{y'}$ and $i\in \fragmentco{0}{x'-x}$.
        In particular, these mergers imply that $S\position{x+i}=S\position{y+i}$, and thus $S$ is a solution of $E$.
        Moreover, for every other solution $S'$, we already have
        $S'\position{x+i}=S'\position{y+i}$, so we never merge positions in different
        classes of $P(S')$.
        Thus, $P(S)$ is coarser than $P(S')$, and hence $S'$ can indeed be obtained from
        $S$ using a letter-to-letter morphism.
    \end{claimproof}

    Next, we turn to the running time of \textsc{SetSubstringsEqual} (in terms of calls to
    the operations of the \pn{Dynamic xSLP Split Substitute} data structure.
    To this end, we argue that each iteration of the \KwSty{while} loop decreases the
    following potential by at least one unit on average (recall that \( \mP_\G\) is the
    set of pseudo-terminals of \(\G\)):
    \[
        \Phi(\G) = \sum_{X\in \mP_\G} \Big( \ 1 + \log |X| \ \Big).
    \]
    In particular, suppose that the \KwSty{while} loop at \cref{ln:copy:while} makes $d$
    iterations in an execution of \textsc{SetSubstringsEqual} and define xSLPs $\G_0,
    \G_1, \ldots, \G_d$ so that $\G_{i-1}$ and $\G_{i}$ are the xSLPs before and after the
    $i$-th iteration, respectively.
    We intend to show that then, we have $\Phi(\G_{0}) - \Phi(\G_{d}) \ge d$.
    To that end, we rely on the following intermediate claim.

    \begin{claim}\label{claim:step}
        Let \( X_i \), \( Y_i \),
        and \( Y_i' \) denote the pseudo-terminals of \( \G_{i-1} \) indicated at
        \cref{ln:copy:name_x} in iteration \( i \in \fragment{1}{d} \) by the variables \(
        X \), \( Y \), and \( Y' \), respectively.

        For every iteration $i \in \fragment{1}{d}$, we have $\Phi'_{i-1} - \Phi'_i \ge 1$,
        where we define \[\Phi'_i = \begin{cases}
        \Phi(\G_i)-1/2\cdot\log|Y_{i+1}| & \text{if $i < d$,}\\
        \Phi(\G_i) & \text{if $i=d$.}
        \end{cases}\]
    \end{claim}
    \begin{claimproof}
        Let \( X_i' \) denote the pseudo-terminal corresponding to the right
        pseudo-terminal that $X_i$ is split into in the call to \textsc{Split} at
        \cref{ln:copy:make_delim_per}/\cref{ln:copy:make_delim}
        (the former call if the \KwSty{if} branch is taken at \cref{ln:copy:per1},
        and the latter if the \KwSty{else} branch is taken at \cref{ln:copy:notper}).
        If $X_i$ is not split in these calls to \textsc{Split}, we say \( X_i' \) is
        undefined.
        In particular, this must happen for $i=d$.

        First, suppose that \( X_i' \) is undefined.
        In that case, $\Phi(\G_{i-1})-\Phi(\G_i)=1+\log|X_i|$.
        Moreover, $\Phi'_i \le \Phi(\G_i)$ and
        $\Phi'_{i-1}=\Phi(\G_{i-1})-1/2\cdot\log|Y_i|$.
        Since $|X_i|\ge |Y_i|$, we conclude that
        \[\Phi'_{i-1} - \Phi'_i\ge \Phi(\G_{i-1})-1/2\cdot\log|Y_i|-\Phi(\G_i)=
        1+\log|X_i|-1/2 \; \log|Y_i|\ge 1+1/2\cdot\log|X_i|\ge 1.\]

        Next, suppose that $X'_i$ exists and $X_i$ is split using the \textsc{Split} call
        at \cref{ln:copy:make_delim}.
        In that case, $Y_{i+1} =X'_i$; see \cref{fig1a}.
        Consequently, $\Phi'_i = \Phi(\G_i)-1/2\cdot\log|X'_i|$ and
        $\Phi'_{i-1}=\Phi(\G_{i-1})-1/2\cdot\log|Y_i|$.
        Moreover, $\Phi(\G_{i-1})-\Phi(\G_i)=\log|X_i|-\log|X'_i|$, and thus,
        \begin{align*}
            \Phi'_{i-1}-\Phi'_{i}
            &=\Phi(\G_{i-1})-1/2\cdot\log|Y_i|-\Phi(\G_i)+1/2\cdot\log|X'_i|\\
            &=\log|X_i|-1/2\cdot\log|Y_i| - 1/2\cdot\log|X'_i| \\
            &=\log(|Y_i|+|X'_i|)-\log\sqrt{|Y_i||X'_i|} \\
            &\ge 1,
        \end{align*}
        where the last step is the inequality between the arithmetic and the geometric means.

        Finally, suppose that $X'_i$ exists and $X_i$ is split using the \textsc{Split} call at \cref{ln:copy:make_delim_per}.
        In that case, $Y_{i+1}\in \{X'_i,Y_i\}$; see \cref{fig1b}.

        If $Y_{i+1}=X'_i$, then $\Phi'_{i-1}-\Phi'_{i}\ge 1$ is proved exactly as in
        the previous case, except that we have $|X_i|\ge |Y_i|+|X'_i|$ rather than
        $|X_i|=|Y_i|+|X'_i|$.

        If $Y_{i+1}=Y_i$, on the other hand, we have
        $\Phi'_{i-1}=\Phi(\G_{i-1})-1/2\cdot\log|Y_i|$ and
        $\Phi'_{i}=\Phi(\G_{i})-1/2\cdot\log|Y_i|$.

        Consequently, we have
        \[\Phi'_{i-1} - \Phi'_i=\Phi(\G_{i-1})-\Phi(\G_i)=\log|X_i|-\log|X'_i|.\]
        As we still have $|X'_i| = |X_i|\bmod r < 1/2\cdot |X_i|$ (since $r\le |X_i|$),
        we obtain $\Phi'_{i-1} - \Phi'_i=\log|X_i|-\log|X'_i| > 1$; which completes
        the proof of the claim.
    \end{claimproof}
    Now, from \cref{claim:step}, we readily obtain
    \begin{equation}
        \Phi(\G_0)-\Phi(\G_d)=\Phi'_0+1/2\cdot\log|Y_1|-\Phi'_d \ge \Phi'_0-\Phi'_d \ge \sum_{i=1}^d ( \Phi'_{i-1}-\Phi'_i)\ge d.
        \label{lem:copy_iterations}
    \end{equation}
    Next, we provide bounds (in terms of \(b\)) on the total number of iterations of the
    \KwSty{while} loop at \cref{ln:copy:while} in \textsc{SetSubstringsEqual}.
    This in turn implies the claimed bounds on the number of calls to the \textsc{Dynamic
    xSLP Split Substitute} data structure.

    \begin{claim}\label{lem:iter_num}
        Suppose we call \textsc{Init}$(n)$ and then call \textsc{SetSubstringsEqual} $b$
        times.
        Then, the \KwSty{while} loop at \cref{ln:copy:while} in the implementation of
        \textsc{SetSubstringsEqual} is executed at most \(\Oh(b \log n)\) times in total.
    \end{claim}
    \begin{claimproof}
        Observe that the potential $\Phi$ increases by at most \( 4 + 4 \log n \) in the
        first two lines of \textsc{SetSubstringsEqual}, as each call to \textsc{Split}
        increases the potential by at most \( 1 + \log n \).
        Overall, across all $b$ calls, the potential increases by \(\Oh(b \log n)\) at
        these two lines.
        On the other hand, by \cref{lem:copy_iterations}, after \( d \) iterations of the
        \KwSty{while} loop, the potential decreases by at least \(d\).
        Since we start with potential \(1 + \log n\), the total number of iterations of
        the \KwSty{while} loop is at most \(\Oh(b \log n)\).
    \end{claimproof}
    In total, this completes the proof.
\end{proof}

\subsection{An Efficient Classical Data Structure for Dynamic xSLP Split Substitute}
\label{sec:encoding_impl}

Before we explain how to implement the \pn{Dynamic xSLP Split Substitute}, we show that the
updates have limited effect on the compressibility of the underlying string
$\val(\G_{\mP})$.
(Recall that $\G_{\mP}$ is obtained from $\G$ by interpreting each pseudo-terminal as a terminal.)

\begin{lemma} \label{lem:c_bound}
    If we construct an xSLP $\G$ using a call \textsc{Init}$(n)$ followed by a sequence of
    $m$ calls to $\textsc{Split}$ or $\textsc{Substitute}$, then $|\LZ(\val(\G_{\mP}))| =
    \Oh(m \log n)$.
\end{lemma}

\begin{proof}
    Kempa and Prezza \cite{KP18} defined an \emph{attractor} of a string \( X \) as a set
    \( \Gamma \subseteq \fragmentco{0}{|X|} \) of positions in \( X \) such that every
    substring \( X\fragmentco{i}{j} \) has an occurrence $X\fragmentco{\hi}{\hj}=X\fragmentco{i}{j}$
    satisfying \( \fragmentco{\hi}{\hj} \cap \Gamma \neq \emptyset \).
    They also showed that if $\Gamma$ is an attractor of $X$, then $|\LZ(X)|=\Oh(|\Gamma|\log |X|)$.

    Our goal is to prove that $\val(\G_{\mP})$ has an attractor of size $m+1$.
    For this, we maintain a set $D\subseteq \fragmentco{0}{n}$ of delimiters in
    $S=\val(\G)$ satisfying the following property:
    for each interval $\fragmentco{p}{q}\subseteq \fragmentco{0}{n}$ containing at least
    one delimiter,
    the substring $S\fragmentco{p}{q}$ has an occurrence $S\fragmentco{\hp}{\hq}$ such
    that $D\cap \fragmentco{\hp}{\hq}\ne \emptyset$.

    Initially, $D=\{0\}$ satisfies the invariant because $0$ is the only delimiter in
    $\fragmentco{0}{n}$.
    If $S'$ is obtained from $S$ using $\textsc{Split}(X, x)$ or $\textsc{Substitute}(X,
    x, y)$, we pick an arbitrary occurrence $S\fragmentco{r}{r+|X|}$ of $\val_\G(X)$ and
    set $D'=D\cup\{r+x\}$ (for \textsc{Split}) or $D'=D\cup\{r+y-x\}$ (for
    \textsc{Substitute}).

    To see that the invariant is satisfied, consider an arbitrary interval
    $\fragmentco{p}{q}\subseteq \fragmentco{0}{n}$ containing at least one delimiter of
    $S'$.

    If $\fragmentco{p}{q}$ contains at least one delimiter of $S$, then
    $S\fragmentco{p}{q}$ has an occurrence $S\fragmentco{\hp}{\hq}$ such that $D\cap
    \fragmentco{\hp}{\hq}\ne \emptyset$.
    The characterization of \cref{lem:sse} implies
    $S'\fragmentco{p}{q}=S'\fragmentco{\hp}{\hq}$, so $S'\fragmentco{\hp}{\hq}$ is an
    occurrence of $S'\fragmentco{p}{q}$ such that $D'\cap \fragmentco{\hp}{\hq}\ne
    \emptyset$.
    Thus, it remains to consider the case when $\fragmentco{p}{q}$ does not contain any delimiters of $S$.
    In particular, $S\fragmentco{p}{q}$ must be a substring of $\val_\G(X)$.

    In the case of $\textsc{Split}(X,x)$, we observe that $S'\fragmentco{p}{q}$ is a
    substring of $\val_\G(Y)\cdot \val_\G(Z)$, where $Y$ and $Z$ are the two new
    pseudo-terminals created, but not a substring of $\val_\G(Y)$ or $\val_\G(Z)$.
    Thus, $S'\fragmentco{p}{q}$ has an occurrence $S'\fragmentco{\hp}{\hq}$ such that
    $D'\cap \fragmentco{\hp}{\hq}=\{r+x\} \ne \emptyset$.

    In the case of $\textsc{Substitute}(X, x,y)$, we conclude that
    $S'\fragmentco{p}{q}$ is a substring of $S\fragmentco{x}{y}^e$.
    If $S'\fragmentco{p}{q}$ is a substring of $S\fragmentco{x}{y}$, then, by the
    inductive assumption,
    it has an occurrence $S\fragmentco{\hp}{\hq}=S'\fragmentco{\hp}{\hq}$ such that $D\cap
    \fragmentoo{\hp}{\hq}\ne \emptyset$.
    Otherwise, the leftmost occurrence $S'\fragmentco{\hp}{\hq}$ of $S'\fragmentco{p}{q}$
    within $S'\fragmentco{r}{r+|X|}$ satisfies $D'\cap \fragmentco{\hp}{\hq}=\{r+(y-x)\}
    \ne \emptyset$.

    Our attractor $\Gamma$ of $S_\mP=\val(\G_\mP)$ of size $m+1$ consists of the
    ($0$-based) ranks of the delimiters in $D$ among all the delimiters of $S$.
    Every fragment $S_\mP\fragmentco{i}{j}$ expands to some fragment $S\fragmentco{p}{q}$
    such that $p$ is a delimiter.
    By the invariant on $D$, we have a matching fragment $S\fragmentco{\hp}{\hq}$ such
    that $\fragmentco{\hp}{\hq}\cap D \ne \emptyset$.
    This yields a fragment $S_\mP\fragmentco{\hi}{\hj}$ matching $S_\mP\fragmentco{i}{j}$
    such that $\fragmentco{\hi}{\hj}\cap \Gamma \ne \emptyset$.
\end{proof}

\Cref{lem:c_bound} allows us to describe an efficient implementation of
\pn{Dynamic xSLP Split Substitute}.

\begin{lemma} \label{lem:split_substitute_impl}
    There is a data structure for \pn{Dynamic xSLP Split Substitute} so that, after initialization with
    $\textsc{Init}(n)$ followed by a sequence of $m$ updates ($\textsc{Split}$ or
    $\textsc{Substitute}$):
    \begin{itemize}
        \item the xSLP $\G$ is of size $\Oh(m \log^3 n)$;
        \item the queries \textsc{FindDelimiter}, \textsc{FindPseudoTerminal}, and
            \textsc{LCE} take time \(\Oh(\log n)\),
        \item the next update (\textsc{Split} or \textsc{Substitute}) takes time $\Oh(m\log^5 n)$.
    \end{itemize}
\end{lemma}
\begin{proof}
    Recall that AVL grammars~\cite{Ryt03,KK20,KL21} are dynamic SLPs that can be
    updated using the following operations:
    \begin{description}
        \item[Insert] Inserts a given character to the alphabet of non-terminals $\Sigma$.
            This operation takes $\Oh(1)$ time.
        \item[Concatenate] Given symbols $A,B\in \mS$, insert a symbol $C$ such that
            $\val(C)=\val(A)\cdot \val(B)$. This operation takes
            $\Oh(1+|\log(|\val(A)|/|\val(B)|)|)$ time and may insert some auxiliary
            symbols.
        \item[Extract] Given a symbol $A\in \mS$ and integers $0\le i < j \le |A|$, insert
            a symbol $B$ such that $\val(B)=\val(A)\fragmentco{i}{j}$. This operation
            takes $\Oh(1+\log |\val(A)|)$ and may insert some auxiliary symbols.
    \end{description}

    We maintain $\G_\mP$ as an AVL grammar of size $\Oh(m\log^3 n)$.
    The xSLP $\G$ is obtained from $\G_\mP$ by re-interpreting the non-terminals as
    pseudo-terminals.

    We implement queries based on \cref{lem:pillar_on_xslp}; note that accessing
    $S\position{x}=\#^A_{i}$ lets us recover both the associated pseudo-terminal $A$ and
    the maximum delimiter $x-i$ in $\fragment{0}{x}$.
    Updates the handled as follows.

    \begin{itemize}
        \item For \textsc{Init}, we initialize $\G_\mP$ with a single terminal interpreted
            in $\G$ as a pseudo-terminal of length $n$.
        \item For \textsc{Split}, we insert two terminals $Y,Z$ to $\G_\mP$, interpreted
            in $\G$ as pseudo-terminals of length $y$ and $|X|-y$, respectively, and
            replace the terminal $X$ with a non-terminal representing $YZ$.
            Note that $\G_\mP$ is no longer an AVL grammar after such a replacement.
        \item For \textsc{Substitute}, we identify a substring $S_{\mP}\fragmentco{i}{j}$
            that expands to $S\fragmentco{x}{y}$, use the \emph{extract} operation to
            extend $\G_\mP$ with a non-terminal representing $S_{\mP}\fragmentco{i}{j}$,
            and then use $\Oh(\log e)$ constant-time applications of the
            \emph{concatenate} operation to extend $\G_\mP$ with a non-terminal
            representing $(S_{\mP}\fragmentco{i}{j})^e$; for the latter, we create symbols
            representing powers of $S_{\mP}\fragmentco{i}{j}$ with exponents
            $\floor{e/2^k}$ and $\ceil{e/2^k}$ for $k\in \fragmentco{0}{\floor{\log e}}$.
            Finally, we replace $X$ with the symbol representing $(S\fragmentco{x}{y})^e$.
            Note that $\G_\mP$ is no longer an AVL grammar after such a replacement.
    \end{itemize}

    After each update, we recompute $\G_{\mP}$ to make sure that it is an AVL grammar of
    size $\Oh(z\log^2 n)$, where $z=|\LZ(\val(\G_\mP))|$; by \cref{lem:c_bound}, this is
    at most $\Oh(m\log^3 n)$.

    For this, we construct a run-length SLP of size $\Oh(z\log n)$ generating the same
    string $\hS$~\cite{tomohiro} and convert the run-length SLP back to an AVL grammar of
    size $\Oh(z\log^2 n)=\Oh(m\log^3 n)$; see~\cite{KK20}. The whole process takes
    $\Oh(m\log^3 n)$ time.

    Finally, we obtain $\G$ by re-interpreting each non-terminal of $\G_\mP$ as a
    pseudo-terminal and preprocess $\G$ using the construction algorithm of
    \cref{lem:pillar_on_xslp}, which takes $\Oh(m\log^5 n)$ time.
\end{proof}

\subsection{An Efficient Algorithm for Solving Systems of Substring Equations}

Finally, we combine the previous steps to obtain our claimed main results.

\dynslpmain
\begin{proof}
    Follows from \cref{tenseven,lem:split_substitute_impl}.
\end{proof}

\solvesubstringequations
\begin{proof}
    We initialize a Dynamic xSLP Representation with $\textsc{Init}(n)$ and call
    $\textsc{SetSubstringsEqual}$ for each equation $e\in E$.
    The claimed bounds then follow from \cref{prp:solve_substring_equations}.
\end{proof}

\subsection{Extension to Equations between Substrings of Multiple Strings}
While \cref{def:substring_equation,def:universal_solution,prp:solve_substring_equations}
apply to equations between substrings of a single string, it is convenient to generalize
these notions to allow equations between substrings of many strings.

\begin{definition}[Substring Equation; see \cref{def:substring_equation}]\label{def:substring_many}\label{def:universal_many}
    Consider formal variables $T_0,\ldots,T_{t-1}$ representing strings of lengths
    $n_0,\ldots,n_{t-1}$.

    A \emph{substring equation} on length-$(n_0,\ldots,n_{t-1})$ strings is a constraint
    of the form $e: T_i\fragmentco{x}{x'}=T_j\fragmentco{y}{y'}$, where $i,j\in
    \fragmentco{0}{t}$ and $0\le x \le x' \le n_i$ as well as $0\le y \le y' \le n_j$ are
    all integers.

    A \emph{system} of substring equations is a set of substring equations $E$ on
    length-$(n_0,\ldots,n_{t-1})$ strings.

    Strings $(S_0,\ldots,S_{t-1})\in \Sigma^{n_0}\times \cdots \times \Sigma^{n_{t-1}}$
    \emph{satisfy} the system $E$ if and only if, for every equation $E\ni
    e:T_i\fragmentco{x}{x'} = T_j\fragmentco{y}{y'}$, the fragments
    $S_i\fragmentco{x}{x'}$ and $S_j\fragmentco{x'}{y'}$ match, that is, they are occurrences
    of the same substring.

    Strings $(S_0,\ldots,S_{t-1})\in \Sigma^{n_0}\times \cdots \times \Sigma^{n_{t-1}}$
    form a \emph{universal solution} of a system $E$ of substring equations on
    length-$(n_0,\ldots,n_{t-1})$ strings if they satisfy $E$ and every other tuple
    $(\hS_0,\ldots,\hS_{t-1})\in \hat{\Sigma}^{n_0}\times \cdots \times
    \hat{\Sigma}^{n_{t-1}}$ that satisfies $E$ is an image of $(S_0,\ldots,S_{t-1})$ under
    a letter-to-letter morphism (there is a function $\phi : \Sigma \to \hat{\Sigma}$ such
    that $\hat{S}_i\position{x}=\phi(S_i\position{x})$ for $i\in \fragmentco{0}{t}$ and
    $x\in \fragmentco{0}{n_i}$).
\end{definition}

\begin{corollary}\label{cor:solve_substring_equations}
    Given a tuple $(n_0,\ldots,n_{t-1})$ and a system $E$ of substring equations on
    length-$(n_0,\ldots,n_{t-1})$ strings, one can in $\Oh((|E|\log n + t)^2\log^5 n)$
    time construct an xSLP of size $\Oh((|E|\log n + t)\log^3 n)$ representing a universal
    solution of $E$, where $n=\sum_{i=0}^{t-1} n_i$.
\end{corollary}
\begin{proof}
    We create a system $E'$ of equations on length-$n$ strings by introducing a formal
    variable $T$ representing a concatenation of strings $T_0\cdots T_{t-1}$.
    For this, every equation $T_i\fragmentco{x}{x'}=T_j\fragmentco{y}{y'}$ is interpreted
    as an equation $T\fragmentco{s_i+x}{s_i+x'}=T\fragmentco{s_j+y}{s_j+y'}$, where
    $0=s_0,\ldots,s_{t}=n$ are the prefix sums of the sequence $n_0,\ldots,n_{t-1}$.
    It is easy to see that $(S_0,\ldots,S_{t-1})$ is a universal solution to $E$ if and
    only if $S\coloneqq S_0\cdots S_{t-1}$ is a universal solution to $E'$.

    Consequently, we proceed as in the proof of \cref{prp:solve_substring_equations} to
    construct an xSLP $\G$ that generates a universal solution $S$ of $E'$.
    Then, we apply $\textsc{Split}(s_i)$ for $i \in \fragmentoo{0}{t}$; this brings the number of updates of the underlying
    \pn{Dynamic xSLP Split Substitute} data structure from $\Oh(|E|\log n)$ to
    $\Oh(|E|\log n + t)$, so the xSLP size and the total running time increase to
    $\Oh((|E|\log n + t)\log^3 n)$ and $\Oh((|E|\log n + t)^2\log^5 n)$, respectively.

    Finally, we exploit the fact the underlying SLP $\G_{\mP}$ is an AVL grammar (see the
    proof of \cref{lem:split_substitute_impl}), so we can use the \emph{extract} operation
    to add symbols $A_i$ with $\val_{\G}(A_i)=S_i=S\fragmentco{s_i}{s_{i+1}}$.

    Each of these extractions takes $\Oh(\log n)$ time and adds $\Oh(\log n)$ symbols to
    $\G_{\mP}$ and $\G$; the total cost of $\Oh(t \log n)$ is dominated by the earlier
    steps of the algorithm.
    We return $\G$ with $A_0,\ldots,A_{t-1}$ interpreted as starting symbols.
\end{proof}

\subsection{A \modelname Implementation for xSLP-Compressed Strings}

In this (sub)section, we discuss how to implement \modelname operations on an xSLP. To
achieve this, we define a method for extracting an SLP from an xSLP.

\begin{definition}\label{lem:sub_xSLP}
    Let $\G$ be an xSLP representing a string $S$.
    We define the SLP $\hG$ representing the string $\hS$ as follows.
    We construct $\hS$ from $S$ by replacing $\#_i^A$ with a new terminal $\$^A$ for every
    pseudo-terminal $A\in \mP_G$ and every $i\in\fragmentoo{0}{|A|}$ (note that $\#_0^A$
    is preserved).
    Formally, we generate $\hS$ using an SLP $\hG$ of size $\Oh(|\G|\log n)$ obtained from
    $\G$ with the following changes:
    \begin{itemize}
        \item For every pseudo-terminal $A\in \mP_\G$ of length $|A|>2$, create $\Oh(\log
            |A|)$ auxiliary non-terminals generating powers of $\$^A$, including a
            non-terminal $A'$ generating $(\$^A)^{|A|-1}$, and set $\rhs_{\hG}(A)=\#_0^A
            A'$.
        \item For every pseudo-terminal $A\in \mP_\G$ of length $|A|=2$, set
            $\rhs_{\hG}(A)=\#_0^A \$^A$.
        \item For every pseudo-terminal $A\in \mP_\G$ of length $|A|=1$, replace $A$ with
            $\#_0^A$ in every production. \qedhere
    \end{itemize}
\end{definition}

It is not difficult to see that $\hG$ can be constructed out of $\G$ in $\Oh(g\log n)$
time.
The next lemma captures a crucial property of $\hS$.

\begin{lemma}\label{clm:hs}
    Let $\G$ be an xSLP representing a string $S$,
    and consider $\hG$ and $\hS$ from \cref{lem:sub_xSLP}.
    Further, consider matching fragments
    $\hS\fragmentco{x}{x+\ell}=\hS\fragmentco{y}{y+\ell}$ such that
    $S\fragmentco{x}{x+\ell}\ne S\fragmentco{y}{y+\ell}$.

    Then, $S\fragmentco{x}{x+\ell}=\#^A_{i}\cdots \#^A_{i+\ell-1}$ and
    $S\fragmentco{y}{y+\ell}=\#^A_{j}\cdots \#^A_{j+\ell-1}$
    holds for some $A\in \mP_\G$ and $i,j\in \fragmentoc{0}{|A|-\ell}$.
\end{lemma}
\begin{proof}
    Consider $r\in \fragmentco{0}{\ell}$ such that $S\position{x+r}\ne S\position{y+r}$.
    By construction of $\hS$, we must have $\hS\position{x+r}=\$^A=\hS\position{y+r}$ and
    $S\position{x+r}=\#^A_i\ne \#^A_j=S\position{y+r}$ for some $A\in \mP_\G$ and distinct
    $i,j\in \fragmentoo{0}{|A|}$; by symmetry, we may assume $i<j$ without loss of
    generality.

    It remains to prove, by induction on $|d-r|$, that  $i+d-r,j+d-r\in
    \fragmentoo{0}{|A|}$ and $S\position{x+d}=\#^A_{i+d-r}\ne
    \#^A_{j+d-r}=S\position{y+d}$ holds for every $d\in \fragmentco{0}{\ell}$.
    The claim holds trivially if $d=r$.

    Next, suppose that $d<r$.
    The inductive assumption shows that $S\position{x+d+1}=\#^A_{i+d+1-r}\ne
    \#^A_{j+d+1-r}=S\position{y+d+1}$ and $j+d+1-r>i+d+1-r>0$.
    Consequently, $j+d-r >0$, $S\position{y+d}=\#^A_{j+d-r}$, and
    $\hS\position{x+d}=\hS\position{y+d}=\$^A$.

    The only character of the form $\#^A_{i'}$ with $i'\in \fragmentoo{0}{|A|}$ that may
    precede $\#^A_{i+d+1-r}$ is $\#^A_{i+d-r}$, which completes the proof of the inductive
    step.

    The case of $d>r$ is symmetric.
    The inductive assumption shows that $S\position{x+d-1}=\#^A_{i+d-1-r}\ne
    \#^A_{j+d-1-r}=S\position{y+d-1}$ and $i+d-1-r<j+d-1-r<|A|$.
    Consequently, $i+d-r <|A|$, $S\position{x+d}=\#^A_{i+d-r}$, and
    $\hS\position{x+d}=\hS\position{y+d}=\$^A$.

    The only character of the form $\#^A_{j'}$ with $j'\in \fragmentoo{0}{|A|}$ that may
    follow $\#^A_{j+d-1-r}$ is $\#^A_{j+d-r}$, which completes the proof of the inductive
    step.
 \end{proof}

Next, we introduce the results from~\cite{BCPT15,P19} to support rank and select
operations on an SLP.

 \begin{lemmaq}{\cite{BCPT15,P19}}\label{lem:rank_queries}
    Given an SLP $\D$ of size $g$ that represents a binary string $D$,
    we preprocess it in $\Oh(|\D|\log n)$ time to answer \emph{rank} and \emph{select}
    queries in $\Oh(\log n)$ time.

    Here, \emph{rank} queries, given a position $i\in \fragment{0}{n}$, determine the
    number of $\one$s in $D\fragmentco{0}{i}$, whereas \emph{select} queries, given a rank
    $r\in \fragmentco{0}{\mathsf{rank}(n)}$ identify a position $i$ such that
    $D\position{i}=\one$ and $\mathsf{rank}(i)=r$.
\end{lemmaq}

With all components in place, we are now ready to give the proof of
\cref{lem:pillar_on_xslp}.

\tentwo
\begin{proof}
    Set $g \coloneqq |\G|$.
    We first construct $\hG$ out of $\G$
    as defined in \cref{lem:sub_xSLP} out of $\G$ in $\Oh(g\log n)$ time.
    Then, we apply \cite[Theorem 7.13]{CKW20} with recent improvements for IPM
    queries~\cite{DK24} so that, after pre-processing $\hG$ in $\Oh(g\log^2 n)$ time,
    \modelname operations on $\hS$ can be implemented in $\Oh(\log n)$ time.

    In order to implement \modelname operations on $S$, we also use rank and select
    queries on an auxiliary binary string $D$ obtained from $\hS$ by replacing every
    terminal of the form $\#_0^A$ (for $A\in \mP_\G$) with $\one$ and every other terminal
    with $\zero$.
    Formally, we generate $D$ with an SLP $\D$ of size $\Oh(g\log n)$ obtained from $\hG$
    by performing the aformentioned replacements in every production.
    We construct $\D$ in $\Oh(g\log n)$ time and use \cref{lem:rank_queries}
    to support rank and select queries.

    As discussed in \cite[Section 7]{CKW20}, efficient implementation of the \modelname
    operations on $S$ reduces to supporting $\textsc{Access}$, $\LCE$, $\LCE^R$, and
    $\LCE$ queries on $S$.
    We cover these queries one by one.
    \begin{description}
    \item[$\textsc{Access}$] Given $i\in \fragmentco{0}{n}$, the task to retrieve
        $S\position{i}$.
    To implement this operation, we first retrieve $\hS\position{i}$.
    If $\hS\position{i}$ is not of the form $\$^A$ for $A\in \mP_\G$, then $S\position{i}=\hS\position{i}$.
    Otherwise, we know that $S\position{i}=\#_j^A$ for some $j\in \fragmentoo{0}{|A|}$,
    but we still need to retrieve $j$.
    For this, we query $D$ to determine $i'\coloneqq \mathsf{select}(\mathsf{rank}(i)-1)$,
    which is the largest position in $\fragmentco{0}{i}$ such that $D\position{i'}=\one$.
    We observe that $S\fragmentco{i'}{i'+|A|}=\val_\G(A)$, and thus
    $S\position{i}=\#_{i-i'}^A$.
    The query complexity is $\Oh(\log n)$.
    \item[$\LCE$] Given two positions $x,y\in \fragment{0}{n}$, the task is to report the
        length of the longest common prefix of $S\fragmentco{x}{n}$ and
        $S\fragmentco{y}{n}$.
    To answer this query, we compute $\ell=\LCE_{\hS}(x,y)$.
    If $\ell=0$, we return $0$.
    Otherwise, we compare $S\position{x}$ with $S\position{y}$ and return $\ell$ or $0$
    depending on whether these two characters are equal or not.
    Since $\hS$ is obtained from $S$ by a letter-to-letter morphism, we have
    $\LCE_S(x,y)\le \LCE_{\hS}(x,y)$ and our algorithm never underestimates the result.
    Moreover, if $\hS\fragmentco{x}{x+\ell}=\hS\fragmentco{y}{y+\ell}$, then \cref{clm:hs}
    shows that $S\fragmentco{x}{x+\ell}=S\fragmentco{y}{y+\ell}$ or $S\position{x}\ne
    S\position{y}$, so the answer is indeed $\ell$ or $0$ depending on whether
    $S\position{x}= S\position{y}$.
    The query complexity is $\Oh(\log n)$.
    \item[$\LCE^R$] Given two positions $x,y\in \fragment{0}{n}$, the task is to report
        the length of the longest common suffix of $S\fragmentco{0}{x}$ and
        $S\fragmentco{0}{y}$.
    These queries are symmetric to the $\LCE$ queries, and our algorithm is analogous: we
    compute $\ell=\LCE^R(x,y)$ and return $0$ if $\ell=0$.
    Otherwise, we compare $S\position{x-1}$ with $S\position{y-1}$ and return $\ell$ or
    $0$ depending on whether the characters are equal or not.
    The correctness proof is symmetric to that for $\LCE$ queries, and the query complexity is
    $\Oh(\log n)$.
    \item[$\IPM$] Given two fragments $P=S\fragmentco{x_p}{x_p+\ell_p}$ and
        $T=S\fragmentco{x_t}{x_t+\ell_t}$ with $\ell_t < 2\ell_p$, the task is to report
        an arithmetic progression representing $\Occ(P,T)$.
    In this case, we first compute $\hO=\Occ(\hat{P},\hat{T})$, where
    $\hat{P}=\hS\fragmentco{x_p}{x_p+\ell_p}$ and
    $\hat{T}=\hS\fragmentco{x_t}{x_t+\ell_t}$.
    If $\hO=\emptyset$, we report $\emptyset$.
    Otherwise, we retrieve $S\position{x_p}$ and report $\hO$ unless
    $S\position{x_p}=\#^A_{i_p}$ for some $A\in \mP_\G$ and $i_p\in
    \fragmentoc{0}{|A|-\ell_p}$.
    In that exceptional case, we also retrieve $S\position{x_t+\min\hO}$, which must be of
    the form $\#^A_{i_t}$ for some $i_t\in \fragmentoc{0}{|A|-\ell_p}$, and return
    $\{\min\hO+i_p-i_t\}\cap \hO$.
    Since $\hS$ is obtained from $S$ by a letter-to-letter morphism, we have
    $\Occ(P,T)\subseteq \hO=\Occ(\hat{P},\hat{T})$.
    In particular, we correctly return $\Occ(P,T)=\emptyset$ if $\hO=\emptyset$.
    By \cref{clm:hs}, we have $\Occ(P,T)=\hO$ unless $P=\#^A_{i_P}\cdots
    \#^A_{i_P+\ell_P-1}$ holds for some $A\in \mP_\G$ and $i_P\in
    \fragmentoc{0}{|A|-\ell_P}$; we check this condition based on
    $P\position{0}=S\position{x_p}$ and correctly return $\hO$ if the condition is not
    satisfied.
    If the condition is satisfied, we have
    $T\fragmentco{\min\hO}{\ell_P+\max\hO}=\#^A_{i_t}\cdots
    \#^A_{i_t+\max\hO-\min\hO+\ell_P-1}$.
    In this case, there can be only one occurrence at position $\min\hO+i_p-i_t$, provided
    that this position is within $\hO=\fragment{\min\hO}{\max\hO}$.
    This completes the correctness proof. The query complexity is $\Oh(\log n)$.
    \qedhere
    \end{description}
\end{proof}

%% file: tikz/length_enhanced_string.tex
\begin{tikzpicture}[yscale=0.5,xscale=.65]
    \draw[very thin, black!70] (0,1.75) -- (0,-1.5);
    \draw[very thin, black!70] (4,1.75) -- (4,-1.5);
    \draw[very thin, black!70] (6,1.75) -- (6,-1.5);
    \draw[very thin, black!70] (7,1.75) -- (7,-1.5);
    \draw[very thin, black!70] (11,1.75) -- (11,-1.5);
    \draw[very thin, black!70] (13,1.75) -- (13,-1.5);

    \draw[thick, rounded corners=2.5pt, fill=black!5] (0,0) rectangle node
    {$X$} (4,1);
    \draw[thick, rounded corners=2.5pt] (4,0) rectangle node {$Y$} (6,1);
    \draw[thick, rounded corners=2.5pt, fill=black!15] (6,0) rectangle node
    {$Z$} (7,1);
    \draw[thick, rounded corners=2.5pt, fill=black!5] (7,0) rectangle node
    {$X$} (11,1);
    \draw[thick, rounded corners=2.5pt] (11,0) rectangle node {$Y$} (13,1);

    \foreach \x in {0,...,12} {
        \node[black!80] (n\x) at (\x+0.5, 1.5) {\tiny ${\x}$};
    }

    \newcommand{\sdollar}{\#_0^X, \#_1^X, \#_2^X, \#_3^X, \#_0^Y, \#_1^Y, \#_0^Z, \#_0^X, \#_1^X, \#_2^X, \#_3^X, \#_0^Y, \#_1^Y}

    \foreach \x [count=\i] in \sdollar {
        \node (s\i) at (\i-1+0.5, -1) {$\x$};
    }

    \node at (-1.4cm, 0.5cm) {\large ${P_S}$};
    \node at (-1.4cm, -0.85cm) {\large ${S}$};

\end{tikzpicture}

%% file: tikz/grammar_noper.tex
\begin{tikzpicture}[scale=0.5]
    \draw[very thin,black!50] (0, 2) node[black,above] {$\hx$} -- (0, -2.5) node[black,below] {$\hy$};

    \draw[thick,cap=round,dashed] (-2, 0) -- (-1, 0);
    \draw[thick] (-1, 0) -- (10, 0);
    \draw[thick,cap=round,dashed] (10, 0) -- (11, 0);

    \draw[thick,cap=round,dashed] (-2, 1) -- (-1, 1);
    \draw[thick] (-1, 1) -- (10, 1);
    \draw[thick,cap=round,dashed] (10, 1) -- (11, 1);

    \draw (0, 1) -- (0, 0);
    \draw (7, 1) -- (7, 0);

    \draw[|-|] (0,-2.25) --node[below] {$r$} (6, -2.25);
    \draw[|-|] (0, 1.5) --node[above] {$\ell$} (7, 1.5);

    \node at (3.5,0.5) {$X$};

    \draw[thick,cap=round,dashed] (-2, -1) -- (-1, -1);
    \draw[thick] (-1, -1) -- (10, -1);
    \draw[thick,cap=round,dashed] (10, -1) -- (11, -1);

    \draw[thick,cap=round,dashed] (-2, -2) -- (-1, -2);
    \draw[thick] (-1, -2) -- (10, -2);
    \draw[thick,cap=round,dashed] (10, -2) -- (11, -2);

    \draw (0, -1) -- (0, -2);
    \draw (1, -1) -- (1, -2);
    \draw (6, -1) -- (6, -2);

    \node at (0.5,-1.5) {$Y$};
    \node at (8.5,-1.5) {$Y'$};

    \draw[thick,fill=black!15] (1,-2) rectangle (6,-1);

    \draw[->, cap=round,line width=.5mm, >=stealth] (5,-3.75) -- (5,-5.25);

    \begin{scope}[yshift=-8.5cm]
        \draw[very thin,black!50] (0, 2) node[black,above] {$\hx$} -- (0, -2.5) node[black,below] {$\hy$};
        \draw[thick,cap=round,dashed] (-2, 0) -- (-1, 0);
        \draw[thick] (-1, 0) -- (10, 0);
        \draw[thick,cap=round,dashed] (10, 0) -- (11, 0);

        \draw[thick,cap=round,dashed] (-2, 1) -- (-1, 1);
        \draw[thick] (-1, 1) -- (10, 1);
        \draw[thick,cap=round,dashed] (10, 1) -- (11, 1);

        \draw (0, 1) -- (0, 0);
        \draw (1, 1) -- (1, 0);
        \draw (7, 1) -- (7, 0);
        \draw (6, 1) -- (6, 0);

        \draw[|-|] (0,-2.25) --node[below] {$r$} (6, -2.25);
        \draw[|-|] (0, 1.5) --node[above] {$\ell$} (7, 1.5);

        \node at (0.5,0.5) {$Y$};

        \draw[thick,fill=black!15] (1,-0) rectangle (6,1);

        \draw[thick,cap=round,dashed] (-2, -1) -- (-1, -1);
        \draw[thick] (-1, -1) -- (10, -1);
        \draw[thick,cap=round,dashed] (10, -1) -- (11, -1);

        \draw[thick,cap=round,dashed] (-2, -2) -- (-1, -2);
        \draw[thick] (-1, -2) -- (10, -2);
        \draw[thick,cap=round,dashed] (10, -2) -- (11, -2);

        \draw (0, -1) -- (0, -2);
        \draw (1, -1) -- (1, -2);
        \draw (6, -1) -- (6, -2);

        \node at (0.5,-1.5) {$Y$};
        \node at (8.5,-1.5) {$Y'$};

        \draw[thick,fill=black!15] (1,-2) rectangle (6,-1);

        \draw[cap=round,color=red!50, line width=1mm, opacity=0.75] (6, 2) -- (6, -2.5);
    \end{scope}

\end{tikzpicture}

%% file: tikz/grammar_per.tex
\begin{tikzpicture}[scale=0.5]
    \draw[very thin,black!60] (0, 2) node[black,above] {$\hx$} -- (0, -2.5) node[black,below] {$\hy$};

    \draw[thick,cap=round,dashed] (-2, 0) -- (-1, 0);
    \draw[thick] (-1, 0) -- (10, 0);
    \draw[thick,cap=round,dashed] (10, 0) -- (11, 0);

    \draw[thick,cap=round,dashed] (-2, 1) -- (-1, 1);
    \draw[thick] (-1, 1) -- (10, 1);
    \draw[thick,cap=round,dashed] (10, 1) -- (11, 1);

    \draw (0, 1) -- (0, 0);
    \draw (7.5, 1) -- (7.5, 0);

    \draw[|-|] (0,-2.25) --node[below] {$r$} (2, -2.25);
    \draw[|-|] (0, 1.5) --node[above] {$\ell$} (7.5, 1.5);

    \node at (3.75,0.5) {$X=Y'$};

    \draw[thick,cap=round,dashed] (-2, -1) -- (-1, -1);
    \draw[thick] (-1, -1) -- (10, -1);
    \draw[thick,cap=round,dashed] (10, -1) -- (11, -1);

    \draw[thick,cap=round,dashed] (-2, -2) -- (-1, -2);
    \draw[thick] (-1, -2) -- (10, -2);
    \draw[thick,cap=round,dashed] (10, -2) -- (11, -2);

    \draw (0, -1) -- (0, -2);
    \draw (1, -1) -- (1, -2);
    \draw (2, -1) -- (2, -2);
    \draw (9.5, -1) -- (9.5, -2);

    \node at (0.5,-1.5) {$Y$};
    \node at (5.75,-1.5) {$X=Y'$};

    \draw[thick,fill=black!15] (1,-2) rectangle (2,-1);

    \draw[->, cap=round,line width=.5mm, >=stealth] (5,-3.75) -- (5,-5.25);

    \begin{scope}[yshift=-8.5cm]
        \draw[very thin,black!60] (0, 2) node[above,black] {$\hx$} -- (0, -2.5)
            node[below,black] {$\hy$};

        \draw[thick,cap=round,dashed] (-2, 0) -- (-1, 0);
        \draw[thick] (-1, 0) -- (10, 0);
        \draw[thick,cap=round,dashed] (10, 0) -- (11, 0);

        \draw[thick,cap=round,dashed] (-2, 1) -- (-1, 1);
        \draw[thick] (-1, 1) -- (10, 1);
        \draw[thick,cap=round,dashed] (10, 1) -- (11, 1);

        \draw (0, 1) -- (0, 0);
        \draw (1, 1) -- (1, 0);
        \draw (2, 1) -- (2, 0);
        \draw (3, 1) -- (3, 0);
        \draw (4, 1) -- (4, 0);
        \draw (5, 1) -- (5, 0);
        \draw (7.5, 1) -- (7.5, 0);
        \draw (6, 1) -- (6, 0);

        \draw[|-|] (0,-2.25) --node[below] {$r$} (2, -2.25);
        \draw[|-|] (0, 1.5) --node[above] {$\ell$} (7.5, 1.5);

        \node at (0.5,0.5) {$Y$};
        \node at (2.5,0.5) {$Y$};
        \node at (4.5,0.5) {$Y$};

        \draw[thick,fill=black!15] (1,-0) rectangle (2,1);
        \draw[thick,fill=black!15] (3,-0) rectangle (4,1);
        \draw[thick,fill=black!15] (5,-0) rectangle (6,1);

        \draw[thick,cap=round,dashed] (-2, -1) -- (-1, -1);
        \draw[thick] (-1, -1) -- (10, -1);
        \draw[thick,cap=round,dashed] (10, -1) -- (11, -1);

        \draw[thick,cap=round,dashed] (-2, -2) -- (-1, -2);
        \draw[thick] (-1, -2) -- (10, -2);
        \draw[thick,cap=round,dashed] (10, -2) -- (11, -2);

        \draw (0, -1) -- (0, -2);
        \draw (1, -1) -- (1, -2);
        \draw (2, -1) -- (2, -2);
        \draw (9.5, -1) -- (9.5, -2);

        \node at (0.5,-1.5) {$Y$};
        \node at (2.5,-1.5) {$Y$};
        \node at (4.5,-1.5) {$Y$};
        \node at (6.5,-1.5) {$Y$};

        \draw[thick,fill=black!15] (1,-2) rectangle (2,-1);
        \draw[thick,fill=black!15] (3,-2) rectangle (4,-1);
        \draw[thick,fill=black!15] (5,-2) rectangle (6,-1);
        \draw[thick,fill=black!15] (7,-2) rectangle (8,-1);

        \draw[color=red!50, line width=1mm, opacity=0.75,cap=round] (6, 2) -- (6, -2.5);
    \end{scope}

\end{tikzpicture}

%% file: s4_qpmwm.tex
\section{Quantum Algorithm for Pattern Matching with Mismatches}
\label{sec:pmwm}

In this section, we give our quantum algorithm for \PMwM, that
is, we prove \cref{thm:qpmwm}.

\qpmwm*
\medskip

As the main technical step, we prove \cref{thm:qhdfindproxystrings}.
Specifically, given a pattern \( P \) of length \( m \), a text \( T \) of length \( n \le
3/2 \cdot m \), and an integer threshold \( k > 0 \), we show how to construct an xSLP for
a proxy pattern $P^\#$ and a proxy text $T^\#$ that are equivalent to the original text and
pattern for the purpose of computing \(\OccH_k(P, T)\).

\begin{restatable*}{theorem}{qhdfindproxystrings}\label{thm:qhdfindproxystrings}
    There exists a quantum algorithm that, given a pattern $P$ of length $m$, a text $T$
    of length $n \le 3/2 \cdot m$, and an integer threshold $k > 0$, outputs an xSLP of
    size $\Ohtilde(k)$ representing strings $P^\#$ and $T^\#$ such that \(\OccH_k(P, T) =
    \OccH_k(P^{\#}, T^{\#})\) and $\MI(P, T\fragmentco{x}{x+m})=\MI(P^\#,
    T^\#\fragmentco{x}{x+m})$ holds for every $x\in \OccH_k(P,T)$.
    The algorithm uses $\Ohtilde(\sqrt{km})$ quantum queries and $\Ohtilde(\sqrt{km}+k^2)$ time.
\end{restatable*}

\subsection{Step 1: Computing a Candidate Set for Occurrences}
\label{sec:qhd_analyze}
\label{sec:qhd_postanalyze}

In this (sub)section we prove \cref{lem:qhd_analyze}.

\qhdanalyze*
\medskip

Recall from the Technical Overview that we start by analyzing the pattern using the
quantum algorithm for {\tt Analyze} from \cite{JN23}.

\quantumanalyze
\medskip

In the following, we tackle each of the three cases. We start with useful support routines
for the periodic case (which in turn also helps with repetitive regions).

\subsubsection{Support Routines for Periodic Patterns}
\label{sec:periodic_case}

We begin by showing that if both $P$ and $T$ have a low distance to $Q^*$ for the same
approximate period $Q^*$, all $k$-mismatch occurrences must align with the period, leading
to a candidate set that forms an arithmetic progression.

\begin{lemma}
    \label{lem:periodiccase_hd_progression}
    Let $P$ denote a pattern of length $m$ and $T$ denote a text of length $n$.
    Suppose there are a positive integer $d$ and a primitive string $Q$\footnote{Access to the string \( Q \) is provided through another string \( Q' \) and a rotation \( r \), such that \( Q = \text{rot}^{r}(Q') \). For simplicity, throughout this paper, we omit the maintained indices when accessing strings derived from rotated versions.}
    with $|Q| \leq m/8d$, $\hd(P,Q^*) \leq d$,
    and $\hd(T,\rot^{a}(Q)^*) \leq 4d$ for some $a \in \fragmentco{0}{|Q|}$.
    Then,  $\OccH_{5d}(P,T) = \{a + b \cdot |Q| \mid b \in \mathbb{Z}\} \cap \fragment{0}{n-m}$.
\end{lemma}

\begin{proof}
    Write $C\coloneqq  \{a + b \cdot |Q| \mid b \in \mathbb{Z}\} \cap \fragment{0}{n-m}$.

    To prove $\OccH_{5d}(P,T) \subseteq C$, consider \( x \in \fragment{0}{n-m}\setminus C \).
    From $m/8d \geq |Q| \geq 1$, we obtain that \( P \) contains at least
    $\floor{{(m-d)}/{|Q|}} \geq \floor{8d - 8d^2/m} \geq 7d$ occurrences of \( Q \) that
    start at positions that are divisible by $|Q|$.
    When $P$ is aligned with $T\fragmentco{x}{x+m}$, at least \( 7d - 4d=3d \) full
    occurrences of \( Q \) are aligned with a cyclic rotation of \( Q \) that is not \( Q
    \) itself.
    Since \( Q \) is a primitive string, each of these copies has two mismatches,
    resulting in \( \hd(P, T\fragmentco{x}{x+m}) \geq 6d > 5d \), and thus $x\notin
    \OccH_{5d}(P,T)$.

    To prove $C \subseteq \OccH_{5d}(P,T)$, observe that for any \( x \in C \), we have \(
    \hd(T\fragmentco{x}{x+m}, Q^*) \leq 4d \). By the triangle inequality, we obtain \[
        \hd(P, T\fragmentco{x}{x+m}) \leq \hd(T\fragmentco{x}{x+m}, Q^*) + \hd(P, Q^*)
    \leq 4d + d \leq 5d, \]
    completing the proof.
\end{proof}

When the distance $\hd(P,Q^*)$ is not too small, we can also identify a candidate set $C$ of size $|C| = \Ohtilde(k)$.

\begin{lemma}
    \label{lem:periodiccase_hd_small_c}
    Let $P$ denote a pattern of length $m$, let $T$ denote a text of length $n$, and let
    $k > 0$ denote an integer threshold.
    Suppose there are a positive integer $d \geq 2k$ and a primitive string $Q$ with $|Q| \leq m/8d$, $\hd(P,Q^*) = d$,
    and $\hd(T,\rot^{a}(Q)^*) \leq 4d$ for some known $a \in \fragmentco{0}{|Q|}$.

    Then, there is a quantum algorithm that outputs a candidate set $\OccH_k(P,T) \subseteq C
    \subseteq \fragment{0}{n-m}$ such that $|C| = \Ohtilde(d)$ using $\Ohtilde(\sqrt{dm})$
    quantum time.
\end{lemma}
\begin{proof}
    Consider the following procedure.
    \begin{enumerate}[(i)]
        \item Use \cref{lem:find_mismatches} to compute $\Mis(P, Q^*)$ and $\Mis(T, \rot^{a}(Q)^*)$.
            \label{alg:periodiccase_hd_small_c:0}
        \item Select uniformly at random $y \in \Mis(P, Q^*)$.
            \label{alg:periodiccase_hd_small_c:i}
        \item Return the candidate set $C' \coloneqq \{y' - y \mid y' \in \Mis(T, \rot^{a}(Q)^*)\} \cap \fragment{0}{n-m}$.
            \label{alg:periodiccase_hd_small_c:ii}
    \end{enumerate}

    \begin{claim}\label{claim:periodiccase_hd_small_c:prob}
         Let $x \in \OccH_k(P,T)$ be arbitrary. Then, with constant probability, we have $x \in C'$.
    \end{claim}
    \begin{claimproof}
        According to \cref{lem:periodiccase_hd_progression}, the $x$ aligns the primitive
        period $Q$ of both $P$ and $T$. Formally, this means that $x$ has a residue of $a$
        modulo $|Q|$. Since the periods align, the mismatches contributing to
        $\hd(P,T\fragmentco{x}{x+m})$ can be entirely deduced by examining the mismatches
        between $P$ and $T$ with $Q^*$.
        More precisely, we have
        \[
            \Mis(P, T\fragmentco{x}{x+m}) = \Big(\Mis(P, Q^*) \cup \Mis(T\fragmentco{x}{x+m}, Q^*)\Big) \setminus M,
        \]
        where $M = \{x' \in \Mis(P, Q^*) \cap \Mis(T\fragmentco{x}{x+m}, Q^*) \mid T\position{x} = P\position{x+x'}\}$.
        From $|\Mis(P, Q^*)| = d$, we obtain
        \[
            |\Mis(P, T\fragmentco{xi}{x+m})| \geq |\Mis(P, Q^*)| - |M| = d - |M|.
        \]
        Combining this with the fact that $|\Mis(P, T\fragmentco{x}{x+m})| \leq k$, we
        obtain $k \geq d - |M|$, which rearranges to $|M| \geq d - k \geq d/2$.
        Thus, with a probability of at least $1/2$, we can select $y$ such that $y \in M$
        in \eqref{alg:periodiccase_hd_small_c:i}. If this occurs, then $x \in C'$.
    \end{claimproof}

    From \cref{claim:periodiccase_hd_small_c:prob} follows that by repeating
    $\tilde{\mathcal{O}}(1)$ times independently the procedure, and setting $C$ to be the
    union of the candidate sets, we
    obtain with high probability a set $\OccH_k(P,T) \subseteq C \subseteq
    \fragment{0}{n-m}$ such that $|C| = \Ohtilde(d)$.

    The quantum time required for a single execution of the procedure is dominated by
    the computation of $\Mis(P, Q^*)$ and $\Mis(T, \rot^{a}(Q)^*)$ in
    \eqref{alg:periodiccase_hd_small_c:0} requiring $\Oh(\sqrt{dm})$ time. The total
    quantum time for the algorithm differs from that of a single execution by only a
    logarithmic factor.
\end{proof}

In the next lemma, we extend
\cref{lem:periodiccase_hd_progression,lem:periodiccase_hd_small_c} to the case where there
is no guarantee that $\hd(T,\rot^{a}(Q)^*)$ is small for some $a$.

\begin{lemma}
    \label{lem:hd_compstructure}
    Let $P$ denote a pattern of length $m$,
    let $T$ denote a text of length $n \leq 3/2 \cdot m$, and let $k > 0$ denote an integer threshold.
    Suppose there are a positive integer $d \geq 2k$ and a primitive string $Q$ with $|Q|
    \leq m/8d$ and $\hd(P,Q^*) \leq d$.

    Then, there is a quantum algorithm that outputs a candidate set $\OccH_k(P,T)
    \subseteq C \subseteq \OccH_{5d}(P,T)$ such that $C$ forms an arithmetic progression.
    Whenever $\hd(P,Q^*) = d$, the algorithm may instead output a set
    $\OccH_k(P,T) \subseteq C \subseteq \fragment{0}{n-m}$ such that $C = \Ohtilde(d)$.
    The algorithm requires $\Ohtilde(\sqrt{dm})$ quantum time.
\end{lemma}

\begin{proof}
    Consider the following procedure.
    \begin{enumerate}[(i)]
        \item Partition $T\fragmentco{m-4d|Q|}{m}$ into $4d$ blocks of length $|Q|$.
            \label{it:alg:hd_compstructure:i}

        \item For each block, check which rotation of $Q$ it is, if any,
            using \cref{claim:hd_compstructure_rot}.
            If no rotation is found at least $4d-d-k$  times, return $C=\emptyset$.
            Otherwise, select a block $T\fragmentco{y}{y+|Q|} \substr
            T\fragmentco{m-4d|Q|}{m}$ for some $y$ such that $T\fragmentco{y}{y+|Q|} =
            \rot^a(Q)$, where $a \in \fragmentco{0}{|Q|}$ is the rotation found at least
            $4d-d-k$ times.
            \label{it:alg:hd_compstructure:ii}

        \item Extend $T\fragmentco{y}{y+|Q|}$ with up to $d+k$ mismatches in each direction.
            That is, via \cref{lem:find_mismatches} find the smallest $\ell$ such that
            $\hd(T\fragmentco{\ell}{y}, \rot^{a+y-\ell}(Q)^*) \leq d + k$,
            and similarly, the largest $r$ such that $\hd(T\fragmentco{y}{r}, \rot^a(Q)^*) \leq d + k$.
            \label{it:alg:hd_compstructure:iii}

        \item If an arithmetic progression is requested for $C$, apply \cref{lem:periodiccase_hd_progression}
            to $P,T' = T\fragmentco{\ell}{r}, k$. If instead a small $C$ is requested and
            $\hd(P,Q^*) = d$, use \cref{lem:periodiccase_hd_small_c} on $P,T' =
            T\fragmentco{\ell}{r}, k$.
            \label{it:alg:hd_compstructure:iv}
    \end{enumerate}

    Next, we demonstrate the correctness of the algorithm.

    First, observe that in Step \eqref{it:alg:hd_compstructure:iv},
    \cref{lem:periodiccase_hd_progression} and \cref{lem:periodiccase_hd_small_c} are
    always applied with valid input since
    \[
        \hd(T\fragmentco{\ell}{r}, \rot^{a+y-\ell}(Q)^*) = \hd(T\fragmentco{\ell}{y},
        \rot^{a+y-\ell}(Q)^*) + \hd(T\fragmentco{y}{r}, \rot^a(Q)^*) \leq 2d + 2k \leq 4d.
    \]
    Therefore, if $\OccH_k(P,T) = \emptyset$, the algorithm correctly returns $\emptyset$.

    Assume from now on that $\OccH_k(P,T) \neq \emptyset$. Let $x \in \OccH_k(P,T)$ be
    arbitrary. For such position $x \leq n - m \leq m/2$ holds. Our goal is to prove that
    $x \in \fragmentco{\ell}{r-m}$.
    The $k$-mismatch occurrence at position $x$ aligns $\hat{T} \coloneqq
    T\fragmentco{m-4d|Q|}{m}$ onto $\hat{P} \coloneqq P\fragmentco{m-4d|Q|-x}{m-x}$,
    meaning that $\hd(\hat{P},\hat{T}) \leq \hd(P,T\fragmentco{x}{x+m}) \leq k$.
    Note, $\hat{P}$ is well defined since $m-4d|Q|-x \geq m/2 - x \geq 0$.

    Further, we have $\hd(\hat{P}, \hat{Q}^*) \leq \hd(P,Q^*) \leq d$, where $\hat{Q}
    \coloneqq \rot^{m-4d|Q|-x}(Q)$. Consequently, $\hd(\hat{T}, \hat{Q}^*) \leq
    \hd(\hat{T}, \hat{Q}^*) + \hd(\hat{P}, \hat{T}) \leq d + k$, and at least $4d - d - k$
    blocks equal to $\hat{Q}$ in \eqref{it:alg:hd_compstructure:ii}.
    Thus, the algorithm proceeds to \eqref{it:alg:hd_compstructure:iii} and sets $a$ to $m-4d|Q|-x$.

    Finally, observe that $k$-mismatch occurrence at position $x$ aligns $\hd(T\fragmentco{y}{x+m}$
    onto $P\fragmentco{y - x}{m}$, which implies
    \begin{align*}
        \hd(T\fragmentco{y}{x+m}, \hat{Q}^*)
        &\leq \hd(T\fragmentco{y}{x+m}, P\fragmentco{y - x}{m}) + \hd(P\fragmentco{y - x}{m}, \hat{Q}^*) \\
        & \leq \hd(T, P) + \hd(P, Q^*)
        \leq d + k
    \end{align*}
    From the definition of $r$ follows $r \leq x+m$, yielding $x \leq r-m$.
    Using a symmetric argument we can show $\ell \leq x$.
    Hence, $x \in \fragmentco{\ell}{r-m}$, which concludes the proof of the correctness.

    For the complexity analysis, observe that step \eqref{it:alg:hd_compstructure:ii}
    requires $\Ohtilde(d\sqrt{|Q|}) = \Ohtilde(\sqrt{dm})$ time. Additionally, in step
    \eqref{it:alg:hd_compstructure:iii}, the calls to \cref{lem:find_mismatches} also take
    at most $\Ohtilde(\sqrt{dm})$ time. Combining these steps with the runtime of
    \cref{lem:periodiccase_hd_small_c}, we achieve the desired query and quantum time for
    the algorithm.
\end{proof}

\subsubsection{Computing a Set of Good Candidates for Occurrences}
\label{sec:three_cases}

Next, we discuss how to obtain a suitable set of good candidates in each of the three
cases of \cref{lem:hd_analyzeP}.
We start with the first case, that is, when \cref{lem:hd_analyzeP} returns breaks.

\begin{restatable}[Candidate Positions in the Presence of Breaks]{lemma}{hdbreakcase}\label{lem:hd_breakcase}
    Let $P$ denote a pattern of length $m$, let $T$ denote a text of length $n \le
    3/2\cdot m$, and let $k > 0$ denote an integer threshold.
    Further, suppose that we are given a collection of \(2k\) disjoint breaks $B_1,
    \ldots, B_{2k}$ of \(P\),
    each having period $\per(B_i)> m/128 k$ and length $|B_i| = \lfloor  m/8 k\rfloor$.

    Then, there is a quantum algorithm that using $\Ohtilde(\sqrt{km})$ queries and
    quantum time,
    computes a set $C$ such that $\OccH_k(P,T) \subseteq C \subseteq \fragment{0}{n-m}$
    and $|C|=\Ohtilde(k)$.
\end{restatable}
\begin{proof}
    Consider the following procedure.
    \begin{enumerate}[(i)]
        \item Select u.a.r. a break $B=P\fragmentco{\beta}{\beta+|B|}$ among $B_1, \ldots, B_{2k}$,
            and using \cref{thm:quantum_matching_all} compute $\Occ(B, T)$.
        \item Return the candidate set $C' =(\Occ(B,T)- \beta) \cap \fragment{0}{n - m}$.
    \end{enumerate}

    Fix a $k$-mismatch occurrence of $P$ in $T$.
    With probability at least $1/2$, the selected break $B$ does
    not contain any mismatch in the $k$-mismatch occurrence,
    meaning that its starting position must be contained in $C'$.
    Note that $|C'| \leq |\Occ(B,T)|\le \lceil n/\per(B)\rceil \le 192k$
    because $\per(B) > m/128k$ and $n\leq 3/2 \cdot m$.
    By repeating the procedure $\mathcal{O}(\log n)$ times independently
    and returning the union of the candidate sets, we
    obtain with high probability a set $C$ as described in the statement of the lemma.

    The quantum time required for a single execution of the procedure is dominated by
    the computation of $\Occ(B,T)$ using \cref{thm:quantum_matching_all}, which takes
    $\Ohtilde(\sqrt{km})$ quantum time. The total quantum time for the algorithm differs
    from that of a single execution by only a logarithmic factor.
\end{proof}

We continue with the second case, that is, when \cref{lem:hd_analyzeP} returns repetitive
regions.

\begin{restatable}[Candidate Positions in the Presence of Repetitive Regions]{lemma}{hdrepregion}\label{lem:hd_repregion}
    Let $P$ denote a pattern of length $m$, let $T$ denote a text of length $n \le
    3/2\cdot m$, and let $k > 0$ denote an integer threshold.
    Further, suppose that we are given
    a collection of disjoint repetitive regions $R_1,\ldots, R_{r}$
    of~total length $\sum_{i=1}^r |R_i| \ge 3/8 \cdot m$ such
    that each region $R_i$ satisfies
    $|R_i| \ge m/8 k$ and has a primitive \emph{approximate period} $Q_i$
    with $|Q_i| \le m/128 k$ and $\hd(R_i,Q_i^*) = \ceil{8 k/m\cdot |R_i|}$.

    Then, there is a quantum algorithm that using $\Ohtilde(\sqrt{km})$ queries and
    quantum time,
    computes a set $C$ such that $\OccH_k(P,T) \subseteq C \subseteq \fragment{0}{n-m}$
    and $|C|=\Ohtilde(k)$.
\end{restatable}

\begin{proof}
    First, we prove a useful claim.

    \begin{claim} \label{claim:hd_repregions_standard_trick}
        Let $R \in \{R_1, \ldots, R_{r}\}$ be a repetitive region. Then, we can compute a
        candidate set $\OccH_{\kappa}(R, T) \subseteq C_R$ of size $|C_R| = \Ohtilde(k)$,
        where $\kappa = \floor{4k/m \cdot |R|}$.
        This requires $\Ohtilde(\sqrt{km})$ queries and quantum time.
    \end{claim}

    \begin{claimproof}
        We combine the Standard Trick with \cref{lem:hd_compstructure}.
        Divide $T$ into $b = \Oh(m/|R|)$ contiguous blocks of length $|R|/2$ (the last block might be shorter),
        and iterate over all blocks.
        When iterating over the $i$-th block of the form $T\fragmentco{x_i}{\min(x_i + |R|/2, |T|)}$,
        consider the segment $S_i = T\fragmentco{x_i}{\min(x_i + 3/2 \cdot |R|, |T|)}$ and
        use \cref{lem:hd_compstructure} on $R$, $S_i$, $\kappa$, $\hd(R,Q^*)$ in the role
        of $P,T,k$ and $d$, respectively.
        As $d = \hd(R,Q^*)$, we may choose to retrieve a candidate set $\OccH_{\kappa}(R,
        S_i) \subseteq C_i$ of size $\Ohtilde(\kappa)$ from \cref{lem:hd_compstructure}.
        After processing all blocks, combine the candidate sets from each block to form
        $C_R = \bigcup_{i=1}^{b} C_i$ of size $|C_R| = \Ohtilde(m / |R| \cdot \kappa) =
        \Ohtilde(k)$.
        The overall quantum time is $\Ohtilde(m/|R| \cdot \sqrt{k/m} \cdot |R|) =
        \Ohtilde(\sqrt{km})$.
    \end{claimproof}

    Next, consider the following procedure.
    \begin{enumerate}[(i)]
        \item Select a repetitive region $R=P\fragmentco{\rho}{\rho+|R|}$ among $R_1,
            \ldots, R_{r}$,
            with probability proportional to their length, that is, $\pr{R = R_i} =
            |R_i|/\sum_{(R_i, Q_i) \in \mathcal{R}} |R_i|$.
        \item Use \cref{claim:hd_repregions_standard_trick} compute a set
            $\OccH_{\kappa}(R, T)\subseteq C_R \subseteq \fragment{0}{n-m}$ of size
            $\Ohtilde(k)$, where $\kappa = \floor{4k/m \cdot |R|}$.
        \item Return the candidate set $C' = (C_R - \rho) \cap \fragment{0}{n - m}$
    \end{enumerate}

    Fix a $k$-mismatch occurrence of $P$ in $T$.
    In \cite[Lemma~6.8]{JN23}, it is shown that with probability at least $1/16$, out of the $k$ mismatches,
    no more than $\kappa$ of them are located in $R$. This means that the starting position
    of the fixed $k$-mismatch occurrence is contained in $C'$.
    Similarly to \cref{lem:hd_breakcase}, we repeat $\tilde{\mathcal{O}}(1)$ times
    independently of the procedure and set $C$ to the union of the candidate sets.

    The quantum time required for a single execution of the procedure is dominated by the
    computation of $C_R$ using \cref{claim:hd_repregions_standard_trick}, which takes
    $\Ohtilde(\sqrt{km})$ quantum time. The total quantum time for the algorithm differs
    from that of a single execution by only a logarithmic factor.
\end{proof}

We continue with the last case, that is, when \cref{lem:hd_analyzeP} returns a (rotation
of an) approximate period.

\begin{restatable}[Candidate Positions in the Approximately Periodic Case]{lemma}{hdapproxpercase}\label{lem:hd_approxpercase}
    Let $P$ denote a pattern of length $m$, let $T$ denote a text of length $n \le
    3/2\cdot m$, and let $k > 0$ denote an integer threshold.
    Suppose that we are given (a rotation of) a primitive \emph{approximate period} $Q$ of \(P\), with $|Q|\le
    m/128 k$ and $\hd(P,Q^*) < 8 k$.

    Then, there is a quantum algorithm that using $\Ohtilde(\sqrt{km})$ queries and
    quantum time,
    computes a set $C$ such that $\OccH_k(P,T) \subseteq C \subseteq \fragment{0}{n-m}$
    and
    \begin{itemize}
        \item $|C|=\Ohtilde(k)$, \textbf{or}
        \item $C$ forms an arithmetic progression and $\OccH_k(P,T) \subseteq C \subseteq \OccH_{10k}(P,T)$.
            \qedhere
    \end{itemize}
\end{restatable}
\begin{proof}
    It suffices to use \cref{lem:hd_compstructure} on $P, T, k$, and $\max(\hd(P,Q^*), 2k)$ in the role of $d$.
    Whenever $\hd(P,Q^*) \geq 2k$, then $d = \hd(P,Q^*)$, and we ask the algorithm of \cref{lem:hd_compstructure}
    to return a candidate set of size $\Ohtilde(k)$. Otherwise, we get as candidate set an arithmetic progression.
\end{proof}

Finally, we combine
\cref{lem:hd_breakcase,lem:hd_repregion,lem:hd_approxpercase,lem:hd_analyzeP} to obtain
\cref{lem:qhd_analyze}.

\qhdanalyze
\begin{proof}
    We use \cref{lem:hd_analyzeP} to determine which case of \cref{prp:I} applies.
    Then, the claim follows from \cref{lem:hd_breakcase,lem:hd_repregion,lem:hd_approxpercase}.
\end{proof}

Finally, it is instructive to strengthen \cref{lem:qhd_analyze} by using
\cref{lem:quantum_gap_hd}.

\begin{restatable}{corollary}{qhdpostanalyze}\label{cor:qhd_postanalyze}
    Let $P$ denote a pattern of length $m$, let $T$ denote a text of length $n \le
    3/2\cdot m$, and let $k > 0$ denote an integer threshold.

    Then, there is a quantum algorithm that using $\Ohtilde(\sqrt{km})$ queries and
    quantum time,
    computes a set $C$ such that
    $\OccH_k(P,T) \subseteq C \subseteq \OccH_{10k}(P,T)$
    and
    \begin{itemize}
        \item $|C|=\Ohtilde(k)$, \textbf{or}
        \item $C$ forms an arithmetic progression.
            \qedhere
    \end{itemize}
\end{restatable}

\begin{proof}
    We begin by using \cref{lem:qhd_analyze} to identify the set \( C \) as described
    there. If \( C \) falls into the latter case (as defined in \cref{lem:qhd_analyze}),
    no further action is needed. Otherwise, we verify each \( x \in C \) by applying
    \cref{lem:quantum_gap_hd} \( \Ohtilde(1) \) times on \( P \), \( T\fragmentco{x}{x+m}
    \), and \( k \). We use majority voting: if we obtain more \yes responses than
    \no responses, we keep \( x \) in \( C \); otherwise, we discard \(x\).

    Standard concentration bounds ensure that we always keep \( x \) if \( x \in
    \OccH_k(P, T) \) and discard \( x \) if \( x \in \fragment{0}{n-m} \setminus
    \OccH_{10k}(P, T) \). This shows the algorithm’s correctness.

    As for running time, verifying each \( x \in C \) takes \( \Ohtilde(k \sqrt{m/k})  = \Ohtilde(\sqrt{km})\) time.
\end{proof}

\subsection{Step 2: From Candidate Positions to Substring Equations}
\label{sec:proxy_strings}

In the next step, we convert the set of good candidates for occurrences from the previous
section into a system of substring equations.
To that end, it is instructive to recall and summarize the structural results from
\cite{CKP19}.

\subsubsection{The Communication Complexity of Pattern Matching with Mismatches}
\label{sec:comm_comp_hd}

In this section, we present the communication complexity results for pattern matching with
mismatches from \cite{CKP19} and reformulate them in a way that is more practical for our
algorithmic applications.
To this end, we fix a threshold \( k \), consider two strings \( P \) and \(T\) whose
lengths $m$ and $n$ satisfy a mild assumption \( n \le 2m\), which is justified by the
Standard Trick.

We begin by gaining a deeper understanding of the information carried by an arbitrary
subset \( S \subseteq \OccH_k(P, T)\). To this end, we define the (mismatch) inference graph \(\bG_{S}\) for
such a subset.

\mismgs

We start with an easy but useful lemma that helps us understand the structure of inference
graphs.

\begin{lemma} \label{lem:hd_period}
    For any positive integers \(m \le n \le 2m\) and any set of integers
    $\{0,n-m\} \subseteq S \subseteq  \fragment{0}{n-m}$, write $g \coloneqq \gcd(S)$.
    \begin{itemize}
        \item The inference graph \(G_{S}\) has \(g\) connected components.
        \item For every $c \in \fragmentco{0}{g}$, vertices
            \(
                \{p_i  \mid i \equiv c \pmod{g}\} \cup \{t_i \mid i \equiv c \pmod{g}\}
            \)
            form a connected component in $G_{S}$.
    \end{itemize}
\end{lemma}
\begin{proof}
    Let us label each connected component of \(G_S\) with a unique identifier.
    For a vertex \(v \in V\), we use \(\$(v)\) to denote the label of the connected component it belongs to.
    Now, consider strings \(P\) of length \(m\) and \(T\) of length \(n\) where \[
        P\position{i} \coloneqq \$(p_i) \qquad\text{and}\qquad T\position{j} \coloneqq
        \$(t_j).
    \]
    From the definition of \(G_S\), for each \(s \in S\) and each \(i \in
    \fragmentco{0}{p}\), the vertices \(p_i\) and \(t_{i + s}\) belong to the same
    connected component, and thus \(P\position{i} = T\position{i + s}\).
    In particular, we obtain that \(S \subseteq \Occ(P, T)\), and therefore $g \coloneqq
    \gcd(S)$ is an integer multiple of $\gcd(\Occ(P,T))$.
    Our assumptions on $m$, $n$, and \(S\) allow us to use \cref{fct:periodicity} to
    conclude that $\gcd(\Occ(P,T))$ and its multiple \(g\) are periods of \(T\) and hence
    also of \(P = T\fragmentco{0}{m}\).

    In particular, we conclude that \(T\) and \(P\) in total have at most \(g\) different
    characters, which in turn implies that \(G_S\) has at most \(g\) connected components
    (as each connected component is labeled with a unique symbol).
    Further, the periodicity of \(P\) and \(T\) implies that all vertices \(p_i\) and
    \(t_j\) with \(i \equiv j \pmod{g}\) belong to the same connected component (as the
    corresponding labels of the connected components need to be equal).
    Thus, for every $c \in \fragmentco{0}{g}$, the vertices
    \(
        \{p_i  \mid i \equiv c \pmod{g}\} \cup \{t_i \mid i \equiv c \pmod{g}\}
    \)
    are indeed connected in~$G_{S}$.

    It remains to show that there is no edge between \(p_i\) and \(t_j\) for \(i \not\equiv j \pmod{g}\).
    To that end, we observe that every edge in \(G_S\) connects \(p_i\) and \(t_j\) such
    that $j = i + s$ for some $s \in S$, which in particular means that $j \equiv i
    \pmod{g}$; thereby completing the proof.
\end{proof}

Turning to inference graphs of strings, as a first observation, we immediately see that,
in the graph $\bG_S$, all characters within a black connected component are equal.
In particular, by \cref{lem:hd_period}, for storing a black component, it suffices to
store \(g = \gcd(S)\) (which is something that we store already anyway), as well as a
single character of the component.

In fact, as we are interested only in being able to check if \(P\) is at Hamming distance
at most \(k\) to some substring of \(T\) (which boils down to being able to check whether
\(P\position{i} = T\position{j}\) for given positions \(i\) and \(j\)), we do not even
need to know (and thus store) the exact character a black component is equal to.

Now, for storing red components, we intend to store only the red edges, along with the
actual characters involved. From our previous discussion of black components, we see that
storing just the red edges is enough to recover the full red components.
Indeed, we can recover exactly any character in a red component: from any such character,
there is a (potentially zero-length) path to a red edge that uses only black edges.
Since the path consists solely of black edges, all characters along the path are
identical.
By storing all characters linked to red edges, we can thus recover the character at the
beginning of the path.

In total, this motivates the following definition.

\begin{definition}\label{7.6-4}
    For a string \(P\) of length $m$ and a string \(T\) of length \(n\), and a set
    $S\subseteq \fragment{0}{n-m}$, the \emph{enhanced set of occurrences \(\eso{S}\)} is
    a set $S$ together with $\MI(P, T\fragmentco{x}{x+m})$ for every $x\in S$.
\end{definition}

From our previous discussion, we obtain the following lemma.

\begin{lemmaq}
    For a string \(P\) of length $m$ and a string \(T\) of length \(n \le 2m\), write
    \(\eso{S}\) to denote an enhanced set of occurrences.
    Then, from \(\eso{S}\), we can fully reconstruct \(\bG_S\), and in particular the set
    $C_S \subseteq \fragmentco{0}{\gcd(S)}$ of indices of the red components of \(\bG_S\).
\end{lemmaq}

\subsubsection{The Strings \(P^\#\) and \(T^\#\) as a Solution to Substring Equations}

Next, it is instructive to formally define strings that correspond to the information
stored in an enhanced set of occurrences.

\defpsh

Next, we show that the strings $P^\#$ and $T^\#$ encode enough information to adequately
compute Hamming distances.

\thmhdsubhash
\begin{proof}
    For~\eqref{it:hd_subhash:ii}, observe that if $P^\#\position{a}=T^\#\position{b}$,
    then either $p_a$ and $t_{b}$ belong to the same black connected component (and thus
    $P\position{a}=T\position{b}$) or $p_a$ and $t_{b}$ belong to red components (and thus
    $P\position{a}=P^\#\position{a}=T^\#\position{b}=T\position{b}$).
    In either case, $P^\#\position{a}=T^\#\position{b}$ implies
    $P\position{a}=T\position{b}$.

    Fix $x \in \fragment{0}{n-m}$ that is an integer multiple of $g$.
    To prove \eqref{it:hd_subhash:i}, it suffices to show that \(\MI(P\position{i},
    T\position{i + x}) = \MI(P^\#\position{i}, T^\#\position{i + x})\) holds for every $i
    \in \fragmentco{0}{m}$.
    To that end, we first observe that, since $g$ divides $x$, we have $i \equiv x + i
    \pmod{g}$.
    Hence, \cref{lem:hd_period} yields that $p_i$ and $t_{x+i}$ belong to the same
    connected component of \(\bG_S\).
    Now, if the component is red, then $P^\#\position{i} = P\position{i}$ and
    $T^\#\position{x + i} = T\position{x + i}$; yielding the claim in this case.
    Otherwise, the component is black, and thus $P\position{i} = T\position{x + i}$ and
    $P^\#\position{i} = T^\#\position{x + i}$.
    This yields $\MI(P\position{i}, T\position{i + x}) =\emptyset=
    \MI(P^\#\position{i},T^\#\position{i + x})$.
\end{proof}

\begin{remark}\label{7.6-1}
    With all the necessary components established, we are now prepared to describe how to
    encode the set $\OccH_k(P, T)$ using $\Oh(k \log n)$ space, that is, $\Oh(k \log n)$
    words in the RAM model. This encoding approach falls short by a factor of $\log n$
    compared to the single-round one-way communication complexity result for Pattern
    Matching with Mismatches from \cite{CKP19}.

    By cropping $T$ we can assume without loss of generality that $\{0, n-m\} \subseteq
    \OccH_k(P,T)$ (the case $\OccH = \empty$ can be encoded trivially in the claimed
    space).
    Given this assumption, it is not difficult to see that there is always a set of
    occurrences \(S\) such that $\{0, n-m\} \subseteq S$ and $\gcd(S) = \gcd(\OccH_k(P,
    T))$ that contains \(\Oh(\log n)\) carefully chosen occurrences.
    To demonstrate this, start with \(S = \{0, n-m\}\) and examine positions \(x\in
    \OccH_k(P, T)\) one by one.
    Either \(x\) divides \(\gcd(S)\) (in which case it is unnecessary to add $x$ to \(S\)), or
    \(x\) does not divide \(\gcd(S)\), and we add $x$ to $S$.
    In the latter case, \(\gcd(S \cup \{x\}) \leq \gcd(S) / 2\), and thus we can make at
    most \(\Oh(\log n)\) additions.

    The encoding of $\OccH_k(P, T)$ consists of \(\eso{S}\), which takes up $\Oh(k \log
    n)$ space.
    To recover $\OccH_k(P, T)$ from this encoded information, it suffices to reconstruct
    the string $P^\#$ and $T^\#$ and compute $\OccH_k(P^\#, T^\#)$. By
    \cref{thm:hd_subhash}\eqref{it:hd_subhash:ii}, we have $\OccH_k(P, T) \subseteq
    \OccH_k(P^\#, T^\#)$, and by \cref{thm:hd_subhash}\eqref{it:hd_subhash:i}, we have
    $\OccH_k(P^\#, T^\#) \subseteq\OccH_k(P, T)$. We conclude that $\OccH_k(P^\#, T^\#) =
    \OccH_k(P, T)$.
\end{remark}

Observe that \cref{7.6-1} works under the assumption that \(\OccH_k(P,T)\) is
known. However, \(\OccH_k(P,T)\) is exactly what we aim to compute!

Now, here is where we instead work with the set of candidates \(\OccH_k(P,T) \subseteq C
\subseteq \OccH_{10k}(P,T) \) from the previous section.
Using \(C\), we construct the set \(S\) (and then \(\eso{S}\)) similarly to \cref{7.6-1}, but this time by
examining positions within \(C\) instead of directly using \(\OccH_k(P,T)\).

Once we have computed a suitably large set \(\eso{S}\), we then convert the starting
positions into substring equations; which we then convert into an xSLP using
\cref{cor:solve_substring_equations}---recall that ultimately, we wish to obtain a highly
compressed representation for \(P^\#\) and \(T^\#\) that allows for fast \modelname
operations.

To that end, we first consider the following system of substring equations. Next, we then
show that \(P^\#\) and \(T^\#\) are indeed the (unique) solution (up to renaming the
placeholder characters).

\nineeighteenthree
\nineeighteenfour
\begin{proof}
    We break down the proof in two claims.
    \begin{claim}\label{claim:qhd_proxy:1}
        The strings \( (T^\#,P^\#,D) \) satisfy \( \exw{P}{T}{S} \).
    \end{claim}

    \begin{claimproof}
        For each rule added in \eqref{lem:qhd_proxy:rule:a} of the form \( e' :
        P\position{j} = D\position{a} \), we know that the vertex \( p_j \) in \( \bG_S \)
        is in a red component, which means it is not substituted in \( P^\# \), and
        therefore \(P^\#\position{j} = D\position{a} \).
        A similar argument applies to the rules of the form \( e : T\position{x+j} =
        D\position{b} \) in \eqref{lem:qhd_proxy:rule:a}.
        On the other hand, for the rules of the form \( e : P\fragmentco{y}{y'} =
        T\fragmentco{x+y}{x+y'} \) added in \eqref{lem:qhd_proxy:rule:b}, we know that,
        for each \( \hat{y} \in \fragmentco{y}{y'} \), there is a black edge between \(
        p_{\hat{y}} \) and \( t_{x+\hat{y}} \), which implies \( P^\#\position{\hat{y}} =
        T^\#\position{x+\hat{y}} \).
        Hence, $P^\#\fragmentco{y}{y'} = T^\#\fragmentco{x+y}{x+y'}$.
    \end{claimproof}

    \begin{claim}\label{claim:qhd_proxy:2}
         Let $(T',P',D')$ be strings that satisfy $\exw{P}{T}{S}$. Then, there is a
         function $\phi$ such that $T'\position{i}=\phi(T^\#\position{i})$ holds for all
         $i\in \fragmentco{0}{n}$, $P'\position{i}=\phi(P^\#\position{i})$ holds for all
         $i\in \fragmentco{0}{m}$, and $D'\position{i}=\phi(D\position{i})$ holds for all
         $i\in \fragmentco{0}{c}$.
    \end{claim}
    \begin{claimproof}
        First, we show that, if two vertices \( p_i\), \(t_j \) belong to the same black
        connected component of \(\bG_S\), then \( P'\position{i} = T'\position{j} \).
        Indeed, since \( p_i \) and \( t_j \) are connected by a path composed entirely of
        black edges in \(\bG_S\), there exists a sequence of equalities of the form
        \eqref{lem:qhd_proxy:rule:b} in \(E\) between \(P\position{i}\) and
        \(T\position{j}\).
        Given that $(T',P',D')$ satisfy \(E\), we have \( P'\position{i} = T'\position{j}
        \).
        The same argument also applies to any pair \( (p_i, p_{i'}) \) or \( (t_j, t_{j'})
        \) of vertices within the same black component.
        This means that we can associate to each black component of \(\bG_S\) a unique
        character from $(T',P',D')$ that is shared among all positions \(P'\position{i}\)
        and \(T'\position{j}\) such that the corresponding vertices \(p_{i}\) and
        \(t_{j}\) belong to the component.

        Second, we show that, if a vertex \( p_i \) belongs to a red connected component
        and \( P\position{i} = P^\#\position{i} = \sigma_a \) holds for some \( a \in
        \fragmentco{0}{c} \), then \( P'\position{i} = D'\position{a} \).
        Observe that there exists a path in \(\bG_S\) composed entirely of black edges
        (possibly of length zero) leading from $p_i$ to a node incident to a red edge,
        which must correspond to \( \sigma_a \).
        Consequently, within \( E \), there is a sequence of equalities of the form
        \eqref{lem:qhd_proxy:rule:b} between \( P\position{i} \) and \(
        T\position{\hat{\imath}} \) for some \( \hat{\imath} \in \fragmentco{0}{n} \) or
        \( P\position{\hat{\imath}} \) for some \( \hat{\imath} \in \fragmentco{0}{m}\),
        along with an equation \( T\position{\hat{\imath}} = D\position{a} \) or \(
        T\position{\hat{\imath}} = D\position{a} \), respectively, of the form
        \eqref{lem:qhd_proxy:rule:a}.
        Since \( (T',P',D') \) satisfy \( E \), we conclude that \( P'\position{i} =
        D'\position{a} \).
        The same reasoning also applies to any vertex \( t_j \) within a red component.

        We can now define the mapping \(\phi\).
        The strings $(T^\#,P^\#,D)$ contain characters of two kinds: sentinel characters,
        coming from black components, and characters from the set \(\{\sigma_0, \ldots,
        \sigma_{c-1}\}\), coming from red components or characters in \(D\).
        A character of the first type is mapped by \(\phi\) to the character in
        $(T',P',D')$ associated with the corresponding black component.
        A character  \(\sigma_a\) of the second type is mapped by \(\phi\) to
        \(D'\position{a}\).
    \end{claimproof}
    Combined, we obtain the claim.
\end{proof}

We immediately obtain the following computational consequence.

\begin{corollary}\label{lem:qhd_proxy}
    Consider a pattern $P$ of length $m$, a text $T$ of length $n\le 3/2\cdot m$, an
    integer threshold $k>0$, and a set $\{0,n-m\} \subseteq S\subseteq \OccH_k(P,T)$.
    There is a (classical) algorithm that, given $m$, $n$, $k$, and $\eso{S}$, constructs
    an xSLP of size $\Ohtilde(k|S|)$ representing strings $P^\#$ and $T^\#$ of
    \cref{def:psh}. The algorithm runs in $\Ohtilde((k|S|)^2)$ time.
\end{corollary}

\begin{proof}
    We use \cref{cor:solve_substring_equations} to solve the system \(\exw{P}{T}{S}\) from
    \cref{def:eq_system_hd}.
    \Cref{cor:solve_substring_equations} yields an xSLP that represents a universal
    solution $(T',P',D')$.
    We replace the terminals occurring in $D'$ with $\sigma_0,\ldots,\sigma_{c-1}$ to
    obtain an xSLP representing $P^\#$ and $T^\#$ as defined in \cref{def:psh}.
    Since \( E \) has size \(\Oh(k|S|)\), with \(\Oh(k)\) equations for each \( x\in S \),
    the resulting xSLP has size \(\Ohtilde(k|S|)\).

    The running time follows directly from \cref{cor:solve_substring_equations}.
\end{proof}

We proceed with the proof of \cref{thm:qhdfindproxystrings}.

\qhdfindproxystrings

\begin{proof}
    We first show how to handle the case when \cref{cor:qhd_postanalyze} gives a set $C$ such that $\{0, n-m\} \subseteq C$.

    \begin{claim}\label{claim:qhdfindproxystringscropped}
        Suppose we are given a set of candidate position $C$ as described in  \cref{cor:qhd_postanalyze} with the additional assumption that $\{0, n-m\} \subseteq C$.
        Then, we can compute an xSLP as described in the statement of \cref{thm:qhdfindproxystrings} for which additionally the two following hold:
        \begin{enumerate}[(a)]
            \item  for every $a\in \fragmentco{0}{m}$ and $b \in \fragmentco{0}{n}$, we have
            that $P^\#\position{a} = T^\#\position{b}$ implies  $P\position{a} = T\position{b}$; and
            \label{claim:qhdfindproxystringscropped:it:a}
            \item for all $x \in \OccH_k(P, T)$, we have $\MI(P, T\fragmentco{x}{x +
            m}) = \MI(P^\#, T^\#\fragmentco{x}{x + m})$.
             \label{claim:qhdfindproxystringscropped:it:b}
        \end{enumerate}
        The algorithm uses $\Ohtilde(\sqrt{km})$ queries and $\Ohtilde(\sqrt{km} + k^2)$ quantum time.
    \end{claim}

    \begin{claimproof}
        Consider the following procedure.
        \begin{enumerate}[(i)]
            \item Construct a set $S$ such that $\{0, n-m\} \subseteq S \subseteq C$ as follows:
                \begin{itemize}
                    \item If $C$ forms an arithmetic progression $0, a, 2a, \ldots, n-m$, then set $S = \{0, a, n-m\}$.
                    \item Otherwise, if $|C| = \Ohtilde(k)$, use an iterative process to construct $S$.
                        Start with \(S = \{0, n-m\}\) and examine positions \(t\in C\) one by one.
                        Add $t$ to $S$ if and only if \(t\) is not an integer multiple of $\gcd(S)$.
                \end{itemize}
                \label{alg:qhdfindproxystringscropped:i}
            \item Compute \( \eso{S} \) out of $S$.
                Using \cref{lem:find_mismatches}, first determine the set \( \Mis(P,
                T\fragmentco{x}{x+m}) \) for each \( x \in S \).
                Then, read the corresponding characters from the string to construct \( \MI(P,
                T\fragmentco{x}{x+m}) \).
                \label{alg:qhdfindproxystringscropped:ii}
            \item Use \cref{lem:qhd_proxy} on $P$, $T$, $10k$, and $S$, and return the
                corresponding xSLPs.
                \label{alg:qhdfindproxystringscropped:iii}
        \end{enumerate}

        The constructed set \( S \) satisfies two key properties: \( \gcd(S) = \gcd(C) \) and
        \( |S| = \Oh(\log n) \).
        The second property holds trivially if \( C \) forms an arithmetic progression.
        When \( |C| = \Ohtilde(k) \), the property holds because each time a new position \( t
        \) is added to \( S \), we have \( \gcd(S \cup \{t\}) \leq \gcd(S)/2 \).
        Consequently, after \( \Oh(\log n) \) additions, the gcd reduces to one, at which
        point no further positions are added to \( S \).

        Recall that the xSLP we return represents the strings \( P^\# \) and \( T^\# \) from
        \cref{def:psh}.
        Since \( \{0, n-m\} \subseteq S \), we can apply \cref{thm:hd_subhash}.
        Consequently, \eqref{claim:qhdfindproxystringscropped:it:a} follows directly from
        \cref{thm:hd_subhash}\eqref{it:hd_subhash:ii}.
        Meanwhile, \eqref{claim:qhdfindproxystringscropped:it:b} results from combining
        \eqref{it:hd_subhash:i} with the fact that every $x\in \OccH_k(P, T)\subseteq C$ is an
        integer multiple of \( \gcd(S) = \gcd(C)\).

        For the running time, note that the construction in step
        \eqref{alg:qhdfindproxystringscropped:i} requires \( \Oh(\log n) \) time and no
        queries.
        In step \eqref{alg:qhdfindproxystringscropped:ii}, applying \cref{lem:find_mismatches}
        to \( S \) (with \( |S| = \Oh(\log n) \)) and collecting the mismatch information
        requires \( \Ohtilde(\sqrt{km}) \) quantum time and queries.
        Since \( |S| = \Oh(\log n) \), step \eqref{alg:qhdfindproxystringscropped:iii}
        constructs an xSLP of size \( \Ohtilde(k) \) in \( \Ohtilde(k^2) \) time using
        \cref{lem:qhd_proxy}.
    \end{claimproof}

    Observe that \cref{claim:qhdfindproxystringscropped} would already provide the desired \(
    P^\# \) and \( T^\# \) if \( \{0, n - m\} \subseteq C \).
    Specifically,
    \cref{claim:qhdfindproxystringscropped}\eqref{claim:qhdfindproxystringscropped:it:a}
    ensures that \( \hd(P^\#, T^\#\fragmentco{x}{x+m}) \geq \hd(P, T\fragmentco{x}{x+m})
    \).
    Combined with
    \cref{claim:qhdfindproxystringscropped}\eqref{claim:qhdfindproxystringscropped:it:b},
    this implies \( \OccH_k(P, T) = \OccH_k(P^\#, T^\#) \).
    However, we cannot generally assume \( \{0, n - m\} \subseteq C \),
    and therefore we need to crop $T$.

    For that purpose, consider the following routine that extends
    \cref{claim:qhdfindproxystringscropped}
    to the general case.
    \begin{enumerate}[(i)]
        \item Use \cref{cor:qhd_postanalyze} to compute a candidate set $C$.
            If $C=\emptyset$, return an xSLP representing strings \( P^\# \) and \( T^\#
            \) of length $m$ and $n$, respectively, each with a unique sentinel character
            at every position.
            For this, we create a new xSLP with two pseudo-terminals of length $m$ and
            $n$, respectively, acting as the starting symbols for $P^\#$ and $T^\#$.
            \label{alg:qhdfindproxystrings:i}
            \item Otherwise, if $C \neq \emptyset$, set $\ell \coloneq \min C$ and $r
                \coloneq \max C$.
            Apply \cref{claim:qhdfindproxystringscropped} to \( P \), \(
            T\fragmentco{\ell}{r+m} \), and \( C - \ell \) in the role of \( P \), \( T
            \), and \( C \), respectively.
            \label{alg:qhdfindproxystrings:ii}
            \item Prepend $\ell$ and append $n-(m + r)$ unique sentinels characters to the
                proxy text $T^\#$ obtained from \cref{claim:qhdfindproxystringscropped}.
            For this, we create two new pseudo-terminals $X,Y$ of length $\ell$ and
            $n-(m+r)$, respectively, and two non-terminals $A',A''$ with $\rhs_\G(A'')=X
            A'$ and $\rhs_\G(A')=AY$, where  $A$ is the representing the original $T^\#$.
            Moreover, we declare $A''$ as the starting symbol representing the updated
            $T^\#$.
            If $\ell = 0$, we do not create $X$, and we set $A''\coloneqq A'$ instead of
            creating $A''$ as a new non-terminal;
            if $n - (m+r)=0$, we do not create $Y$, and we set $A' \coloneqq A$ instead of
            creating $A'$ as a new non-terminal.
            \label{alg:qhdfindproxystrings:iii}
    \end{enumerate}

    We briefly argue about the correctness of the routine.
    If the routine returns at \eqref{alg:qhdfindproxystrings:i}, we observe that from
    $\Occ_k(P,T) \subseteq C$ follows \( \Occ_k(P,T) = \emptyset \).
    In this case, the returned strings \( P^\# \) and \( T^\# \) satisfy \( \hd(P,
    T\fragmentco{t}{t+m}) \leq \ed(P^\#, T^\#\fragmentco{t}{t+m}) \) for all $t \in
    \fragment{0}{n-m}$, thereby fulfilling the condition \( \OccH(P^\#,T^\#) =
    \emptyset\), as desired.

    Otherwise, if the routine does not at \eqref{alg:qhdfindproxystrings:i},
    observe that $P^\#$ and $T^\#$, even after the addition of sentinel characters at
    \eqref{alg:qhdfindproxystrings:iii}, satisfy
    \cref{claim:qhdfindproxystringscropped}\eqref{claim:qhdfindproxystringscropped:it:a}.
    Consequently, $\OccH_k(P^\#, T^\#) \subseteq \OccH_k(P, T)$.
    It remains to argue that $\MI(P, T\fragmentco{x}{x+m})=\MI(P^\#,
    T^\#\fragmentco{x}{x+m})$
    holds for an arbitrary $x\in \OccH_k(P,T)$.
    This directly follows from
    \cref{claim:qhdfindproxystringscropped}\eqref{claim:qhdfindproxystringscropped:it:b},
    and from the fact that after removing $\ell$ characters from $T$ before applying
    \cref{claim:qhdfindproxystringscropped}, we prepend the same number of characters to
    the returned string.

    We obtain the claimed query and quantum time by combining
    \cref{claim:qhdfindproxystringscropped}, \cref{cor:qhd_postanalyze}, and the fact that
    all manipulations of xSLP involved in the rotuine require no more than $\Ohtilde(1)$
    time.
\end{proof}

\subsection{Step 4: Bringing It All Together}

We conclude this section with deriving \cref{thm:qpmwm} from \cref{thm:qhdfindproxystrings}.

\qpmwm

\begin{proof}
    First, note that by combining \cref{thm:qhdfindproxystrings}, \modelname
    implementation for xSLPs (\cref{lem:pillar_on_xslp}), and a \modelname algorithm for
    \PMwM (\cref{thm:hdalg}), we obtain a quantum algorithm
    solving the \PMwM problem using $\Ohtilde(\sqrt{km})$
    queries and $\Ohtilde(\sqrt{km} + k^{2})$ time when $n \leq 3/2\cdot m$.

    We employ the Standard Trick to partition \(T\) into \(\Oh(n/m)\) contiguous blocks,
    each of length \(\ceil{m/2}\) (with the last block potentially being shorter).
    For every block that starts at position $i \in \fragmentco{0}{n}$, we consider the
    segment $T\fragmentco{i}{\min(n, i + \floor{3/2 \cdot m})}$ having length at most $3/2
    \cdot m$.
    Note that every $k$-mismatch occurrence is contained in one of these segments.

    For the first claim, we iterate over all such segments, applying the quantum algorithm
    to \(P\), the current segment, and \(k\).
    The final set of occurrences is obtained by taking the union of all sets returned by
    the algorithm.
    For the second claim, \GS is employed over all segments, utilizing a
    function that determines whether the set of $k$-mismatch occurrences of $P$ in the
    segment is empty or not.
\end{proof}

%% file: s5_qpmwe.tex
\section{Quantum Algorithm for Pattern Matching with Edits}
\label{sec:pmwe}

In this section, we give our quantum algorithm for \PMwE, that
is, we prove \cref{thm:qpmwe}.

\qpmwe*
\medskip

As the main technical step, we prove \cref{thm:qedfindproxystrings},
which is the edit counterpart of \cref{thm:qhdfindproxystrings}.
That is, given a pattern \( P \) of length \( m \), a text \( T \) of length \( n \le
3/2 \cdot m \), and an integer threshold \( k > 0 \), we show how to construct an xSLP for
a proxy pattern $P^\#$ and a proxy text $T^\#$ that are equivalent to the original text and
pattern for the purpose of computing \(\OccE_k(P, T)\).

\begin{restatable*}{theorem}{qedfindproxystrings}\label{thm:qedfindproxystrings}
    There exists a quantum algorithm that, given a pattern $P$ of length $m$, a text $T$
    of length $n \le 3/2 \cdot m$, and an integer threshold $k > 0$, outputs xSLP of size
    $\Ohhat(k)$ representing strings $P^\#$ and $T^\#$ such that, for every fragment
    $T\fragmentco{t}{t'}$:
    \begin{enumerate}[(a)]
    \item if $\ed(P,T\fragmentco{t}{t'})\le k$, then $\ed(P, T\fragmentco{t}{t'}) =
        \ed(P^\#, T^\#\fragmentco{t}{t'})$ and $\sE_{P, T}(\mX) = \sE_{P^\#, T^\#}(\mX)$
        holds for every optimal alignment $\mX : P \onto T\fragmentco{t}{t'}$;
    \item $\ed(P, T\fragmentco{t}{t'}) \leq \ed(P^\#, T^\#\fragmentco{t}{t'})$.
    \end{enumerate}
    The algorithm uses $\Ohhat(\sqrt{km})$ queries and $\Ohhat(\sqrt{km}+k^2)$ time.
\end{restatable*}
\subsection{Step 1: Computing a Candidate Set for Occurrences}
\label{sec:qed_analyze}

Fortunately, for edits, we can reuse the results from \cite{KNW24}, which already provide
an efficient implementation for computing a candidate set. This candidate set
consists either of $\Ohtilde(k)$ windows of size $k$ where $k$-error occurrences may start,
or the candidate set can be
represented in $\Oh(1)$ space, with an additional guarantee on the cost of the cheapest
alignments starting in the candidate positions.

\begin{lemma}[{\cite[Lemma~4.6]{KNW24}}]\label{lem:qed_candidate_set}
    There is a quantum algorithm that, given a pattern $P$ of length $m$, a text $T$ of
    length $n$, and an integer threshold $k > 0$ such that $n < 3/2 \cdot m$, computes a
    set $C$ of \emph{candidate positions} such that $\OccE_k(P,T) \subseteq C \subseteq
    \fragment{0}{n}$, and outputs one of two compressed representations of $C$:
    \begin{itemize}
        \item either it outputs $\floor{C/k}$ of size $|\floor{C/k}|=\Ohtilde(k)$, \textbf{or}
        \item it outputs $q \in \Z_{>0}$ and $I \subseteq \fragmentco{0}{q}$ of size $|I|
            = \Oh(k)$ such that $C = \{ t \in \fragment{0}{n-m-k} : t \ \mathrm{mod} \ q
            \in I\}$
        and $\OccE_k(P,T) \subseteq C \subseteq \OccE_{44k}(P,T)$.
    \end{itemize}
    The algorithm takes $\Ohtilde(\sqrt{km})$ query complexity and $\Ohtilde(\sqrt{km}+k^2)$ quantum time.
    \lipicsEnd
\end{lemma}

\subsection{Step 2: From Candidate Positions to Substring Equations}
\label{sec:proxy_strings_e}

As for \PMwM, we next discuss how to turn the candidates
from the previous section into a system of substring equations.
To that end, it is instructive, to present the communication complexity results for
\PMwE from \cite{KNW24}.

\subsubsection{The Communication Complexity of Pattern Matching with Edits}

Compared to \PMwM, this time the results of \cite{KNW24} are already in a form that is
useful for our algorithmic applications.
Hence, we can restrict ourselves to presenting the key definitions and lemmas from
\cite{KNW24} without a proof.

In this (sub)section, we fix an threshold \( k > 0\), and consider two strings \( P \) and
\(T\) of length $n$ and $m$.
Furthermore, we consider a set \( S \subseteq \mA^{\leq k}_{P,T}\) of alignments
of $P$ onto fragments of $T$ of cost at most $k$.
In the $k$-edit setting, \( S \) still allows the construction of a graph \( \bG_{S} \).

\begin{definition}[{\cite[Definition~4.1]{KNW24}}]\label{def:bg}
    For two strings $P$ and $T$ of length $m$ and $n$, respectively,
    and a set $S \subseteq \mA_{P,T}$,
    we define the undirected graph $\bG_{S} = (V, E)$ as follows.
    \begin{description}
        \item[The vertex set $V$]contains
        \begin{itemize}
            \item $|P|$ vertices representing characters of $P$;
            \item $|T|$ vertices representing characters of $T$; and
            \item one special vertex $\bot$.
        \end{itemize}
        \item[The edge set $E$]contains the following edges for each alignment $\mX\in S$
        \begin{enumerate}
            \item $\{P\position{x},\bot\}$ for every character $P\position{x}$ that $\mX$
                deletes;\label{it:bg:i}
            \item $\{\bot, T\position{y}\}$ for every character $T\position{y}$ that $\mX$
                inserts;\label{it:bg:ii}
            \item $\{P\position{x},T\position{y}\}$ for every pair of characters
                $P\position{x}$ and $T\position{y}$ that $\mX$ aligns.\label{it:bg:iii}
        \end{enumerate}
        We say that an edge $\{P\position{x},T\position{y}\}$ is \emph{black} if $\mX$
        matches $P\position{x}$ and $T\position{y}$;
        all the remaining edges are \emph{red}.
    \end{description}

    We say that a connected component of $\bG_{S}$ is \emph{red} if it contains at least
    one red edge; otherwise, we say that the connected component is \emph{black}.
    We denote with $\bc(\bG_{S})$ the number of black components in $\bG_{S}$.
\end{definition}

The graph $\bG_S$ can be fully constructed using the \emph{complete edit information}.

\begin{definition}\label{def:eso_ed}
    For two strings $P,T$ and a set $S \subseteq \mA_{P,T}$,
    the \emph{complete edit information} $\eso{S}$
    is the collection of all $\sE_{P,T}(\mX)$ such that $\mX \in S$.
\end{definition}

Assuming \( S \subseteq \mA^{\leq k}_{P,T}\), we observe that the complete edit
information
$\eso{S}$ takes up $\Oh(k|S|)$ space. Moreover, similarly to mismatches, we have that
$\eso{S}$ allows us to infer all characters of the red component by propagating characters
through them. Characters of black component remain unknown, with the guarantee however
that characters belonging to the same black component are the same.

As in the Hamming case, if \(\bc(\bG_{S}) = 0\), we already have enough information to
retrieve \(\OccE_k(P, T)\) by fully reconstructing \( P \) and \( T \).
The more challenging case, \(\bc(\bG_{S}) > 0\), presents some differences compared to the
Hamming case. Firstly, red components are much less structured. Black components, however,
still retain a periodic structure, under mild structural assumptions on the length of $P$
and $T$, and on $S$.

\begin{definition}[{\cite[Definition~4.2]{KNW24}}]
    Consider a string $P$ of length $m$, a string $T$ of length $n$, and a set $S
    \subseteq \mA^{\leq k}_{P,T}$.
    We say $S$ \emph{encloses $T$} if $n \le 2m - 2k$ and there exist two distinct
    alignments $\mXpref, \mXsuf \in S$ such that $\mXpref$ aligns $P$ with a prefix of $T$
    and $\mXsuf$ aligns $P$ with a suffix of $T$, or equivalently $(0,0) \in \mXpref$ and
    $(|P|,|T|) \in \mXsuf$.
\end{definition}

The periodic structure can be observed in the strings \( T_{|S} \) and \( P_{|S} \), which
are constructed by retaining characters contained in black connected components.

\begin{definition}[{\cite[Definition~4.3]{KNW24}}]
    Consider a string $P$ of length $m$, a string $T$ of length $n$, and a set $S
    \subseteq \mA^{\leq k}_{P,T}$ such that $\bc(\bG_S) \neq 0$.
    We let $T_{|S}$ denote the subsequence of $T$ consisting of characters contained in
    black components of~$\bG_{S}$.
    Similarly, $P_{|S}$ denotes the subsequence of $P$ consisting of characters contained
    in black components of~$\bG_{S}$.
\end{definition}

\begin{lemma}[{\cite[Lemma~4.4]{KNW24}}]\label{lem:periodicity}
    Consider a string $P$ of length $m$, a string $T$ of length $n$, and a set $S
    \subseteq \mA^{\leq k}_{P,T}$ such that $\bc(\bG_S) \neq 0$ and $S$ encloses $T$.
    For every $c \in \fragmentco{0}{\bc(\bG_S)}$, there exists a black connected component
    with node set
    \[
        \{P_{|S}\position{i} : i \equiv_{\bc(\bG_S)} c\} \cup \{T_{|S}\position{i} : i
        \equiv_{\bc(\bG_S)} c\},
    \]
    that is, there exists a black connected component containing all characters of
    $P_{|S}$ and $T_{|S}$ appearing at positions congruent to $c$ modulo $\bc(\bG_S)$.
    Moreover, the last characters of $P_{|S}$ and $T_{|S}$ are contained in the same black
    connected component, that is, $|T_{|S}| \equiv_{\bc(\bG_{S})} |P_{|S}|$.
    \lipicsEnd
\end{lemma}

The characters of \( P_{|S} \) and \( T_{|S} \) can be mapped back to \( P \) and \( T \).

\begin{definition}[{\cite[Definition~4.7]{KNW24}}]\label{def:pitau}
    Consider a string $P$ of length $m$, a string $T$ of length $n$, and a set $S
    \subseteq \mA^{\leq k}_{P,T}$ such that $\bc(\bG_S) \neq 0$ and $S$ encloses $T$.
    For $c \in \fragmentco{0}{\bc(\bG_S)}$, define the \emph{$c$-th black connected
    component} as the black connected component containing $P_{|S}\position{c}$ and set
    \[
        m_c \coloneqq \left\lceil\frac{|P_{|S}| - c}{\bc(\bG_{S})} \right\rceil
        \quad\text{and}\quad n_c \coloneqq \left\lceil \frac{|T_{|S}| - c}{\bc(\bG_{S})} \right\rceil,
    \]
    as the number of characters in $P$ and $T$, respectively, belonging to the $c$-th
    black connected component.
    Furthermore,
    \begin{itemize}
        \item for $c \in \fragmentco{0}{\bc(\bG_{S})}$ and $j \in \fragmentco{0}{m_c}$,
            define $\pi_j^c \in \fragmentco{0}{|P|}$ as the position of $P_{|S}\position{c
            + j \cdot \bc(\bG_{S})}$ in $P$; and
        \item for $c \in \fragmentco{0}{\bc(\bG_{S})}$ and $i \in \fragmentco{0}{n_c}$,
            define $\tau_i^c \in \fragmentco{0}{|T|}$ as the position of
            $T_{|S}\position{c + i \cdot \bc(\bG_{S})}$ in $T$. \qedhere
    \end{itemize}
\end{definition}

Observe that we have \( m_c \in \{m_0, m_0-1\} \) and \( n_c \in \{n_0, n_0-1\} \).
Moreover, observe that the characterization of the black components implies that \(
\pi_{j}^c < \pi_{j'}^{c'} \) if and only if either \( j < j' \) or \( j = j' \) and \( c <
c' \) (analogously for \( \tau_j^c < \tau_{j'}^{c'} \)).

\subsubsection{The Strings \(P^\S\) and \(T^\S\) as a Solution to Substring Equations}

Ideally, similar to the Hamming case, we would like to maintain a set $S \subseteq
\mA_{P,T}$ with low cost, ensuring that the information in $\eso{S}$ is sufficient to
recover all $k$-edit occurrences by replacing characters in black components  with
sentinel characters.
However, this approach does not work for edits.
In the presence of insertions and deletions, we lose the guarantee that characters in
black components align only with other characters in black components within any $k$-edit
alignment.

Although a replacement criterion as the one used for Hamming distance does not work here,
we still define the corresponding strings. However, we use a different name
to highlight that they are not the final strings that we compute.

\begin{definition}\label{def:sub_hash_ed_bad}
    For a string \(P\) of length $m$, a string \(T\) of length \(n \le 2m\), and a fixed
    set $S\subseteq \fragment{0}{n-m}$, we transform \(P\) and \(T\) into $P^\S$ and
    $T^\S$ by iterating over black connected component of $\bG_S$ and placing a sentinel
    character (unique to the component) at every position that belongs to the component.
\end{definition}

The real purpose of the strings \( P^\S \) and \( T^\S \) is to help us to find the positions \( \tau_i^c \) and \( \pi_j^c \).
Such positions are easily identifiable in \( P^\S \) and \( T^\S \),
as these positions are exactly those positions that contain sentinel characters.
For answering fast queries related to \( \tau_i^c \) and \( \pi_j^c \), it is key to
represent \( P^\S \) and \( T^\S \) through substring equations and converting them into
an xSLP.

\begin{definition}\label{def:eq_system_ed_1}
    Consider a string \(P\) of length \(m\), a string \(T\) of length \(n \le 3/2 \cdot m\),
    and a set $S \subseteq \mA_{P,T}^{\leq k}$.
    Further, write \(D \coloneqq \sigma_0 \cdots \sigma_{c-1}\) for a string of all of the \(c\) different
    characters appearing across $\eso{S}$.

    We write \(\exw{P}{T}{S}\) for the set
    of substring equalities constructed in the following manner.
    For each $\mX \in S$, we add the following substring equations to $E$.
    \begin{enumerate}[(a)]
        \item For each \((i, \sigma_a \mid j, \sigma_b) \in \sE_{P, T}(\mX)\), we add
        \begin{itemize}
            \item $e : P\position{i} = D\position{a}$ if the corresponding edit is a deletion;
            \item $e : T\position{j} = D\position{b}$ if the corresponding edit is an insertion;
            \item $e : P\position{i} = D\position{a}$ and $e' : T\position{j} = D\position{b}$ otherwise.
        \end{itemize}
        \item For each maximal interval $\fragmentco{y}{y'} \subseteq \fragmentco{0}{m}
            \setminus \{i \mid (i,\sigma_a \mid j,\sigma_b) \in \sE_{P, T}(\mX_i)\}$, we
            add $e : P\fragmentco{y}{y'} = T\fragmentco{x}{x'}$ where $x,x'$ are such that
            $\mX$ matches $P\fragmentco{y}{y'}$ with $T\fragmentco{x}{x'}$.
        \qedhere
    \end{enumerate}
\end{definition}

\begin{lemma}  \label{lem:str_eq_to_xslp_ed_1}
    Consider a string \(P\) of length \(m\), a string \(T\) of length \(n \le 3/2 \cdot m\),
    and a set of complete edit information \(\eso{S}\).
    Further, write \(D \coloneqq \sigma_0 \cdots \sigma_{c-1}\) for a string of all of the \(c\) different
    characters appearing in \(\eso{S}\).

    The strings \( (T^\S,P^\S,D) \) are the unique solution of \( \exw{P}{T}{S} \) (up to
    renaming of placeholder characters).
\end{lemma}

\begin{proof}
    The same proof as in \cref{lem:str_eq_to_xslp_hd} applies as there is a one to one
    correspondence between black edges in $\bG_S$ and equalities at a character level
    introduced by equations of the form \eqref{lem:qhd_proxy:rule:b} in \cref{def:eq_system_ed_1},
    and between nodes incident to red components in $\bG_S$ and equalities introduced by
    equations of the form \eqref{lem:qhd_proxy:rule:a} in \cref{def:eq_system_ed_1}.
\end{proof}

\subsubsection{The Strings \(P^\#\) and \(T^\#\) as a Solution to Substring Equations}
\label{sec:pt_hash_ed}

As a workaround to the previous discussion,
we could learn some of the characters within black components as well,
resulting in less masked out characters than in $P^\S $ and $T^\S $.
The authors of \cite{KNW24} proved that storing a small number of these characters is
indeed sufficient.

In \cite{KNW24}, the first step in determining which characters to learn involves finding
a function \( \w_S :\fragmentco{0}{\bc(\bG_S)} \rightarrow \Zz \)
that assigns weights to black components.
On a high level, this function should have the property to estimate whether edits in \( S
\) are somehow ``close'' to a black component.
If a function \( \w_S \) has such desirable properties, it is said to \emph{cover \( S
\)}.
The specific properties required for this function are not relevant to this summary of the
results from \cite{KNW24}.
However, what is important is the \emph{total weight} \( w = \sum_{c=0}^{\bc(\bG_S)-1}
\w_S(c) \),
as the number of characters we need to learn is directly proportional to $w$.
The authors of \cite{KNW24} provide a construction of such a function of total weight
equal to the sum of the costs of the alignment in $S$.

Based on \( \w_S \), the authors identify a subset \( C_S \subseteq \mA^{\leq k}_{P,T} \),
referred to as a \emph{black cover with respect to $\w_S$},
which determines the specific characters that need to be learned.
This leads us to the formal definition of the proxy strings \( P^\# \) and \( T^\# \) for edits.

\begin{definition}\label{def:pt_hash_ed}
    For a string \(P\) of length $m$, a string \(T\) of length \(n \le 2m\),
    and a fixed set $S \subseteq \mA_{P,T} $,
    we transform \(P\) and \(T\) into $P^\#$ and $T^\#$ by iterating over each black
    connected component $c  \notin C_S$ of $\bG_S$ and placing a sentinel character
    (unique to the component) at every position that belongs to the component.
\end{definition}

A black cover \( C_S \) with respect to a function \( \w_S \),
which covers \( S \) with total weight \( w \),
can be constructed such that $\{(c, T\position{\tau_0^c}):c\in C_S\}$ is encodable in \( \Oh(w +
k|S|) \) space,
provided that construction time is not a concern—such as in the context of communication
complexity.
In \cite{KNW24}, the authors also present a black cover such that $\{(c, T\position{\tau_0^c}):c\in
C_S\}$
can be encoded in \( \Ohtilde(w + k|S|) \) space,
with efficient construction in the quantum setting; we elaborate on this in
\cref{sec:g_s_constr_ed}.

What is left to describe is the condition that must hold for a \( k \)-error occurrence of
\( P \) in \( T \) to ensure it is preserved in \( P^\# \) and \( T^\# \). In the Hamming
case, a \( k \)-mismatch occurrence \( x \in \OccH_K(P, T) \) is \emph{captured} if \( x
\) is divisible by \( \gcd(S) \).

The following definition formalizes when a $k$-edit occurrence is captured.

\begin{definition}[{\cite[Definition~4.28]{KNW24}}] \label{def:scomplete}
    Consider a string $P$ of length $m$, a string $T$ of length $n$, and a set $S
    \subseteq \mA^{\leq k}_{P,T}$ such that $\bc(\bG_S) \neq 0$ and $S$ encloses $T$.
    Further, let $\w_S$ cover $S$ and $C_S$ be a black cover with respect to $\w_S$.
    For a $k$-error occurrence $T\fragmentco{t}{t'}$ of $P$ in $T$,
    we say that $S$ \emph{captures} $T\fragmentco{t}{t'}$ if exactly one of the following
    two conditions holds:
    \begin{itemize}
        \item $\bc(\bG_{S}) = 0$; or
        \item $\bc(\bG_{S}) > 0$ and $|\tau_i^{0} - t - \pi_0^{0}| \leq w + 3k$
            holds for some $i\in \fragmentco{0}{n_0}$.\qedhere
    \end{itemize}
\end{definition}

The capturing property must also ensure that if it fails to hold for all \( k \)-error
occurrences,
we can identify an uncaptured occurrence, add it to \( S \),
and thereby make progress toward having \( S \) capture all \( k \)-error occurrences.
Ideally, no more than \( \Oh(\log n) \) alignments should need to be added to \( S \) to
accomplish this.
In the case of Hamming distance, we had the guarantee that the number of components in \(
\bG_S \) would decrease by at least a factor of two.
For edits, the same holds, but for the number of black components.

\begin{lemma}[{\cite[Lemma~4.27]{KNW24}}]
    \label{lem:periodhalves}
    Consider a string $P$ of length $m$, a string $T$ of length $n$, and a set $S
    \subseteq \mA^{\leq k}_{P,T}$ such that $\bc(\bG_S) \neq 0$ and $S$ encloses $T$.
    Further, let $\w_S$ cover $S$ and $C_S$ be a black cover with respect to $\w_S$.
    Let $\mY : P \onto T\fragmentco{t}{t'}$ be an alignment of cost at most $k$.
    If $|\tau_{i}^{0} - t - \pi_0^{0}|>w+2k$ holds for every $i\in \fragment{0}{n_0-m_0}$,
    then $\bc(\bG_{S \cup \{\mY\}}) \leq \bc(\bG_{S})/2$. \lipicsEnd
\end{lemma}

Note that this is still sufficient for our purposes.
Once all black components have been eliminated, only the red components remain,
allowing us to fully reconstruct \( P \) and \( T \) using only the information in \(
\eso{S} \).

For a summary of the previous discussion, let us restate \cite[Theorem 4.27]{KNW24}.

\begin{theorem}[{\cite[Theorem~4.28]{KNW24}}] \label{prp:ed_subhash}
    Consider a string $P$ of length $m$, a string $T$ of length $n$, and a set $S
    \subseteq \mA^{\leq k}_{P,T}$ such that $\bc(\bG_S) \neq 0$ and $S$ encloses $T$.
    Further, let $\w_S$ cover $S$ and $C_S$ be a black cover with respect to $\w_S$.
    Then, the two following hold.
    \begin{enumerate}[(i)]
        \item For every $a\in \fragmentco{0}{m}$ and $b \in \fragmentco{0}{n}$, we have
            that $P^\#\position{a} = T^\#\position{b}$ implies  $P\position{a} =
            T\position{b}$.\\
            (No new equalities between characters are created.)
            \label{it:ed_subhash:i}
        \item If $S$ captures the k-error occurrence $T\fragmentco{t}{t'}$, then $\ed(P,
            T\fragmentco{t}{t'}) \leq \ed(P^\#, T^\#\fragmentco{t}{t'})$.
            Moreover, for all optimal $\mX \mid P \onto T\fragmentco{t}{t'}$ we have
            $\sE_{P, T}(\mX) = \sE_{P^\#, T^\#}(\mX)$.\\
            (For captured $k$-error occurrences, the edit information and the edit
            distance are preserved.)
            \label{it:ed_subhash:ii}
            \lipicsEnd
    \end{enumerate}
\end{theorem}

\begin{remark}
     We describe how to obtain the single-round one-way communication complexity result
     for \PMwE from \cite{KNW24}. That is, we briefly sketch how to encode the set
     $\OccE_k(P, T)$ using $\Oh(k \log n)$ space, that is, $\Oh(k \log n)$ words in the RAM
     model.

     We proceed very similarly to the Hamming case.
     By cropping $T$ we can assume without loss of generality that there exist $k$-edit
     alignments $\mXpref, \mXsuf$ as described in the introduction of this (sub)section
     (the case $\OccE_k(P,T) = \emptyset$ can be encoded trivially in the claimed space).

     Our goal is to construct a set $S$ of size $|S| = \Oh(\log n)$ that captures all
     $k$-error occurrences. We begin with \(S \coloneq \{\mXpref, \mXsuf\}\) and examine
     one by one each optimal alignment \(\mY : P \onto T\fragmentco{t}{t'}\) of cost at
     most $k$. If \(\bc(\bG_{S}) = 0\), then $T\fragmentco{t}{t'}$ is already captured, and
     there is no need to add $\mY$ to $S$. Otherwise, we construct the function $\w_S$ from
     \cite{KNW24} of total weight $w$ equal to the sum of costs of the alignments in $S$.

     This gives us the guarantee $w \leq k|S|$. Next, we consider $\Delta_S(t) \coloneqq
     \min_{i \in \fragmentco{0}{n_0}} |\tau_i^0 - t - \pi_0^0|$. If $\Delta_S(t) \leq w +
     3k$, then $T\fragmentco{t}{t'}$ is already captured as well.

     This leaves the case $\Delta_S(t) > w + 3k$ in which we add \(\mY\) to \(S\). By
     \cref{lem:periodhalves}, we have \(\bc(\bG_{S \cup \{\mY\}}) \leq \bc(\bG_{S})/2\),
     meaning we can make at most \(\Oh(\log n)\) additions before \(\bc(\bG_{S}) = 0\) and
     all $k$-error occurrences are captured.

     The encoding of $\OccE_k(P, T)$ depends on whetner \(\bc(\bG_{S}) = 0\) or
     \(\bc(\bG_{S}) > 0\).

     If the former holds, then it suffices to store $\sE_{P,T}(\mX)$ for all $\mX \in S$.
     By the discussion of above, this suffices to fully reconstruct $P,T$ through $\bG_S$.

     If the latter holds, then we encode, along with the edit information the black cover
     $C_S$ that takes up $\Oh(k|S|) = \Oh(k\log n)$ space. This allows us to reconstruct
     the string $P^\#$ and $T^\#$ and compute $\OccE_k(P^\#, T^\#)$. By
     \cref{prp:ed_subhash}\eqref{it:ed_subhash:i}, we have $\OccE_k(P, T) \subseteq
     \OccE_k(P^\#, T^\#)$, and by \cref{prp:ed_subhash}\eqref{it:ed_subhash:ii}, we have
     $\OccE_k(P^\#, T^\#) \subseteq\OccE_k(P, T)$. We conclude that $\OccE_k(P^\#, T^\#) =
     \OccE_k(P, T)$.
 \end{remark}

We conclude \cref{sec:pt_hash_ed} by formalizing the encoding form of the black cover
that the authors of \cite{KNW24} provide,
and by giving the construction of the system of substring equations
whose solution represent strings $P^\#$ and $T^\#$ of \cref{def:pt_hash_ed}.

\begin{definition}
    Consider a string $P$ of length $m$, a string $T$ of length $n$, and a set $S
    \subseteq \mA^{\leq k}_{P,T}$ such that $\bc(\bG_S) \neq 0$ and $S$ encloses $T$.
    Further, let $\w_S$ cover $S$ and $C_S$ be a black cover with respect to $\w_S$.
    We say an encoding of the set $\{(c, T\position{\tau_0^c}):c\in C_S\}$ is \emph{a Lempel--Ziv
    compression of $C_S$},
    if it consists of $\mathsf{sz} = \Ohtilde(w)$ intervals
    $\fragment{a_1}{b_1}, \ldots, \fragment{a_\mathsf{sz}}{b_\mathsf{sz}} \subseteq
    \fragmentco{0}{\bc(\bG_S)}$
    such that $C_S = \bigcup_{i=1}^{\mathsf{sz}} \fragment{a_i}{b_i}$.
    Additionally, for each $i \in \fragment{1}{\mathsf{sz}}$ either
    \begin{center}
        \quad $\LZ(T\fragment{\tau_0^{a_i}}{\tau_0^{b_i}})$ \quad or \quad
        $\LZ(\rev{T\fragment{\tau_0^{a_i}}{\tau_0^{b_i}}})$
    \end{center}
    is provided. For such factorizations we have
    \[
        \sum_{i=1}^{\mathsf{sz}} (1 - \mathrm{I}_i) \cdot
        |\LZ(T\fragment{\tau_0^{a_i}}{\tau_0^{b_i}})|
        + \mathrm{I}_i \cdot \LZ(\rev{T\fragment{\tau_0^{a_i}}{\tau_0^{b_i}}})
        = \Ohtilde(w + k|S|),
    \]
    where $\mathrm{I}_i \in \{0,1\}$ indicates whether the compressed representation along
    with the $i$-th interval is reversed.
\end{definition}

\begin{definition}\label{def:psh_ed}
    Consider a string $P$ of length $m$, a string $T$ of length $n$, and a set $S
    \subseteq \mA^{\leq k}_{P,T}$ such that $\bc(\bG_S) \neq 0$ and $S$ encloses $T$.
    Further, let $\w_S$ cover $S$, let $C_S$ be a black cover with respect to $\w_S$, and let $R$
    be a Lempel--Ziv compression of $C_S$.

    Write \(D = \sigma_0 \cdots \sigma_{c-1}\) for a string of all of the \(c\)
    characters appearing across $\eso{S}$ and across all parses in $R$
    such that all characters present in $\eso{S}$ appear in a prefix of $D$.
    We denote by $\exwf{P}{T}{S}{R}$ the set of substring equations obtained by adding to
    $\exw{P}{T}{S}$
    of \cref{def:eq_system_ed_1} the following equation for each phrase $F_j$ in a
    LZ77-like representation contained in $R$.
    \begin{enumerate}[(a)]
        \item If $F_j=T\fragmentco{i}{i+\ell}$
            is a previous factor phrase encoded as $(i',\ell)$, we add $e :
            T\fragmentco{i}{i+\ell} = T\fragmentco{i'}{i'+\ell}$.
            \label{eq:psh_ed:a}
        \item If $F_j=\rev{T\fragmentco{i}{i+\ell}}$
            is a previous factor phrase encoded as $(i',\ell)$, we add $e :
            T\fragmentco{i}{i+\ell} = T\fragmentoc{i'-\ell}{i'}$.
            \label{eq:psh_ed:b}
        \item Otherwise, if $F_j=T\position{i}$ is encoded as $(\sigma_{a}',0)$, we add $e
            : T\position{i} = D'\position{a}$.
            \label{eq:psh_ed:c}
            \qedhere
    \end{enumerate}
\end{definition}

\begin{lemma}  \label{lem:str_eq_to_xslp_ed_2}
    Consider a string $P$ of length $m$, a string $T$ of length $n$, and a set $S
    \subseteq \mA^{\leq k}_{P,T}$ such that $\bc(\bG_S) \neq 0$ and $S$ encloses $T$.
    Further, let $\w_S$ cover $S$, let $C_S$ be a black cover with respect to $\w_S$, and
    let $R$
    be a Lempel--Ziv compression of $C_S$.
    Write \(D = \sigma_1 \cdots \sigma_c\) for a string of all of the \(c\)
    characters appearing across $\eso{S}$ and across all parses in $R$
    such that all characters present in $\eso{S}$ appear in a prefix of $D$.

    The strings \( (T^\#,P^\#,D) \) are the unique solution of \( \exwf{P}{T}{S}{R} \) (up
    to
    renaming of placeholder characters).
\end{lemma}

\begin{proof}\label{claim:qed_proxy:1}
    We break down the proof in two claims.
    \begin{claim}
        \( (T^\#,P^\#,D) \) satisfy \( \exwf{P}{T}{S}{R} \).
    \end{claim}
    \begin{claimproof}
        First, we prove that \( (T^\#,P^\#,D) \) satisfies $\exw{P}{T}{S}$ of
        \cref{def:eq_system_ed_1}.
        To this end, define the concatenations $W^\# = T^\# P^\# D$ and $W^\S = T^\S P^\S
        D$.
        Observe that for all $i,j \in \fragmentco{0}{n+m+c}$, we have that
        $W^\#\position{i}=W^\#\position{j}$ implies $W^\S\position{i}=W^\S\position{j}$.
        Consequently, every set of substring equations satisfied by \( (T^\S,P^\S,D) \)
        is also satisfied by \( (T^\#,P^\#,D) \).
        From \cref{lem:str_eq_to_xslp_ed_1} follows that
        \( (T^\#,P^\#,D) \) satisfies $\exw{P}{T}{S}$.

        Next, we demonstrate that \( (T^\#, P^\#, D) \) also satisfies the equations
        \eqref{eq:psh_ed:a}, \eqref{eq:psh_ed:b}, and \eqref{eq:psh_ed:c}
        introduced in \cref{def:psh_ed}.
        It is important to note that all these equations are satisfied by \( T \),
        as they come from a well-formed LZ77-like representation of fragments of \( T \).
        To complete the proof, it suffices to show that for every fragment \(
        T\fragmentco{x}{y} \) that appears in any of the equations \eqref{eq:psh_ed:a},
        \eqref{eq:psh_ed:b}, and \eqref{eq:psh_ed:c},
        we have \( T\fragmentco{x}{y} = T^\#\fragmentco{x}{y} \).
        Indeed, for each such \( x \) and \( y \), we know that \( \fragmentco{x}{y}
        \subseteq \fragment{\tau_0^a}{\tau_0^b} \) for some \( a, b \in
        \fragmentco{0}{\bc(\bG_S)} \),
        where \( \fragment{a}{b} \) corresponds to one of the intervals stored in \( R \).

        Crucially, every character in \( T\fragmentco{\tau_0^a}{\tau_0^b} \)
        remains unchanged when transforming \(T\) into \( T^\# \),
        as it is either part of a red component or, if part of a black component, belongs
        to \( C_S \).

        Consequently, we have \( T\fragment{\tau_0^a}{\tau_0^b} =
        T^\#\fragment{\tau_0^a}{\tau_0^b}\),
        which implies \( T\fragmentco{x}{y} = T^\#\fragmentco{x}{y} \).
    \end{claimproof}

    \begin{claim}\label{claim:qed_proxy:2}
        Let $(T',P',D')$ be strings that satisfy $\exwf{P}{T}{S}{R}$. Then, there is a
        function $\phi$ such that $T'\position{i}=\phi(T^\#\position{i})$ holds for all
        $i\in \fragmentco{0}{n}$, $P'\position{i}=\phi(P^\#\position{i})$ holds for all
        $i\in \fragmentco{0}{m}$, and $D'\position{i}=\phi(D\position{i})$ holds for all
        $i\in \fragmentco{0}{c}$.
   \end{claim}
   \begin{claimproof}
       Let \( c_p \in \fragment{0}{c} \) be the smallest index such that the prefix \( D_p
       \coloneqq D\fragmentco{0}{c_p} \) contains all characters present in \( \eso{S} \).
       Define \( \#_d \) as the sentinel character substituted in the \( d \)-th black
       connected component for \( d \in \fragmentco{0}{\bc(\bG_S)} \setminus C_S \), used
       in \( P^\# \) and \( T^\# \).

       Also, define \( P^\S \) and \( T^\S \) as in \cref{def:sub_hash_ed_bad},
       substituting the same sentinel character \( \#_d \) for \( d \in
       \fragmentco{0}{\bc(\bG_S)} \setminus C_S \), and new sentinel characters \( \#_d \)
       for \( d \in C_S \).

       Next, observe that \( (T', P', D'\fragmentco{0}{c_p}) \) satisfy the subset of
       equations \( \exw{P}{T}{S} \subseteq \exwf{P}{T}{S}{R} \).

       From \cref{lem:str_eq_to_xslp_ed_1}, it follows that there exists \( \psi \) such
       that \( T'\position{i} = \psi(T^\S\position{i}) \) for all \( i \in
       \fragmentco{0}{n} \), \( P'\position{i} = \psi(P^\S\position{i}) \) for all \( i
       \in \fragmentco{0}{m} \), and \( D'\position{i} = \psi(D\position{i}) \) for all \(
       i \in \fragmentco{0}{c_p} \).

       Now, we construct \( \phi \) from \( \psi \) by defining:
       \[
           \phi(D\position{a}) = D'\position{a} \text{ for } a \in \fragmentco{0}{c},
           \quad \text{and} \quad \phi(\#_d) = \psi(\#_d) \text{ for } d \in
           \fragmentco{0}{\bc(\bG_S)} \setminus C_S.
       \]

       By the definition of \( \phi \), it is clear that \( D'\position{i} =
       \psi(D\position{i}) \) holds for all \( i \in \fragmentco{0}{c} \).
       To show that the same holds for characters in \( P \), consider an arbitrary \( i
       \in \fragmentco{0}{m} \). We distinguish two cases:
       \begin{itemize}
           \item If \( P\position{i} \) is in a red component, then $P^\S\position{i} =
               P^\#\position{i} = D\position{a} \text{ for some } a \in
               \fragmentco{0}{\bc(\bG_S)} \text{ such that } P\position{i} =
               D\position{a}$.
               Consequently, \( \phi(P^\#\position{i}) = \phi(D\position{a}) =
               D'\position{a} = \psi(D\position{a}) = \psi(P^\S\position{i}) =
               P'\position{i} \).

        \item Otherwise, if \( P\position{i} \) is in a black component:
            \begin{itemize}
                \item If \( d \notin C_S \), then we immediately have $\phi(P^\#\position{i})
                    = \phi(\#_d) = \psi(\#_d) = P'\position{i}$.
                \item If \( d \in C_S \), then \( P^\#\position{i} = D\position{a} \) for
                    some \( a \in \fragmentco{0}{c} \). To show that \( P'\position{i} =
                    D'\position{a} \), observe that there is a path of black edges in \(
                    \bG_S \) from \( P\position{i} \) to \( T\position{\tau_0^d} \), and
                    an LZ77-like encoding in \( R \) encoding \( T\position{\tau_0^d} \).
                    The black edges imply a chain of equalities between \( P\position{i}
                    \) and \( T\position{\tau_0^d} \) in \( \exwf{P}{T}{S}{R} \), and the
                    LZ77 encoding implies a chain of equalities from \(
                    T\position{\tau_0^d} \) to \( D\position{a} \). As \( (T', P', D') \)
                    satisfies \( \exwf{P}{T}{S}{R} \), we have \( P'\position{i} =
                    T'\position{\tau_0^d} = D'\position{a} \).
            \end{itemize}
        \end{itemize}
        As \( i \) was arbitrary, we conclude that \( P'\position{i} =
        \psi(P^\#\position{i}) \) holds for all \( i \in \fragmentco{0}{m} \). A symmetric
        argument shows that \( T'\position{i} = \psi(T^\#\position{i}) \) holds for all \(
        i \in \fragmentco{0}{n} \).
    \end{claimproof}
    This concludes the proof of \cref{lem:str_eq_to_xslp_ed_2}.
\end{proof}

By using the same approach as in \cref{lem:qhd_proxy}, we obtain the following corollary.

\begin{corollary}\label{lem:qed_proxy_2}
    Consider a string $P$ of length $m$, a string $T$ of length $n$, and a set $S
    \subseteq \mA^{\leq k}_{P,T}$ such that $\bc(\bG_S) \neq 0$ and $S$ encloses $T$.
    Further, let $\w_S$ cover $S$, let $C_S$ be a black cover with respect to $\w_S$, and let $R$
    be a Lempel--Ziv compression of $C_S$.
    There is a (classical) algorithm that, given $m$, $n$, $k$, $\eso{S}$, and $R$
    constructs
    an xSLP of size $\Ohtilde(k|S|+w)$ representing strings $P^\#$ and $T^\#$ of
    \cref{def:psh_ed}. The algorithm runs in $\Ohtilde((k|S|+w)^2)$ time.
    \lipicsEnd
\end{corollary}

\subsubsection{Compressed Construction of Black Covers and Other Related Data Structures}
\label{sec:g_s_constr_ed}

We now present efficient implementations for the various operations discussed earlier,
through the xSLP derived from the system of substring equations in
\cref{def:eq_system_ed_1}. This xSLP has three key applications. First, when $\bc(\bG_S) =
0$, it yields an xSLP representing both $P$ and $T$. Second, it supports \emph{projection
queries}, which can be used to construct a Lempel–Ziv compression of a black cover.
Finally, it enables verification of whether a $k$-error occurrence is captured.

\begin{definition}[{\cite[Definition~4.10]{KNW24}}]
    Let $P$ be a string of length $m$, let $T$ be a string of length $n \leq 3/2 \cdot m$,
    and let $k > 0$ be an integer threshold.
    Further, let $S \subseteq \mA_{P,T}^{\leq k}$.

    Given a position $x \in \fragmentco{0}{m}$, we define a function $\Pi_P(x)$ that maps
    $x$ to the largest position in $P$ that is less than or equal to $x$ and belongs to a
    black connected component. If no such character exists in $P$ preceding $x$, then
    $\Pi_P(x) = -1$. We define a symmetric function $\Pi_T(y)$ for all positions $y \in
    \fragmentco{0}{n}$.
\end{definition}

\begin{definition}[{\cite[Definition~6.3]{KNW24}}]
    Let $P$ be a string of length $m$, $T$ a string of length $n \leq 3/2 \cdot m$,
    and let $k > 0$ be an integer threshold. Additionally, let $S \subseteq
    \mA_{P,T}^{\leq k}$ be such that $S$ encloses $T$ and $\bc(\bG_S) \neq 0$.

    A \emph{projection query} can take one of the following two forms:
    \begin{enumerate}[(a)]
        \item Given $y \in \fragmentco{0}{n}$, output $c$ such that $\Pi_T(y) = \tau_i^c$
            if $\Pi_T(y) \neq -1$, or output $\bc(\bG_S) - 1$  otherwise.
        \label{def:proj_ds:a}
        \item Given $x \in \fragmentco{0}{m}$, output $c$ such that $\Pi_P(x) = \pi_i^c$
            if $\Pi_P(x) \neq -1$, or output $\bc(\bG_S) - 1$ otherwise.
        \label{def:proj_ds:b} \qedhere
    \end{enumerate}
\end{definition}

\begin{proposition}\label{prop:proj_ds}
    Let $P$ be a string of length $m$, let $T$ be a string of length $n \leq 3/2 \cdot m$,
    and let $k > 0$ be an integer threshold.
    Further, let $S \subseteq \mA_{P,T}^{\leq k}$ be such that $S$ encloses $T$.

    There exists a (classical) algorithm that, given $m,n,k$ and $\eso{S}$, computes
    $\bc(\bG_S)$ in $\Ohtilde((k|S|)^2)$ time. Furthermore:
    \begin{itemize}
        \item If $\bc(\bG_S) = 0$, it outputs an xSLP of size $\Ohtilde(k|S|)$
            representing $P$ and $T$.
        \item If $\bc(\bG_S) > 0$, it uses $\Ohtilde(k|S|)$ more time and outputs a data
            structure capable of answering the following types of queries in $\Oh(\log n)$
            time:
            \begin{enumerate}[(i)]
                \item projection queries, and
                \item given $t \in \fragmentco{0}{n}$, compute $\min_{i \in \fragment{0}{n_0 - m_0}} |\tau_i^0 - (t + \pi_0^0)|$.
                \lipicsEnd
            \end{enumerate}
    \end{itemize}
\end{proposition}
\begin{proof}
    We use \cref{cor:solve_substring_equations} to solve the system \(\exw{P}{T}{S}\) from
    \cref{def:eq_system_ed_1}.
    \Cref{cor:solve_substring_equations} yields an xSLP that represents a universal
    solution $(T',P',D')$.
    We replace the terminals occurring in $D'$ with $\sigma_0,\ldots,\sigma_{c-1}$ to
    obtain an xSLP $\G$ representing $P^\S$ and $T^\S$ as defined in \cref{def:sub_hash_ed_bad}.
    Since \( E \) has size \(\Oh(k|S|)\), with \(\Oh(k)\) equations for each \( \mX\in S \),
    it follows that $\G$ has size \(|\G| = \Ohtilde(k|S|)\).

    Next, in time \(\Ohtilde(k|S|)\), we can verify whether there are any remaining
    characters in \(P^\S\) and \(T^\S\) aside from \(\sigma_0, \ldots, \sigma_{c-1}\). If
    no such characters remain, this implies that all characters of \(P\) and \(T\) are
    contained within a red component, and \(P^\S = P\) and \(T^\S = T\). This brings us to
    the following claim.

    \begin{claim}
        If \(\bc(\bG_S) = 0\), we can output an xSLP representing $P$ and $T$.
    \end{claim}
    \begin{claimproof}
       If \(\bc(\bG_S) = 0\), then $P^\S = P$ and $T^\S = T$.
       Hence,  it simply suffices to output $\G$.
    \end{claimproof}

    Using $\Ohtilde(k|S|)$ time we construct $\hG$ from $\G$ as defined in
    \cref{lem:sub_xSLP}.
    This brings us an SLP $\hG$ representing strings $\hat{P}, \hat{T}$, defined
    accordingly to \cref{lem:sub_xSLP} with respect to $P$ and $T$.

    \begin{claim}\label{claim:proj_ds:proj}
        If \(\bc(\bG_S) > 0\), we can support projection queries in $\Oh(\log n)$ time.
    \end{claim}
    \begin{claimproof}
        To implement projection queries, we use rank and select queries on an auxiliary
        binary strings \(D_P\) and \(D_T\), derived from \(\hat{P}\) and \(\hat{T}\) by
        replacing each terminal of the form \(\#_0^A\) (for \(A \in \mP_\G\)) or \(\$\)
        with \(\one\), and all other terminals with \(\zero\). Formally, we generate
        \(D_P\) and \(D_T\) using an SLP \(\D\) of size \(\Ohtilde(k|S|)\), obtained from
        \(\hG\) by applying these replacements to each production. We construct \(\D\) in
        \(\Ohtilde(k|S|)\) time, and utilize \cref{lem:rank_queries} to support rank and
        select queries.

        Additionally, by examining the grammar, we can compute \(\bc(\bG_S)\) in
        \(\Ohtilde(k|S|)\) time, as this corresponds to the sum \(\sum_{A \in \mP_\G}
        |A|\), directly derived from \(\hG\).

        Now, observe that \(\mathsf{rank}(y)\) (with respect to \(D_T\)) gives the number
        of characters of \(T\fragmentco{0}{y}\) contained in a black connected component,
        as terminals are substituted with \(\one\) exactly when the corresponding
        character belongs to such component. Now, if $\mathrm{rank}(y)=0$, then \(\Pi_T(y)
        = -1\). Otherwise, $\mathrm{select}(\mathrm{rank}(y)-1)=\Pi_T(y)$. Taking this
        value modulo \(\bc(\bG_S)\) yields the desired result for the projection query
        with input \(y\).

        This show how we can support project queries for $T$.
        Similarly, we can support them for $P$.
    \end{claimproof}

    \begin{claim}
        If \(\bc(\bG_S) > 0\), we can support queries of the following type in $\Oh(\log
        n)$ time:
        Given $t \in \fragmentco{0}{n}$, compute $\min_{i \in \fragment{0}{n_0 - m_0}}
        |\tau_i^0 - (t + \pi_0^0)|$.
    \end{claim}
    \begin{claimproof}
        We use rank and select queries on the auxiliary binary strings \(D_P\) and
        \(D_T\), but this time they are defined differently from the approach in
        \cref{claim:proj_ds:proj}. In \(\Ohtilde(k|S|)\) time, we identify the leftmost
        pseudoterminal \(A \in \mP_\G\) in \(\G\) with respect to the production rules for
        \(P\).

        Next, we define \(D_P\) and \(D_T\), derived from \(\hat{P}\) and
        \(\hat{T}\), by replacing the terminal \(\#_0^A\) with \(\one\) and all other
        terminals with \(\zero\). By doing so, \(D_P\) and \(D_T\) contain
        a $\one$ exactly in the positions of $P$ and $T$ belonging to the $0$-th black
        component.
        As before, we generate \(D_P\) and \(D_T\) using an SLP \(\D\) of size
        \(\Ohtilde(k|S|)\), obtained from \(\hG\) by applying these replacements to each
        production.

        The construction of \(\D\) takes \(\Ohtilde(k|S|)\) time, and we use
        \cref{lem:rank_queries} to support rank and select queries.

        After this processing, we can describe how to compute \(\min_{i \in
        \fragment{0}{n_0 - m_0}} |\tau_i^0 - (t + \pi_0^0)|\) for a given \(t \in
        \fragmentco{0}{n}\). Let \(t' = \min(t + \pi_0^0, n)\). The minimization is
        achieved either at the largest \(\ell \in \fragment{0}{n_0 - m_0}\) such that
        \(\tau_\ell^0 < t'\), or at the smallest \(r \in \fragment{0}{n_0 - m_0}\) such
        that \(t' \leq \tau_r^0\) (it is not guaranteed that both indices exist, but at
        least one of $\ell$ and $r$ always exists).

        To answer the query algorithmically, we first obtain \(\pi_0^0\) by querying
        \(\mathsf{select}(0)\) with respect to \(D_P\). This gives us \(t' = \min(t +
        \pi_0^0, n)\). Next, we explain how to compute \(\ell\) and \(r\), which are
        sufficient to answer the query. By querying \(\mathsf{rank}(t')\) with respect to
        \(D_T\), we determine the number of characters in the 0-th component within
        \(T\fragmentco{0}{t'}\). Consequently, the queries
        \(\mathsf{select}(\mathsf{rank}(t') - 1)\) and
        \(\mathsf{select}(\mathsf{rank}(t'))\) with respect to \(D_T\) yield \(\ell\) and
        \(r\), respectively.
    \end{claimproof}

    This concludes the proof of \cref{prop:proj_ds}.
\end{proof}

Finally, we give the lemma from \cite{KNW24} constructing a Lempel--Ziv compressession of
a black cover.

\begin{lemma}[{\cite[Proposition~6.5 and Lemma~6.6]{KNW24}}] \label{prp:quantum_blackcover}
    Let $P$ be a string of length $m$, let $T$ be a string of length $n \leq 3/2 \cdot m$,
    and let $k > 0$ be an integer threshold.
    Further, let $S \subseteq \mA_{P,T}^{\leq k}$ be such that $S$ encloses $T$ and
    $\bc(\bG_S) \neq 0$.

    There exists a quantum algorithm that,
    given $m,n,k,\eso{S}$ and access to $\Ohtilde(1)$-time projection queries,
    outputs a Lempel--Ziv compressession $R$ of a black cover \(C_S\)
    with respect to a weight function $\w_S$ of total weight $w$ equal to the sum of costs
    all alignments in $S$
    using \(\Ohtilde(\sqrt{wm})\) queries and \(\Ohtilde(\sqrt{wm} + w^2)\) time.
    \lipicsEnd
\end{lemma}

\subsubsection{Quantum Algorithm for Gap Edit Distance}

As with mismatches,
a fundamental component is a quantum algorithm capable of distinguishing when the (edit)
distance between two strings is large or small.
Unlike mismatches, where sampling already provides a simple solution,
the edit case seems inherently harder.
For this reason, the edit case has been extensively studied in the classical setting and
is known specifically as the \gaped Problem.

\defproblemc{$(\beta,\alpha)$-\gaped}
{strings $X,Y$ of length $|X| = |Y| =n$, and integer thresholds $\alpha \geq \beta \geq
0$.}
{\yes if $\ed(X,Y) \leq \beta$, \no if $\ed(X,Y) > \alpha$, and an arbitrary answer
otherwise.}
{}

The authors in \cite{KNW24} provide a $\Oh(\sqrt{m})$-query quantum algorithm for this
problem,
though their algorithm falls short of achieving sublinear time.
In this (sub)section, we demonstrate how to modify their algorithm (which is a direct
adaptation of Goldenberg, Kociumaka, Krauthgamer, and Saha's algorithm~\cite{gapED} in the
classical setting) to achieve sublinear quantum time, while maintaining the same gap.

\begin{restatable*}{lemma}{tennine} \label{lem:quantum_gap_ed}
    There is a quantum algorithm that, given strings $X,Y$ and a threshold $k$,
    solves the $(k,kn^{o(1)})$-\gaped Problem on $X,Y$ using $\Ohhat(\sqrt{n})$ queries
    and time.
\end{restatable*}

The algorithm presented in \cite{gapED} solves an instance of \gaped by querying an oracle
that solves smaller instances of a similar problem: the \shifted \ \gaped Problem.
The \shifted \ \gaped Problem distinguishes between small \emph{shifted edit distance} and
large edit distance.
Given two strings \(X, Y\) and a threshold \(\beta\), the \emph{\(\beta\)-shifted edit
distance} \(\edshifted{\beta}(X,Y)\) is defined as
\[
    \edshifted{\beta}(X,Y) \coloneq \min \left(\bigcup_{\Delta=0}^{\min(|X|,|Y|,\beta)}
    \left\{\ed(X\fragmentco{\Delta}{|X|}, Y\fragmentco{0}{|Y| - \Delta}),
    \ed(X\fragmentco{0}{|X| - \Delta}, Y\fragmentco{\Delta}{|Y|})\right\}\right).
\]
Given thresholds \(\alpha \geq \beta \geq \gamma \geq 0\),
the \emph{\(\beta\)-\shifted \((\gamma,3\alpha)\)-\gaped Problem} requires to output
\textsf{yes} if \(\edshifted{\beta}(X,Y) \leq \gamma\), \textsf{no} if \(\ed(X,Y) >
3\alpha\), and an arbitrary answer otherwise.
The insight of \cite{gapED} is that an instance of \shifted \ \gaped can be reduced back
to smaller instances of \gaped.
In this way, \cite{gapED} obtains a recursive algorithm switching between instances of
\shifted \ \gaped and \gaped.

Out of the two reductions,
only the one from \gaped instances to \shifted \gaped instances has the linear time
bottleneck in their quantum adaptation.
The algorithm that reduces \gaped instances to \shifted \gaped instances partitions the
strings \(X\) and \(Y\) into blocks of different lengths (see \cref{alg:gapED}).
It samples these blocks and utilizes them as input for instances of the \shifted \gaped
Problem.
If a limited number of instances return \no, the routine outputs \yes.
Otherwise, it returns \no.

The authors of \cite{KNW24} adapted \cref{alg:gapED} by using \GS on the
iterations of \cref{ln:for_gaped}.
However, to do this, they need to ensure that the recursive calls to the \shifted \gaped
Problem are deterministic
(for instances that fall in the middle of the gap, there is not even a bound on the
probability of outputting \yes/\no).
To address this, they generated a linear seed of random bits in advance, which cost them
linear quantum time.

\begin{algorithm}[t]
\KwInput{An instance $(X,Y)$ of $(\beta, \alpha)$-\gaped, and a parameter $\phi \in
\mathbb{Z}_+$ satisfying
\[
    \phi \geq \beta \geq \psi \coloneq \left\lfloor \frac{112\beta \phi
\ceil{\log n}}{\alpha} \right\rfloor.
\]}
Set $\rho \coloneq {84\phi}/{\alpha}$\;
\For{$p \in \fragment{\ceil{\log (3\phi)}}{\floor{\log(\rho n)}}$}{ \label{ln:level_gaped}
    Set $m_p \coloneq \ceil{{n}/{2^p}}$ to be number of blocks of length $2^p$ we partition $X,Y$ into (last block might be shorter)\;
    For $i \in \fragmentco{0}{m_p}$ set $X_{p,i} = X\fragmentco{i \cdot 2^p}{\min(n, (i+1)2^p)}$ and $Y_{p,i} = Y\fragmentco{i \cdot 2^p}{\min(n, (i+1)2^p)}$\;
    \For {$t \in \fragmentco{0}{\ceil{\rho m_p}}$}{ \label{ln:for_gaped}
        Select $i \in \fragmentco{0}{m_p}$ uniformy at random\; \label{ln:sample}
        Solve the instance $(X_{p,i},Y_{p,i})$ of the $\beta$-\shifted $(\psi,3\phi)$-\gaped Problem\; \label{ln:gapedrec}
    }
}

\If{If we obtained at most 5 times \no at \cref{ln:gapedrec}}{ \label{ln:if_gaped}
    \Return{\yes}
}\Else{
    \Return{\no}
}
    \caption{Randomized Algorithm reducing \gaped to \shifted \ \gaped}\label{alg:gapED}
\end{algorithm}

\tennine
\begin{proof}
    In this proof, we proceed as follows.
    First, we explain how the correctness proof for \cref{alg:gapED} is structured.
    Then, we describe how \cref{alg:gapED} can be modified to preserve correctness while
    allowing a faster quantum implementation.
    Lastly, we briefly elaborate on how a faster quantum time for \cref{alg:gapED} leads
    to the claim in \cref{lem:quantum_gap_ed}.

    The correctness of \cref{alg:gapED} is proved in Lemma 4.3 of \cite{gapED},
    which shows that with probability at least \(1-e^{-1}\),
    the algorithm returns \yes if \(\ed(X,Y) \leq \beta\) and \no if \(\ed(X,Y) >
    \alpha\).
    This property needs to be preserved when modifying \cref{alg:gapED}.

    If \(\ed(X,Y) > \alpha\), the proof in \cite{gapED} considers the sets \(G_p = \{i \in
    \fragmentco{0}{m_p} \mid \ed(X_{p,i}, Y_{p,i}) > \psi\}\),
    representing the number of instances \((X_{p,i}, Y_{p,i})\) at iteration \(p\) of
    \cref{ln:level_gaped} for which the \(\beta\)-shifted \((\psi,3\alpha)\)-\gaped
    Problem may return \no.
    It is shown that these sets satisfy \(\sum_{p=\ceil{\log (3\phi)}}^{\floor{\log(\rho
    n)}} |G_p| \leq 3/(2\rho)\).
    This implies that the random variable \(\hat{b}\) for the number of times \no is
    obtained at \cref{ln:gapedrec} satisfies
    \begin{align} \label{eq:gs}
        \ex{\hat{b}}
        = \sum\nolimits_{p=\ceil{\log (3\phi)}}^{\floor{\log(\rho n)}} \frac{|G_p|
        \ceil{\rho m_p}}{m_p}
        \leq 2\rho \sum\nolimits_{p=\ceil{\log (3\phi)}}^{\floor{\log(\rho n)}} |G_p| \leq
        3.
    \end{align}
    Using suitable Chernoff Bounds, it follows that \(\pr{\hat{b} \leq 5} \leq 1/e\).

    Otherwise, if \(\ed(X,Y) \leq \beta\),
    the proof in \cite{gapED} considers the sets \[B_p = \{i \in \fragmentco{0}{m_p} \mid
    \ed(X_{p,i}, Y_{p,i}) > 3\phi\},\]
    which represent the number of instances \((X_{p,i}, Y_{p,i})\) for which the
    \(\beta\)-shifted \((\psi,3\phi)\)-\gaped is guaranteed to return \no.
    It is shown that \(\sum_{p=\ceil{\log (3\phi)}}^{\floor{\log(\rho n)}} |B_p| \geq
    10/\rho\), which implies
    \begin{align} \label{eq:bs}
        \ex{\hat{b}}
        = \sum\nolimits_{p=\ceil{\log (3\phi)}}^{\floor{\log(\rho n)}} \frac{|B_p|
        \ceil{\rho m_p}}{m_p}
        \geq \rho \sum\nolimits_{p=\ceil{\log (3\phi)}}^{\floor{\log(\rho n)}} |B_p| \geq
        10.
    \end{align}
    Again, using suitable Chernoff Bounds, \(\pr{\hat{b} \geq 6} \leq 1/e\).

    Set \(I = \{(p, t) \mid p \in\fragment{\ceil{\log (3\phi)}}{\floor{\log(\rho n)}}, t
    \in  \fragmentco{0}{\ceil{\rho m_p}}\}\).
    For \((p, t) \in I\), let \(i(p, t)\) be the sampled \(i\) at \cref{ln:sample}.
    Define \(I^+ = \{(p, t) \in I \mid  \ed(X_{p,i(p,t)}, Y_{p,i(p,t)}) > 3\alpha\}\) and
    \(I^- = \{(p, t) \in I \mid  \edshifted{\beta}(X_{p,i(p,t)}, Y_{p,i(p,t)}) \leq
    \psi\}\).
    Intuitively, we aim to apply \cref{lem:groverext} to all the \(\beta\)-shifted
    \((\psi,3\phi)\)-\gaped instances from \(I\) with \(I^+\) as positive elements and
    \(I^{-}\) as negative elements.
    However, \cref{lem:groverext} only allows us to distinguish between the case \(\hat{b}
    \geq |I^+| > 0\) and the case \(\hat{b} \leq |I \setminus I^-| = 0\),
    but not between \(\hat{b} \geq c\) and \(\hat{b} < c\) for some constant \(c > 0\).

    To resolve this, we modify \cref{alg:gapED} as follows.
    We change the number of iterations at \cref{ln:for_gaped} to \(\ceil{(\rho/6) \cdot
    m_p}\) from \(\ceil{\rho m_p}\) times (changing accordingly sets \(I\), \(I^+\), and
    \(I^-\)).
    We then use \cref{lem:groverext} over all \(I\) with \(I^+\) as positive elements and
    \(I^-\) as negative elements.
    From an implementation perspective, we generate \(i(p,t)\) inside the oracle calls and
    not before,
    to ensure we do not use linear quantum time anymore.
    To boost probability, we repeat \cref{lem:groverext} \(t = 300\) times, and if we
    obtain at least \(200\) \yes results from \cref{lem:groverext}, we output \yes;
    otherwise, we output \no.

    Consequently, if \(\ed(X,Y) > \alpha\), for a single repetition of
    \cref{lem:groverext}, using the same calculations as in \eqref{eq:gs}, we obtain
    \(\ex{\hat{b}} \leq 1/2\).
    Via Markov's inequality \(\pr{\hat{b} \geq 1} \leq 1/2\), and
    it follows that in expectation we obtain at least \yes over all repetitions.
    Otherwise, if \(\ed(X,Y) \leq \beta\), then for single repetition
    \[
        \pr{\hat{b} = 0}
        = \prod_{p=\ceil{\log (3\phi)}}^{\floor{\log(\rho n)}} \left(1 -
        \frac{|B_p|}{m_p}\right)^{\ceil{(\rho/6) \cdot m_p}}
        \leq \exp\left(-\sum_{p=\ceil{\log (3\phi)}}^{\floor{\log(\rho n)}} \frac{|B_p|
        \ceil{(\rho/6) \cdot m_p}}{m_p}\right)
        \leq e^{-5/3},
    \]
    from which follows \(\pr{\hat{b} \geq 1} = 1 - e^{-5/3} \geq 4/5\).
    Thus, we expect at least 240 \yes over all repetitions.
    Using concentration bounds, we obtain that we return the correct answer with
    probability at least \(1 - e^{-1}\) in both cases.
    This concludes the proof of the correctness of our modification.

    Regarding the running time, we observe that now a call to \cref{alg:gapED} requires
    \(\Ohtilde(\sqrt{m})\) queries and quantum time.
    This improvement allows the adaptation of \cite{KNW24} to achieve
    \(\Ohtilde(\sqrt{m})\) queries and quantum time.
    There is no need to formally prove this here, as for every step we now achieve the
    same query complexity and quantum time, and we obtain the exact same analysis for the query
    time here as the one for the quantum time in \cite{KNW24}.
\end{proof}

\subsubsection{From Candidate Positions to Proxy Strings}

Finally, we aim to prove our main technical step from \cref{thm:qedfindproxystrings}.

Our goal is to use the candidate set from \cref{lem:qed_candidate_set} to construct a set
\( S \) of alignments of \( P \) onto fragments of \( T \) that not only captures all \( k
\)-error alignments but also avoids unnecessarily costly alignments. If the candidate set
\( C \) is already in the second compressed form from \cref{lem:qed_candidate_set}, we
know that \( C \) inherently avoids such costly alignments.

However, if \( |C| = \Ohtilde(k) \), we aim to apply the quantum \gaped algorithm to
filter out candidate positions where costly \( k \)-error occurrences arise. A challenge
we face is that the \gaped algorithm requires two input strings of the same length, while
the candidate set only provides the starting positions of the \( k \)-error occurrences.

The following corollary of \cref{lem:quantum_gap_ed} helps to handle this mismatch in
formulation.

\begin{corollary}[of \cref{lem:quantum_gap_ed}]\label{cor:quantum_gap_ed}
    Let $P$ be a string of length $m$, let $T$ be a string of length $n \leq 3/2 \cdot m$,
    and let $k > 0$ be an integer threshold.
    There exists a quantum algorithm, that given as input $n,m,k$ and $t \in
    \fragment{0}{n}$, outputs \yes if $t \in \OccE_{k}(P,T)$, \no if $t \notin
    \OccE_{K}(P,T)$
    for some $K=k\cdot m^{o(1)}$, and an arbitrary \yes/\no otherwise.
    The algorithm requires $\Ohhat(\sqrt{m})$ queries and time.
\end{corollary}
\begin{proof}
    Denote $T_\$ \coloneqq T\cdot \$^{m-1}$ and $\Sigma_\$=\Sigma\cup\{\$\}$, where $\$\notin \Sigma$
    is a unique special character.
    As algorithm, we simply use the algorithm from \cref{lem:quantum_gap_ed} on strings
    $P,T_\$\fragmentco{t}{t+m}$ and threshold $2k$.

    If $t\in \OccE_k(P,T)$, then there is $t'\in \fragment{t}{n}$ such that $\ed(P,
    T\fragmentco{t}{t'})\le k$,
    and thus $\ed(P, T_\$\fragmentco{t}{t'})\le k$.
    Moreover, we have
    \begin{align*}
        \ed(P, T_\$\fragmentco{t}{t+m})
        &\le \ed(P, T_\$\fragmentco{t}{t'})+|t+m-t'|\\
        &= \ed(P, T\fragmentco{t}{t'})+\big||P|-|T\fragmentco{t}{t'}|\big|\\
        &\le 2\ed(P, T\fragmentco{t}{t'}) \le 2k.
    \end{align*}
    Hence, the algorithm from \cref{lem:quantum_gap_ed} returns \yes with high
    probability.

    Finally, note that the algorithm from \cref{lem:quantum_gap_ed} returns \no with high
    probability
    unless \[\ed(P,T_\$\fragmentco{t}{t+m})\le (2k)\cdot m^{o(1)}.\]

    If $t+m \le n$, the latter immediately implies $\ed(P,T\fragmentco{t}{t+m})\le
    (2k)\cdot m^{o(1)}$.
    Otherwise, we have
    \[\ed(P,T_\$\fragmentco{t}{t+m})=\ed(P\fragmentco{0}{p},T\fragmentco{t}{n})+\ed(P\fragmentco{p}{m},\$^{t+m-n})
    \] for some $p\in \fragment{0}{m}$, and thus
    \begin{align*}
        \ed(P,T\fragmentco{t}{n})
        &\le \ed(P\fragmentco{0}{p},T\fragmentco{t}{n})+
        \ed(P\fragmentco{p}{m},\varepsilon) \\
        &=   \ed(P\fragmentco{0}{p},T\fragmentco{t}{n})
        + m-p \\
        &\le
        \ed(P\fragmentco{0}{p},T\fragmentco{t}{n})+\ed(P\fragmentco{p}{m},\$^{t+m-n})\\
        &=\ed(P,T_\$\fragmentco{t}{t+m})
        \le (2k)\cdot m^{o(1)}.
    \end{align*}
    In either case, $t\in \OccE_{K}(P,T)$ for $K = (2k)\cdot m^{o(1)}$.
\end{proof}

Next, we discuss an oracle that, given a set \( F \subseteq \fragmentco{0}{n} \), either
identifies a position \( t \in F \cap \OccE_{K}(P, T) \), where \( K = k \cdot m^{o(1)}
\), or confirms that \( F \cap \OccE_k(P, T) = \emptyset \).

Even though \cref{lem:verifier} is formulated for general $F$, we use \( F \) as the
starting points of all \( K \)-error occurrences that the set \( S \) already captures.
Then, the oracle helps us in either finding an uncaptured \( K \)-error occurrence or
confirming that all \( k \)-error occurrences are already captured.

\begin{lemma}\label{lem:verifier}
    There exists a quantum algorithm that, given a pattern $P$ of length $m$, a text $T$
    of length $n\le 3/2\cdot m$, an integer $k>0$, and a set $F\subseteq \fragment{0}{n}$,
    outputs one of the following:
    \begin{itemize}
        \item a position $t\in F\cap \OccE_{K}(P,T)$, where $K=k\cdot m^{o(1)}$; or
        \item $\bot$, indicating that $F\cap \OccE_k(P,T)=\emptyset$.
    \end{itemize}
    The algorithm uses $\Ohhat(\sqrt{km}+k^2)$ time and $\Ohhat(\sqrt{km})$ queries to
    $P$, $T$, and the characteristic function of $F$.

    Within the same complexity bounds, we can guarantee that the reported position satisfies
    $t\le \min(F\cap \OccE_k(P,T))$ or, alternatively, that it satisfies $t \ge \max(F\cap
    \OccE_k(P,T))$.\footnote{We assume $\min \emptyset = +\infty$ and $\max \emptyset =
    -\infty$.}
\end{lemma}
\begin{proof}
    First, we apply \cref{lem:qed_candidate_set}, which returns $\floor{C/k}$ for a set
    $\OccE_k(P,T)\subseteq C \subseteq \fragment{0}{n}$ of candidate positions.
    The two cases of \cref{lem:qed_candidate_set} govern further behavior of the
    algorithm.

    If $|\floor{C/k}|=\Ohtilde(k)$, we combine \cref{cor:quantum_gap_ed} with
    \cref{lem:groverext}:

    We apply \cref{lem:groverext} to find $c\in \floor{C/k}$ such that the algorithm from
    \cref{cor:quantum_gap_ed} (with threshold $2k$) returns \yes on $P,T$ and $ck$ and
    such that $\fragmentco{ck}{ck+k}\cap F \ne \emptyset$ (using \GS to test
    the latter condition). When doing this, the sets $I^{\bm{+}}, I^{\bm{\sim}}$ and
    $I^{\bm{-}}$ in \cref{lem:groverext} are set to be $\floor{(C \cap \OccE_k(P,T))/k}$,
    $\floor{(C \cap \OccE_K(P,T))/k}$ and $\floor{C/k} \setminus (I^{\bm{+}} \cup
    I^{\bm{\sim}})$, respectively.

    Moreover, the failure probability $\delta$ in \cref{lem:groverext} is set to be such
    that the search succeeds w.h.p. If the search is successful, the algorithm reports an
    element of $\fragmentco{ck}{ck+k}\cap F$; otherwise, it outputs $\bot$.

    Whenever a guarantee on the reported position is requested,
    we use \cref{cor:groverext} instead of \cref{lem:groverext}.

    If $t\in F\cap \OccE_k(P,T)$, then $\floor{t/k}\in \floor{C/k}$, $k\cdot \floor{t/k}\in
    \OccE_{2k}(P,T)$, and $\fragmentco{\floor{t/k}k}{\floor{t/k}k+k}\cap F \ne \emptyset$, so
    the search is successful with high probability.
    Moreover, if both \cref{cor:groverext} and the inner
    \GS return the
    smallest (the largest) witnesses, then the reported value is guaranteed to be at most
    (at least, respectively) $t$.

    For the converse implication, note that the reported position $t\in
    \fragmentco{ck}{ck+k}\cap F$ satisfies $t\in \OccE_{K}(P,T)$ for $K=K'+k$, where
    $K'=(2k)\cdot m^{o(1)}$ is the higher threshold of \cref{cor:quantum_gap_ed} such that
    $ck\in \OccE_{K'}(P,T)$ holds with high probability.

    As for the complexity, observe that a single evaluation of the oracle for the outer
    \GS takes $\Ohhat(\sqrt{m}+\sqrt{k})$ queries and $\Ohhat(m+\sqrt{k})$
    time.
    Since the search space is of size $|\floor{C/k}|=\Ohtilde(\min(k,{m}/{k}))$, the
    total query and time complexity is $\Ohhat(\sqrt{km})$ and $\Ohhat(m\sqrt{k})$,
    respectively.

    In the remaining case, we have $C\subseteq \OccE_{44k}(P,T)$.
    We apply \GS to find $x\in \fragment{0}{n-m-k}$ such that $x \in F$ and $x
    \ \mathrm{mod} \ q \in I$.
    If the search is successful, the algorithm reports such $x$; otherwise, it outputs
    $\bot$.

    If $t\in F\cap \OccE_k(P,T)$, then from $C \subseteq \OccE_k(P,T)$, follows that the
    search is successful.
    For the converse implication, note that for the reported $t$ always $t \in C\subseteq
    \OccE_{44k}(P,T)$ holds.

    Since the search space is of size $\Oh(m)$, the total query and time complexity is
    $\Ohhat(\sqrt{m})$.
\end{proof}

With all components in place, we now proceed to prove \cref{thm:qedfindproxystrings}.

\qedfindproxystrings

\begin{proof}
    For this proof, let \( K \coloneqq k\cdot m^{o(1)} \) be defined as in
    \cref{lem:verifier} for $m$, $n$, and $k$.
    Additionally, set \( K' \coloneqq 2K + k \).

    We first show how to prove the lemma in the case we have two alignments that enclose
    $T$.

    \begin{claim}\label{claim:qedfindproxystringscropped}
        Suppose we are given $\sE_{P,T}(\mX_{\pref})$ and $\sE_{P,T}(\mX_{\suf})$ for
        alignments $\mX_{\pref}$ and $\mX_{\suf}$ of cost at most $K'$ such that $S
        \coloneqq \{\mX_{\pref}, \mX_{\suf}\}$ encloses $T$.

        Then, we can construct xSLPs of size $\Ohhat(k)$ representing strings $P^\#,T^\#$
        for which the two following hold:
        \begin{enumerate}[(a)]
            \item For every $a\in \fragmentco{0}{m}$ and $b \in \fragmentco{0}{n}$, we have
                that $P^\#\position{a} = T^\#\position{b}$ implies  $P\position{a} =
                T\position{b}$.
                \label{claim:qedfindproxystringscropped:it:a}
            \item for all optimal alignments $\mX : P \onto T\fragmentco{t}{t'}$ of cost
                at most $k$, we have $\ed(P, T\fragmentco{t}{t'}) = \ed(P^\#,
                T^\#\fragmentco{t}{t'})$ and $\sE_{P, T}(\mX) = \sE_{P^\#, T^\#}(\mX)$.
                \label{claim:qedfindproxystringscropped:it:b}
        \end{enumerate}
        We use $\Ohhat(\sqrt{km})$ queries and $\Ohhat(\sqrt{km}+k^2)$ quantum time.
    \end{claim}

    \begin{claimproof}
        Consider the following procedure.
        \begin{enumerate}[(i)]
            \item Start with $S = \{\mX_{\pref}, \mX_{\suf}\}$ containing alignments of
                cost at most $K'$. Iterate through the following steps, adding new
                elements to $S$.
                Instead of storing each $\mX \in S$ explicitly, store only
                $\sE_{P,T}(\mX)$.
                Continue this process until the loop is not interrupted.
                \label{alg:qedfindproxystringscropped:i}
                \begin{enumerate}[(1)]
                    \item Check whether $\bc(\bG_S) = 0$ via \cref{prop:proj_ds}.
                        If $\bc(\bG_S) = 0$, interrupt the loop.
                        \label{alg:qedfindproxystringscropped:1}
                    \item Otherwise, if $\bc(\bG_S) > 0$, let $w$ be the sum of the costs
                        of all        alignments in $S$.
                        Apply \cref{lem:verifier} for $F = \{t \in \fragment{0}{n} :
                        \forall_{i\in \fragment{0}{n_0-m_0}}\, |\tau_i^0-t-\pi_0^0|>
                    w+3K'\}$, using queries supported by \cref{prop:proj_ds} to test in
                    $\Oh(\log n)$ time whether a given position belongs to $F$.
                        If \cref{lem:verifier} returns $\bot$, interrupt the loop.
                        \label{alg:qedfindproxystringscropped:2}
                    \item If \cref{lem:verifier} returns $t\in F$, apply
                        \cref{prp:quantumed_w_info} to retrieve $\sE_{P,T}(\mY)$ for an
                        optimal alignment $\mY : P \onto T\fragmentco{t}{\min(t + m, n)}$
                        and $\mY$ to $S$.
                        \label{alg:qedfindproxystringscropped:3}
                \end{enumerate}
            \item If the loop was interrupted in \eqref{alg:qedfindproxystringscropped:1},
                reconstruct the xSLPs representing $P,T$ via \cref{prop:proj_ds} and
                return the xSLPs.
                Otherwise, if the loop was interrupted in
                \eqref{alg:qedfindproxystringscropped:2},
                retrieve the compression $R$ of a black cover from
                \cref{prp:quantum_blackcover},
                construct the xSLPs representing $P^\#,T^\#$ via \cref{lem:qed_proxy_2}
                using $R$, and return the xSLPs.
        \end{enumerate}

        Each time we add \( \mY \) to \( S \), the corresponding position $t$ satisfies \(
        |\tau_{i}^{0} - t - \pi_0^{0}| > w + 3K' \) for all \( i \in \fragment{0}{n_0-m_0} \)
        and $\ed(P,T\fragmentco{t}{t'})\le K$ for some $t'\in \fragment{t}{n}$.

        Since $|\min(t+m,n)-t'| \le |t+m-t'|\le  \ed(P,T\fragmentco{t}{t'}) \le K$, we
        conclude that $\ed(P,T\fragmentco{t}{\min(t+m,n)})\le 2K$, and thus the cost of $\mY$
        is at most $2K \le K'$.

        Consequently, adding $\mY$ to $S$ preserves the invariant and, by
        \cref{lem:periodhalves}, halves the number of black components, that is, \( \bc(S \cup
        \{\mY\}) \leq \bc(S)/2 \).

        If the loop is interrupted in \eqref{alg:qedfindproxystringscropped:2}, then there is
        no $k$-error occurrence $T\fragmentco{t}{t'}$ of $P$ such that $t\in F$, and thus $S$
        captures all $k$-error occurrences of $P$ in $T$, so the correctness of the procedure
        follows from \cref{prp:ed_subhash}.

        If the loop is interrupted in \eqref{alg:qedfindproxystringscropped:1}, then we obtain
        directly xSLPs representing $P$ and $T$ from \cref{prop:proj_ds}.

        Regarding computational complexity, observe that the main loop iterates $\Oh(\log n)$
        times because the number of black components is initially $\Oh(n+m)$, it halves upon
        the successful completion of every iteration, and the algorithm exits the loop as soon
        as the number of black components reaches zero.

        Consequently, we always maintain the size of $S$ bounded by $|S| = \Oh(\log n)$.
        For each iteration, the runtime is dominated by the application of \cref{lem:verifier},
        the construction of the data structure from \cref{prop:proj_ds},
        and the retrieval of the edit information in \cref{prp:quantumed_w_info}.
        All of them take $\Ohhat(\sqrt{km})$ queries and $\Ohhat(\sqrt{km}+k^2)$ time.
    \end{claimproof}

    The following routine reconnects the procedure from
    \cref{claim:qedfindproxystringscropped} to the general case.
    \begin{enumerate}[(i)]
        \item If \( 4K' > m \), read the entire strings \( P \) and \( T \), and return
            them directly.
            \label{alg:qedfindproxystrings:i}
        \item Apply \cref{lem:verifier} for $F=\fragment{0}{n}$, which either outputs
            $\bot$, indicating that $\OccE_k(P,T)=\emptyset$,
            or returns positions $t_{\pref},t_{\suf}\in \OccE_{K}(P,T)$ such that $t_\pref
            \le \min \OccE_k(P,T)$ and $t_\suf \ge \max \OccE_k(P,T)$.
            \label{alg:qedfindproxystrings:iii}
        \item If \cref{lem:verifier} outputs $\bot$, return an xSLP representing strings
            \( P^\# \) and \( T^\# \) of length $m$ and $n$, each with a unique sentinel
            character at every position. To achieve this it suffices to represent
            \( P^\# \) and \( T^\# \) with two pseudoterminals of length $m$ and $n$.
            Otherwise, use \cref{prp:quantumed_w_info} to collect \(
            \sE_{P,T}(\mX_{\pref}) \) and \( \sE_{P,T}(\mX_{\suf}) \) for some optimal
            alignments \( \mX_{\pref} : P \onto  T\fragmentco{t_{\pref}}{\min(t_{\pref} +
            m, n)} \) and \( \mX_{\suf} : P \onto T\fragmentco{t_{\suf}}{t_{\mathrm{end}}}
            \), where $t_{\mathrm{end}} \coloneqq \min(t_{\suf} + m+k, n)$.
            \label{alg:qedfindproxystrings:iv}
        \item Apply \cref{claim:qedfindproxystringscropped} on the pattern \( P \), the
            cropped text \( T' = T\fragmentco{t_{\pref}}{t_{\mathrm{end}}} \), the edit
            information \( \sE_{P,T}(\mX_{\pref}) \) and \( \sE_{P,T}(\mX_{\suf}) \) with
            indices shifted so that they are defined on \( T' \).
            \label{alg:qedfindproxystrings:v}
        \item Finally, modify the returned \( P^\# \) and \( T^\# \) by prepending \(
            t_{\pref} \) unique sentinels to \( T^\# \) and appending \( n -
            t_{\mathrm{end}} \) unique sentinels to \( T^\# \). To do this it suffices to
            modify the xSLP exactly as done in \cref{thm:qhdfindproxystrings}.
            \label{alg:qedfindproxystrings:vi}
    \end{enumerate}

    Next, we argue the correctness of the routine.

    If the routine returns at \eqref{alg:qedfindproxystrings:i}, then \( P \) and \( T \)
    trivially satisfy the claimed properties. Additionally, since \( 4K' > m \) implies \(
    k = \hat{\Omega}(K) = \hat{\Omega}(m) \), the \(\Oh(m)\) queries needed to read the
    strings are within the claimed query complexity of \(\Ohhat(\sqrt{km})\).

    On the other hand, if the routine returns at \eqref{alg:qedfindproxystrings:iv}, then
    \( \OccE_k(P,T) = \emptyset \).
    In this case, the returned strings \( P^\# \) and \( T^\# \) satisfy \( \ed(P,
    T\fragmentco{t}{t'}) \leq \ed(P^\#, T^\#\fragmentco{t}{t'}) \) for all $t,t'$, thereby
    fulfilling the condition \( \OccE_k(P^\#,T^\#) = \emptyset\) as desired.

    Finally, let us consider the case where we neither return at
    \eqref{alg:qedfindproxystrings:i} nor at \eqref{alg:qedfindproxystrings:iii}.

    Since $t_{\pref},t_\suf \in \OccE_{K}(P,T)$ and the underlying $K$-error occurrences
    have lengths in $\fragment{m-K}{m+K}$, the costs of $\mX_{\pref}$ and $\mX_\suf$ are
    at most $2K<K'$ and $2K+k=K'$, respectively.

    Furthermore, $4K' \le m$ implies $|T'|\le n \le 3/2 \cdot m \le 2m-2K'$, so the shifted
    alignments $\mX_{\pref}'$ and $\mX_{\suf}'$, which are defined on $T'$ instead of $T$
    and are used as input for \cref{claim:qedfindproxystringscropped}, enclose $T'$.

    Therefore, the assumptions of \cref{claim:qedfindproxystringscropped} are satisfied.

    Let $\mX : P \onto T\fragmentco{t}{t'}$ be an alignment with cost $k' \leq k$. Our
    goal is to prove the first part of the statement, speciffically that $\mX : P^\# \onto
    T^\#\fragmentco{t}{t'}$ also has a cost of $k' \leq k$.

    By construction, $t\in \OccE_k(P,T)$, so $t_\pref \le t \le t_\suf$ and $t' \le
    \min(n, t+m+k) \le \min(n,t_\suf+m+k)=t_{\mathrm{end}}$.

    This implies the existence of a corresponding alignment $\mX' : P \onto
    T'\fragmentco{t-t_{\pref}}{t'-t_{\pref}}$ with shifted indices and the same cost $k'$.
    According to
    \cref{claim:qedfindproxystringscropped}\eqref{claim:qedfindproxystringscropped:it:b},
    this alignment is preserved in the resulting $P^\#$ and $T^\#$. Prepending sentinel
    characters to $T^\#$ then ensures that $\mX : P^\# \onto T^\#\fragmentco{t}{t'}$ keeps
    cost $k'$.

    For the reverse direction, observe that $P^\#$ and $T^\#$, even after the addition of
    sentinel characters, satisfy
    \cref{claim:qedfindproxystringscropped}\eqref{claim:qedfindproxystringscropped:it:a}.
    Therefore, for any alignment $\mX : P \onto T\fragmentco{t}{t'}$, we have
    $\edal{\mX}(P, T\fragmentco{t}{t'}) \leq \edal{\mX}(P^\#, T^\#\fragmentco{t}{t'})$.
    This proves the other direction of the statement.
\end{proof}

\subsection{Step 4: Bringing It All Together}

We conclude this section with deriving \cref{thm:qpmwe} from
\cref{thm:qedfindproxystrings}.

\qpmwe

\begin{proof}
    We first show how to prove \cref{thm:qpmwe} assuming that we have a quantum algorithm
    $\mathsf{A}$ that solves \PMwE when $n \leq 5/4 \cdot m+k$, using $\Ohhat(\sqrt{km})$
    queries and $\Ohhat(\sqrt{km} + k^{3.5})$ time.

    Consider partitioning \(T\) into \(\Oh(n/m)\) contiguous blocks of length
    \(\ceil{m/4}\) each (with the last block potentially being shorter).

    For every block that starts at position $i \in \fragmentco{0}{|T|}$,  we consider a
    segment $T_i\coloneqq T\fragmentco{i}{\min(n, i + \floor{5/4 \cdot m}+k)}$ of length
    at most $5/4 \cdot m+k$.

    Every $k$-error occurrence \(T\fragmentco{t}{t'}\) starting at $t\in
    \fragmentco{i}{i+\ceil{m/4}}$ is of length at most $m+k$ and thus is fully contained
    in the segment $T_i$.

    For the first claim, we iterate over all such segments, applying $\mathsf{A}$ to
    \(P\), the current segment $T_i$, and \(k\).
    The final set of occurrences is obtained by taking the union of all sets returned by
    the quantum algorithm.

    For the second claim, we employ \GS over all segments, utilizing a
    function that determines whether the set of $k$-error occurrences of $P$ in the
    segment $T_i$ is empty or not.

    For the purpose of showing how to satisfy the assumption,  we distinguish three
    different regimes of $k$:
    \begin{itemize}
        \item If $k < m/4$, we set as $\mathsf{A}$ an algorithm that first uses
            \cref{thm:qedfindproxystrings} to obtain xSLPs for the proxy strings $P^\#$
            and $T^\#$, uses \cref{lem:pillar_on_xslp} to enable the use of \modelname
            operations on them,
            and then applies the classical algorithm of \cref{thm:edalg} to compute
            $\OccE_k(P^\#, T^\#)$.
            The result is returned as $\OccE_k(P, T)$.
            Observe that the conditions for \cref{thm:qedfindproxystrings} are satisfied since
            $5/4\cdot m +k < 3/2 \cdot m$.
        \item If $m/4 \le k \leq m$, we set as $\mathsf{A}$ an algorithm that first reads
            $P$ and $T$ completely and then directly applies the classical algorithm of
            \cref{thm:edalg} to compute $\OccE_k(P, T)$.
            From $k = \Omega(m)$ and $n =\Oh(m)$ follows that this requires $m+n =
            \Oh(\sqrt{km})$ queries and $\Ohtilde(m + k^{3.5})$ time.
        \item If $m \leq k$, then $\mathsf{A}$ simply returns $\OccE_k(P,T)=\fragment{0}{n}$ represented as an arithmetic progression.
            \qedhere
    \end{itemize}
\end{proof}

%% file: main.bbl
\newcommand{\etalchar}[1]{$^{#1}$}
\begin{thebibliography}{HMDW03}

\bibitem[ABI{\etalchar{+}}20]{ABIKKPSS20}
Andris Ambainis, Kaspars Balodis, J\={a}nis Iraids, Kamil Khadiev, Vladislavs
  K\c{l}evickis, Kri\v{s}j\={a}nis Pr\={u}sis, Yixin Shen, Juris Smotrovs, and
  Jevg\={e}nijs Vihrovs.
\newblock Quantum lower and upper bounds for {2D}-grid and {Dyck} language.
\newblock In {\em 45th International Symposium on Mathematical Foundations of
  Computer Science, {MFCS} 2020}, 2020.
\newblock \href {https://doi.org/10.4230/LIPIcs.MFCS.2020.8}
  {\path{doi:10.4230/LIPIcs.MFCS.2020.8}}.

\bibitem[Abr87]{Abr87}
Karl~R. Abrahamson.
\newblock Generalized string matching.
\newblock {\em {SIAM} Journal on Computing}, 1987.
\newblock \href {https://doi.org/10.1137/0216067} {\path{doi:10.1137/0216067}}.

\bibitem[ABR00]{ABR00}
Stephen Alstrup, Gerth~St{\o}lting Brodal, and Theis Rauhe.
\newblock Pattern matching in dynamic texts.
\newblock In {\em Proceedings of the Eleventh Annual {ACM-SIAM} Symposium on
  Discrete Algorithms, SODA 2000}, 2000.

\bibitem[AGS19]{AGS19}
Scott Aaronson, Daniel Grier, and Luke Schaeffer.
\newblock A quantum query complexity trichotomy for regular languages.
\newblock In {\em 60th {IEEE} Annual Symposium on Foundations of Computer
  Science, {FOCS} 2019}, 2019.
\newblock \href {https://doi.org/10.1109/FOCS.2019.00061}
  {\path{doi:10.1109/FOCS.2019.00061}}.

\bibitem[AJ23]{AJ22}
Shyan Akmal and Ce~Jin.
\newblock Near-optimal quantum algorithms for string problems.
\newblock {\em Algorithmica}, 2023.
\newblock \href {https://doi.org/10.1007/S00453-022-01092-X}
  {\path{doi:10.1007/S00453-022-01092-X}}.

\bibitem[ALP04]{AmirLP04}
Amihood Amir, Moshe Lewenstein, and Ely Porat.
\newblock Faster algorithms for string matching with $k$ mismatches.
\newblock {\em Journal of Algorithms}, 2004.
\newblock \href {https://doi.org/10.1016/S0196-6774(03)00097-X}
  {\path{doi:10.1016/S0196-6774(03)00097-X}}.

\bibitem[Amb04]{ambainis2004quantum}
Andris Ambainis.
\newblock Quantum query algorithms and lower bounds.
\newblock In {\em Classical and New Paradigms of Computation and their
  Complexity Hierarchies}, Trends in Logic, 2004.
\newblock \href {https://doi.org/10.1007/978-1-4020-2776-5_2}
  {\path{doi:10.1007/978-1-4020-2776-5_2}}.

\bibitem[BBC{\etalchar{+}}95]{PhysRevA.52.3457}
Adriano Barenco, Charles~H. Bennett, Richard Cleve, David~P. DiVincenzo, Norman
  Margolus, Peter Shor, Tycho Sleator, John~A. Smolin, and Harald Weinfurter.
\newblock Elementary gates for quantum computation.
\newblock {\em Physical Review A}, 1995.
\newblock \href {https://doi.org/10.1103/PhysRevA.52.3457}
  {\path{doi:10.1103/PhysRevA.52.3457}}.

\bibitem[BCPT15]{BCPT15}
Djamal Belazzougui, Patrick~Hagge Cording, Simon~J. Puglisi, and Yasuo Tabei.
\newblock Access, rank, and select in grammar-compressed strings.
\newblock In {\em Algorithms - {ESA} 2015 - 23rd Annual European Symposium},
  Lecture Notes in Computer Science, 2015.
\newblock \href {https://doi.org/10.1007/978-3-662-48350-3_13}
  {\path{doi:10.1007/978-3-662-48350-3_13}}.

\bibitem[BdW02]{DBLP:journals/tcs/BuhrmanW02}
Harry Buhrman and Ronald de~Wolf.
\newblock Complexity measures and decision tree complexity: a survey.
\newblock {\em Theoretical Computer Science}, 2002.
\newblock \href {https://doi.org/10.1016/S0304-3975(01)00144-X}
  {\path{doi:10.1016/S0304-3975(01)00144-X}}.

\bibitem[BEG{\etalchar{+}}21]{BEGHS21}
Mahdi Boroujeni, Soheil Ehsani, Mohammad Ghodsi, MohammadTaghi Hajiaghayi, and
  Saeed Seddighin.
\newblock Approximating edit distance in truly subquadratic time: Quantum and
  {MapReduce}.
\newblock {\em Journal of the ACM}, 2021.
\newblock \href {https://doi.org/10.1145/3456807} {\path{doi:10.1145/3456807}}.

\bibitem[BHMT02]{brassard2002quantum}
Gilles Brassard, Peter H{\o}yer, Michele Mosca, and Alain Tapp.
\newblock Quantum amplitude amplification and estimation.
\newblock In {\em Quantum computation and information}, Contemporary
  Mathematics. 2002.
\newblock \href {https://doi.org/10.1090/conm/305/05215}
  {\path{doi:10.1090/conm/305/05215}}.

\bibitem[BK23]{BK23}
Sudatta Bhattacharya and Michal Kouck{\'{y}}.
\newblock Streaming $k$-edit approximate pattern matching via string
  decomposition.
\newblock In {\em 50th International Colloquium on Automata, Languages, and
  Programming, {ICALP} 2023}, 2023.
\newblock \href {https://doi.org/10.4230/LIPIcs.ICALP.2023.22}
  {\path{doi:10.4230/LIPIcs.ICALP.2023.22}}.

\bibitem[BKW19]{bkw19}
Karl Bringmann, Marvin K{\"{u}}nnemann, and Philip Wellnitz.
\newblock Few matches or almost periodicity: Faster pattern matching with
  mismatches in compressed texts.
\newblock In {\em 30th Annual {ACM-SIAM} Symposium on Discrete Algorithms,
  {SODA} 2019}, 2019.
\newblock \href {https://doi.org/10.1137/1.9781611975482.69}
  {\path{doi:10.1137/1.9781611975482.69}}.

\bibitem[BPS21]{BPS19}
Harry Buhrman, Subhasree Patro, and Florian Speelman.
\newblock {A Framework of Quantum Strong Exponential-Time Hypotheses}.
\newblock In {\em 38th International Symposium on Theoretical Aspects of
  Computer Science (STACS 2021)}, 2021.
\newblock \href {https://doi.org/10.4230/LIPIcs.STACS.2021.19}
  {\path{doi:10.4230/LIPIcs.STACS.2021.19}}.

\bibitem[CFP{\etalchar{+}}16]{CFPSS16}
Rapha\"{e}l Clifford, Allyx Fontaine, Ely Porat, Benjamin Sach, and Tatiana
  Starikovskaya.
\newblock The k-mismatch problem revisited.
\newblock In {\em Proceedings of the Twenty-Seventh Annual ACM-SIAM Symposium
  on Discrete Algorithms}, SODA '16, 2016.

\bibitem[CGK{\etalchar{+}}20]{CGKKP20}
Timothy~M. Chan, Shay Golan, Tomasz Kociumaka, Tsvi Kopelowitz, and Ely Porat.
\newblock Approximating text-to-pattern hamming distances.
\newblock In {\em Proceedings of the 52nd Annual ACM SIGACT Symposium on Theory
  of Computing}, STOC 2020, 2020.
\newblock \href {https://doi.org/10.1145/3357713.3384266}
  {\path{doi:10.1145/3357713.3384266}}.

\bibitem[CGK{\etalchar{+}}22]{CGKMU22}
Rapha{\"{e}}l Clifford, Pawel Gawrychowski, Tomasz Kociumaka, Daniel~P. Martin,
  and Przemyslaw Uznanski.
\newblock The dynamic k-mismatch problem.
\newblock In {\em 33rd Annual Symposium on Combinatorial Pattern Matching,
  {CPM} 2022}, 2022.
\newblock \href {https://doi.org/10.4230/LIPICS.CPM.2022.18}
  {\path{doi:10.4230/LIPICS.CPM.2022.18}}.

\bibitem[CH02]{ColeH98}
Richard Cole and Ramesh Hariharan.
\newblock Approximate string matching: {A} simpler faster algorithm.
\newblock {\em SIAM Journal on Computing}, 2002.
\newblock \href {https://doi.org/10.1137/S0097539700370527}
  {\path{doi:10.1137/S0097539700370527}}.

\bibitem[CKP19]{CKP19}
Rapha{\"{e}}l Clifford, Tomasz Kociumaka, and Ely Porat.
\newblock The streaming $k$-mismatch problem.
\newblock In {\em 30th Annual {ACM-SIAM} Symposium on Discrete Algorithms,
  {SODA} 2019}, 2019.
\newblock \href {https://doi.org/10.1137/1.9781611975482.68}
  {\path{doi:10.1137/1.9781611975482.68}}.

\bibitem[CKW20]{CKW20}
Panagiotis Charalampopoulos, Tomasz Kociumaka, and Philip Wellnitz.
\newblock Faster approximate pattern matching: {A} unified approach.
\newblock In {\em 61st {IEEE} Annual Symposium on Foundations of Computer
  Science, {FOCS} 2020}, 2020.
\newblock \href {https://doi.org/10.1109/FOCS46700.2020.00095}
  {\path{doi:10.1109/FOCS46700.2020.00095}}.

\bibitem[CKW22]{CKW22}
Panagiotis Charalampopoulos, Tomasz Kociumaka, and Philip Wellnitz.
\newblock Faster pattern matching under edit distance: {A} reduction to dynamic
  puzzle matching and the seaweed monoid of permutation matrices.
\newblock In {\em 63rd {IEEE} Annual Symposium on Foundations of Computer
  Science, {FOCS} 2022}, 2022.
\newblock \href {https://doi.org/10.1109/FOCS54457.2022.00072}
  {\path{doi:10.1109/FOCS54457.2022.00072}}.

\bibitem[CPR{\etalchar{+}}24]{CPR24}
Panagiotis Charalampopoulos, Solon~P. Pissis, Jakub Radoszewski, Wojciech
  Rytter, Tomasz Waleń, and Wiktor Zuba.
\newblock Approximate circular pattern matching under edit distance.
\newblock In {\em 41st International Symposium on Theoretical Aspects of
  Computer Science, {STACS} 2024}, 2024.
\newblock \href {https://arxiv.org/abs/2402.14550} {\path{arXiv:2402.14550}},
  \href {https://doi.org/10.48550/arxiv.2402.14550}
  {\path{doi:10.48550/arxiv.2402.14550}}.

\bibitem[DK24]{DK24}
Anouk Duyster and Tomasz Kociumaka.
\newblock Logarithmic-time internal pattern matching queries in compressed and
  dynamic texts.
\newblock In {\em 31st International Symposium on String Processing and
  Information Retrieval, {SPIRE} 2024}, LNCS, 2024.
\newblock \href {https://doi.org/10.1007/978-3-642-34109-0_24}
  {\path{doi:10.1007/978-3-642-34109-0_24}}.

\bibitem[FW65]{FW65}
Nathan~J. Fine and Herbert~S. Wilf.
\newblock Uniqueness theorems for periodic functions.
\newblock {\em Proceedings of the American Mathematical Society}, 1965.
\newblock \href {https://doi.org/10.1090/S0002-9939-1965-0174934-9}
  {\path{doi:10.1090/S0002-9939-1965-0174934-9}}.

\bibitem[GG86]{GG86}
Zvi Galil and Raffaele Giancarlo.
\newblock Improved string matching with $k$ mismatches.
\newblock {\em SIGACT News}, 1986.
\newblock \href {https://doi.org/10.1145/8307.8309}
  {\path{doi:10.1145/8307.8309}}.

\bibitem[GJKT24]{GJKT24}
Daniel Gibney, Ce~Jin, Tomasz Kociumaka, and Sharma~V. Thankachan.
\newblock Near-optimal quantum algorithms for bounded edit distance and
  {Lempel-Ziv} factorization.
\newblock In {\em 35th {ACM-SIAM} Symposium on Discrete Algorithms, {SODA}
  2023}, 2024.
\newblock \href {https://doi.org/10.1137/1.9781611977912.11}
  {\path{doi:10.1137/1.9781611977912.11}}.

\bibitem[GK17]{GK17}
Pawel Gawrychowski and Tomasz Kociumaka.
\newblock Sparse suffix tree construction in optimal time and space.
\newblock In {\em 28th Annual {ACM-SIAM} Symposium on Discrete Algorithms,
  {SODA} 2017}, 2017.
\newblock \href {https://doi.org/10.1137/1.9781611974782.27}
  {\path{doi:10.1137/1.9781611974782.27}}.

\bibitem[GKKS22]{gapED}
Elazar Goldenberg, Tomasz Kociumaka, Robert Krauthgamer, and Barna Saha.
\newblock Gap edit distance via non-adaptive queries: Simple and optimal.
\newblock In {\em 63rd {IEEE} Annual Symposium on Foundations of Computer
  Science, {FOCS} 2022}, 2022.
\newblock \href {https://doi.org/10.1109/FOCS54457.2022.00070}
  {\path{doi:10.1109/FOCS54457.2022.00070}}.

\bibitem[GKR{\etalchar{+}}20]{GKRR20}
Pawel Gawrychowski, Tomasz Kociumaka, Jakub Radoszewski, Wojciech Rytter, and
  Tomasz Walen.
\newblock Universal reconstruction of a string.
\newblock {\em Theor. Comput. Sci.}, 2020.
\newblock \href {https://doi.org/10.1016/J.TCS.2019.10.027}
  {\path{doi:10.1016/J.TCS.2019.10.027}}.

\bibitem[Gro96]{DBLP:conf/stoc/Grover96}
Lov~K. Grover.
\newblock A fast quantum mechanical algorithm for database search.
\newblock In {\em 28th Annual {ACM} Symposium on the Theory of Computing, STOC
  1996}, 1996.
\newblock \href {https://doi.org/10.1145/237814.237866}
  {\path{doi:10.1145/237814.237866}}.

\bibitem[GS23]{GS22}
Fran{\c{c}}ois~Le Gall and Saeed Seddighin.
\newblock Quantum meets fine-grained complexity: Sublinear time quantum
  algorithms for string problems.
\newblock {\em Algorithmica}, 2023.
\newblock \href {https://doi.org/10.1007/S00453-022-01066-Z}
  {\path{doi:10.1007/S00453-022-01066-Z}}.

\bibitem[GU18]{GU18}
Pawel Gawrychowski and Przemyslaw Uznanski.
\newblock {Towards Unified Approximate Pattern Matching for Hamming and
  L\underline1 Distance}.
\newblock In {\em 45th International Colloquium on Automata, Languages, and
  Programming (ICALP 2018)}, 2018.
\newblock \href {https://doi.org/10.4230/LIPIcs.ICALP.2018.62}
  {\path{doi:10.4230/LIPIcs.ICALP.2018.62}}.

\bibitem[HD80]{HD80}
Patrick A.~V. Hall and Geoff~R. Dowling.
\newblock Approximate string matching.
\newblock {\em {ACM} Comput. Surv.}, 1980.
\newblock \href {https://doi.org/10.1145/356827.356830}
  {\path{doi:10.1145/356827.356830}}.

\bibitem[HMDW03]{groverext}
Peter H\o{}yer, Michele Mosca, and Ronald De~Wolf.
\newblock Quantum search on bounded-error inputs.
\newblock In {\em Proceedings of the 30th International Conference on Automata,
  Languages and Programming}, ICALP'03, 2003.

\bibitem[HV03]{HV03}
Ramesh Hariharan and V.~Vinay.
\newblock String matching in {$\tilde{O}(\sqrt{n}+\sqrt{m})$} quantum time.
\newblock {\em Journal of Discrete Algorithms}, 2003.
\newblock \href {https://doi.org/10.1016/S1570-8667(03)00010-8}
  {\path{doi:10.1016/S1570-8667(03)00010-8}}.

\bibitem[I17]{tomohiro}
Tomohiro I.
\newblock {Longest Common Extensions with Recompression}.
\newblock In {\em 28th Annual Symposium on Combinatorial Pattern Matching (CPM
  2017)}, 2017.
\newblock \href {https://doi.org/10.4230/LIPIcs.CPM.2017.18}
  {\path{doi:10.4230/LIPIcs.CPM.2017.18}}.

\bibitem[Jez15]{Jez15}
Artur Jez.
\newblock Faster fully compressed pattern matching by recompression.
\newblock {\em {ACM} Trans. Algorithms}, 2015.
\newblock \href {https://doi.org/10.1145/2631920} {\path{doi:10.1145/2631920}}.

\bibitem[JN23]{JN23}
Ce~Jin and Jakob Nogler.
\newblock Quantum speed-ups for string synchronizing sets, longest common
  substring, and $k$-mismatch matching.
\newblock In {\em 34th {ACM-SIAM} Symposium on Discrete Algorithms, {SODA}
  2023}, 2023.
\newblock \href {https://doi.org/10.1137/1.9781611977554.CH186}
  {\path{doi:10.1137/1.9781611977554.CH186}}.

\bibitem[KJP77]{KMP77}
Donald~E. Knuth, James H.~Morris Jr., and Vaughan~R. Pratt.
\newblock Fast pattern matching in strings.
\newblock {\em {SIAM} J. Comput.}, 1977.
\newblock \href {https://doi.org/10.1137/0206024} {\path{doi:10.1137/0206024}}.

\bibitem[KK20]{KK20}
Dominik Kempa and Tomasz Kociumaka.
\newblock Resolution of the burrows-wheeler transform conjecture.
\newblock In {\em 61st {IEEE} Annual Symposium on Foundations of Computer
  Science, {FOCS} 2020}, 2020.
\newblock \href {https://doi.org/10.1109/FOCS46700.2020.00097}
  {\path{doi:10.1109/FOCS46700.2020.00097}}.

\bibitem[KL21]{KL21}
Dominik Kempa and Ben Langmead.
\newblock Fast and space-efficient construction of {AVL} grammars from the
  {LZ77} parsing.
\newblock In {\em 29th Annual European Symposium on Algorithms, {ESA} 2021},
  2021.
\newblock \href {https://doi.org/10.4230/LIPICS.ESA.2021.56}
  {\path{doi:10.4230/LIPICS.ESA.2021.56}}.

\bibitem[KNW24]{KNW24}
Tomasz Kociumaka, Jakob Nogler, and Philip Wellnitz.
\newblock On the communication complexity of approximate pattern matching.
\newblock In {\em Proceedings of the 56th Annual ACM Symposium on Theory of
  Computing}, STOC 2024, 2024.
\newblock \href {https://doi.org/10.1145/3618260.3649604}
  {\path{doi:10.1145/3618260.3649604}}.

\bibitem[Kos87]{Kos87}
S.R. Kosaraju.
\newblock Efficient string matching.
\newblock Manuscript, 1987.

\bibitem[KP18]{KP18}
Dominik Kempa and Nicola Prezza.
\newblock At the roots of dictionary compression: string attractors.
\newblock In {\em Proceedings of the 50th Annual ACM SIGACT Symposium on Theory
  of Computing}, STOC 2018, 2018.
\newblock \href {https://doi.org/10.1145/3188745.3188814}
  {\path{doi:10.1145/3188745.3188814}}.

\bibitem[KPS21]{KPS21}
Tomasz Kociumaka, Ely Porat, and Tatiana Starikovskaya.
\newblock Small-space and streaming pattern matching with $k$ edits.
\newblock In {\em 62nd {IEEE} Annual Symposium on Foundations of Computer
  Science, {FOCS} 2021}, 2021.
\newblock \href {https://doi.org/10.1109/FOCS52979.2021.00090}
  {\path{doi:10.1109/FOCS52979.2021.00090}}.

\bibitem[KR87]{KR87}
Richard~M. Karp and Michael~O. Rabin.
\newblock Efficient randomized pattern-matching algorithms.
\newblock {\em {IBM} J. Res. Dev.}, 1987.
\newblock \href {https://doi.org/10.1147/RD.312.0249}
  {\path{doi:10.1147/RD.312.0249}}.

\bibitem[Lev65]{Levenshtein66}
Vladimir~Iosifovich Levenshtein.
\newblock Binary codes capable of correcting deletions, insertions and
  reversals.
\newblock {\em Doklady Akademii Nauk SSSR}, 1965.

\bibitem[LV88]{LV88}
Gad~M. Landau and Uzi Vishkin.
\newblock Fast string matching with $k$ differences.
\newblock {\em Journal of Computer and System Sciences}, 1988.
\newblock \href {https://doi.org/10.1016/0022-0000(88)90045-1}
  {\path{doi:10.1016/0022-0000(88)90045-1}}.

\bibitem[LV89]{LandauV89}
Gad~M. Landau and Uzi Vishkin.
\newblock Fast parallel and serial approximate string matching.
\newblock {\em Journal of Algorithms}, 1989.
\newblock \href {https://doi.org/10.1016/0196-6774(89)90010-2}
  {\path{doi:10.1016/0196-6774(89)90010-2}}.

\bibitem[Nav01]{Nav01}
Gonzalo Navarro.
\newblock A guided tour to approximate string matching.
\newblock {\em {ACM} Comput. Surv.}, 2001.
\newblock \href {https://doi.org/10.1145/375360.375365}
  {\path{doi:10.1145/375360.375365}}.

\bibitem[Nav21]{N21}
Gonzalo Navarro.
\newblock Indexing highly repetitive string collections, part i: Repetitiveness
  measures.
\newblock {\em ACM Comput. Surv.}, March 2021.
\newblock \href {https://doi.org/10.1145/3434399} {\path{doi:10.1145/3434399}}.

\bibitem[PP09]{PP09}
Benny Porat and Ely Porat.
\newblock Exact and approximate pattern matching in the streaming model.
\newblock In {\em 50th Annual {IEEE} Symposium on Foundations of Computer
  Science, {FOCS} 2009}, 2009.
\newblock \href {https://doi.org/10.1109/FOCS.2009.11}
  {\path{doi:10.1109/FOCS.2009.11}}.

\bibitem[Pre19]{P19}
Nicola Prezza.
\newblock Optimal rank and select queries on dictionary-compressed text.
\newblock In {\em 30th Annual Symposium on Combinatorial Pattern Matching,
  {CPM} 2019}, 2019.
\newblock \href {https://doi.org/10.4230/LIPICS.CPM.2019.4}
  {\path{doi:10.4230/LIPICS.CPM.2019.4}}.

\bibitem[Rub19]{Rubinstein19}
Aviad Rubinstein.
\newblock {Q}uantum {DNA} sequencing and the ultimate hardness hypothesis.
\newblock
  \url{https://theorydish.blog/2019/12/09/quantum-dna-sequencing-the-ultimate-hardness-hypothesis/},
  2019.
\newblock Accessed October 9, 2024.

\bibitem[Ryt03]{Ryt03}
Wojciech Rytter.
\newblock Application of lempel-ziv factorization to the approximation of
  grammar-based compression.
\newblock {\em Theoretical Computer Science}, 2003.
\newblock \href {https://doi.org/10.1016/S0304-3975(02)00777-6}
  {\path{doi:10.1016/S0304-3975(02)00777-6}}.

\bibitem[Sel80]{S80}
Peter~H. Sellers.
\newblock The theory and computation of evolutionary distances: Pattern
  recognition.
\newblock {\em Journal of Algorithms}, 1980.
\newblock \href {https://doi.org/10.1016/0196-6774(80)90016-4}
  {\path{doi:10.1016/0196-6774(80)90016-4}}.

\bibitem[SV96]{SV96}
S{\"{u}}leyman~Cenk Sahinalp and Uzi Vishkin.
\newblock Efficient approximate and dynamic matching of patterns using a
  labeling paradigm (extended abstract).
\newblock In {\em 37th Annual {IEEE} Symposium on Foundations of Computer
  Science, {FOCS} 1996}, 1996.
\newblock \href {https://doi.org/10.1109/SFCS.1996.548491}
  {\path{doi:10.1109/SFCS.1996.548491}}.

\bibitem[WY24]{WY20}
Qisheng Wang and Mingsheng Ying.
\newblock Quantum algorithm for lexicographically minimal string rotation.
\newblock {\em Theory of Computing Systems}, 2024.
\newblock \href {https://doi.org/10.1007/S00224-023-10146-8}
  {\path{doi:10.1007/S00224-023-10146-8}}.

\bibitem[ZL77]{DBLP:journals/tit/ZivL77}
Jacob Ziv and Abraham Lempel.
\newblock A universal algorithm for sequential data compression.
\newblock {\em IEEE Transactions on Information Theory}, 1977.
\newblock \href {https://doi.org/10.1109/TIT.1977.1055714}
  {\path{doi:10.1109/TIT.1977.1055714}}.

\end{thebibliography}
